\documentclass[english,12pt,reqno]{smfart}

\usepackage[OT2,OT1]{fontenc}
\usepackage{linearb}
\usepackage[T1]{fontenc}
\usepackage[latin9]{inputenc}

\usepackage[table]{xcolor}
\usepackage{hyperref}
\usepackage{amsfonts}
\usepackage{amsmath}
\usepackage{amssymb}
\usepackage{amscd}
\usepackage{latexsym}
\usepackage{graphicx}
\usepackage{esint}
\usepackage{rotating}
\usepackage[all]{xy}
\usepackage{tikz}
\usetikzlibrary{calc,decorations.pathreplacing,decorations.markings}
\hypersetup{colorlinks=true,linkcolor=magenta,anchorcolor=green,citecolor=cyan,filecolor=black,menucolor=black,urlcolor=black}

\DeclareMathAlphabet{\mathbbold}{U}{bbold}{m}{n}

\newlength{\abstractwidth}
\newcount\hh
\newcount\mm
\mm=\time
\hh=\time
\divide\hh by 60
\divide\mm by 60
\multiply\mm by 60
\mm=-\mm
\advance\mm by \time
\def\thistime{\number\hh:\ifnum\mm<10{}0\fi\number\mm}

\makeatletter
\newcommand\footnoteref[1]{\protected@xdef\@thefnmark{\ref{#1}}\@footnotemark}
\makeatother

\numberwithin{equation}{section}

\def\nn{\nonumber}
\def\su{\circleddash}
\def\cEs{\mathcal E_\circleddash}

\def\Li_#1(#2){\textrm{Li}_{#1}\left(#2\right)}
\def\cLi_#1(#2){\mathcal{L}_{#1}\left(#2\right)}
\def\bLi_#1(#2){\mathbf{L}_{#1}\left(#2\right)}

\def\cI{\mathcal I}
\def\cY{\mathcal Y}
\def\cIs{\mathcal{I}_\circleddash}

\def\ZZ{{\mathbb Z}}
\def\IC{{\mathbb C}}

\def\RR{{\mathbb R}}
\def\IN{{\mathbb N}}
\def\IP{{\mathbb P}}

\def\cM{\mathcal{M}}

\def\cI{\mathcal{I}}
\def\cE{\mathcal{E}}

\def\cV{\mathcal{V}}
\def\cR{\mathcal{R}}

\def\Imm{\Im\textrm{m}}

\def\ay{\mathbf{i}}

\def\um{\underline{m}}

\title[]{\bf Local mirror symmetry and the sunset Feynman integral}

\author{Spencer Bloch}
\address{5765 S. Blackstone Ave., Chicago, IL 60637, USA}
\email{spencer\_bloch@yahoo.com}

\author{Matt Kerr}
\address{Department of Mathematics, Campus Box 1146\\
Washington University in St. Louis\\
St. Louis, MO, 63130, USA}
\email{matkerr@math.wustl.edu}

\author{Pierre Vanhove}
 \address{Department of Applied Mathematics and Theoretical Physics\\
 Wilberforce Road, Cambridge CB3 0WA, UK\hfill\break
Institut de physique th\'eorique, Universit\'e Paris Saclay,
 CEA, CNRS, F-91191 Gif-sur-Yvette}
\email{pierre.vanhove@cea.fr}

\thanks{IPHT-t15/135, DAMTP-2016-1}
\date{\today}

\begin{document}

 \begin{abstract}

We study the sunset Feynman integral defined as the
scalar two-point self-energy at two-loop order in a two dimensional space-time.

\medskip

We firstly compute the Feynman
integral, for arbitrary internal masses, 
in terms of the regulator of a class in the motivic cohomology of 
a 1-parameter family of open elliptic
curves. Using an Hodge theoretic (B-model) approach, we show that the integral is given by a sum of elliptic
dilogarithms evaluated at the divisors determined by the punctures. 

\medskip

Secondly we associate to the sunset elliptic curve a local non-compact Calabi-Yau 3-fold,
obtained as a limit of elliptically fibered compact Calabi-Yau 3-folds.  By considering the limiting mixed Hodge structure of the Batyrev dual A-model, we arrive at an expression for the sunset Feynman integral in terms of the local Gromov-Witten
prepotential of the del Pezzo surface of degree 6.  This expression is obtained by proving a strong form
of local mirror symmetry which identifies this prepotential with the second regulator period of the motivic cohomology class.


\end{abstract}

\maketitle
\newpage
{\hypersetup{linkcolor=black}
\tableofcontents}
\newpage

\addtocontents{toc}{\protect\enlargethispage{2\baselineskip}}
\part{Introduction}
\label{sec:introduction}

\section{Overview and discussion}

\tikzset{->-/.style={decoration={ markings,mark=at position #1 with {\arrow{>}}},postaction={decorate}}}

\begin{center}
\tikzpicture[scale=1.7]
\scope[xshift=-5cm,yshift=-0.4cm]
\draw [draw=black,->-=.6]        (0,0) circle (1 and .7);
\draw [draw=black,->-=.4]        (0,0) circle (1 and .7);
\draw[->-=.5] (-1.5,0) to (-1,0);
\draw[->-=.5] (-1,0) to (1,0);
\draw[->-=.5] (1,0) to (1.5,0);
\draw(0,.5) node{$m_1$};
\draw(0,-.2) node{$m_2$};
\draw(0,-.9) node{$m_3$};
%
%
\draw(1.6,0) node{$K$};
\draw(-1.6,0) node{$K$};
\endscope
\endtikzpicture 
\end{center}

This work concerns the Feynman integral in two dimensional space-time associated to
the sunset graph  in the above figure, given by 

\begin{equation}\label{e:sunsetint}
  \cIs(s):= s\int_{x\geq0\atop y\geq0} \,{dx dy\over s(\xi_1^2x+\xi_2^2y+\xi_3^2)(xy+x+y)-xy} \, .
\end{equation}
Here $\xi_i=m_i/\mu$ ($i=1,2,3$) are positive non vanishing real
numbers, given by the ratios of the internal masses by the arbitrary
infrared scale $\mu$, and $s$ is the inverse of the norm of the external momentum $K^2=\mu^2/s$.
(See~\cite{Bogner:2010kv,Bloch:2013tra,PVstringmath} for
a derivation of \eqref{e:sunsetint} from the usual Feynman representation.)

This integral is a multivalued function of $s$ on $\IC\backslash
[(\xi_1+\xi_2+\xi_3)^{-2},+\infty[$.  In general, the multivalueness of the Feynman integral plays
an important role in physics, as this is imposed by  unitarity of
quantum field theory~\cite{Elvang:2013cua}. A large class of Feynman
integrals can be easily determined from their differential
equations~\cite{Argeri:2007up,Remiddi:2013joa,Henn:2013pwa,Lee:2014ioa,Henn:2014qga,Tancredi:2015pta}, and more generally are associated to
motivic period  integrals~\cite{Brown:2015fyf,Brown:2015qmm}.   

The geometry of the graph hypersurface is a family of elliptic
curves 
$$
\mathcal E_\su:=\{ xyz -
s(\xi_1^2x+\xi_2^2y+\xi_3^2z)(xy+xz+yz)| (x,y,z)\in\IP^2\}\,.
$$
The structure of the motive associated to \eqref{e:sunsetint}, discussed in sections~\ref{sec:motives}
and~\ref{sec:sunsetBmodel}, differs from  the one given in the single
masses case in~\cite{Bloch:2013tra}, because we now have a family of
open elliptic curves, no longer modular, and the motive has a  Kummer  extension quotient.

We show that the sunset Feynman integral is given by (see section~\ref{sec:three-mass})\footnote{It would be interesting to relate this expression to the one using  
multiple polylogarithm presented in~\cite{Adams:2014vja,Adams:2015gva,Adams:2015ydq,Adams:2016vdo}.}
\begin{equation}\label{e:I1}
 \cIs(s)\equiv {\ay \varpi_r   \over\pi} \left(\hat E_2\left(x(P_1)\over
      x(P_2)\right)+\hat E_2\left(x(P_2)\over x(P_3)\right)+\hat
    E_2\left(x(P_3)\over x(P_1)\right)\right)\,
\mod  \textrm{periods}\, ,
\end{equation}
where $\hat E_2(x)$ is the elliptic dilogarithm
\begin{equation}
  \hat E_2(x)= \sum_{n\geq0} \left(\Li_2(q^nx)-\Li_2(-q^nx)\right)
    -\sum_{n\geq1} \left(\Li_2(q^n/x)-\Li_2(-q^n/x)\right)  \, .
\end{equation}
In \eqref{e:I1}, $\hat{E}_2$ is evaluated at the ratios of the images of the points $P_1:=[1:0:0]$,
$P_2:=[0:1:0]$ and $P_3:=[0:0:1]$  in $\mathbb C^\times/q^\ZZ$, where
$\log(q)/(2\pi\ay)$ is the complex
structure given by the period
ratio of the elliptic curve; 
and  $\varpi_r$ is the elliptic curve  period which is  real
on the line $s>(\xi_1+\xi_2+\xi_3)^{-2}$.   

The elliptic dilogarithm $\hat E_2(x)$ is not invariant under $x\to
xq$ (see equation~\eqref{e:Eh2shift}), reflecting the multivalued nature of the
Feynman integral. This was already the case for elliptic polylogarithm expansions of 
the Feynman integrals for the two-loop sunset \cite{Bloch:2013tra} and three-loop banana \cite{BKV}
with equal masses. The result in~\eqref{e:I1} generalizes the expression for the all equal
masses case $\xi_1=\xi_2=\xi_3=1$ given in terms of elliptic dilogarithm in~\cite{Bloch:2013tra}.

The motivic approach in section~\ref{sec:motives} shows how the theory
of motives can yield information about Feynman integrals. In general,
the motive associated to a Feynman integral will depend on a family of
hypersurfaces $X_{m,q} \subset \mathbb P^n$ depending on masses $m$ and
external momenta $q$. The motive at $(m,q)$ is associated to the
cohomology group $H^n(\mathbb P^n-X_{m,q},\Delta)$ where $\Delta$ is the simplex defined by the vanishing of the product of the homogeneous coordinates. A general motivic analysis would begin by a study of $X_{m,q}\cap \Delta$. In simple cases like the sunset, this intersection is manageable and we are able to prove a duality
$$H^n(\mathbb P^n-X_{m,q},\Delta) \cong H^n(\mathbb P^n-\Delta,X_{m,q})(n)^\vee.
$$
The motive on the right is related to the Milnor symbol
$\{x_1,\dotsc,x_n\}$ on $X_{m,q}\cap \mathbb G_m^n$, where the $x_i$ are the
Laurent coordinates on $\mathbb P^n-\Delta = \mathbb G_m^n$.  In the sunset case, this approach identifies the amplitude with an elliptic dilogarithm. A similar attack may be possible for more general graphs, though the above duality will no longer be perfect. The challenge will be to understand the role played by the structure at infinity $X_{m,q}\cap \Delta$.

In part \ref{part:Two}, we revisit the approach of \cite{CKYZ} to local mirror symmetry,
by semi-stably degenerating a family of elliptically-fibered Calabi-Yau 3-folds $\mathrm{X}_{z_0,\underline{z}}$ (defined by \eqref{MKeqnI0})
to a singular compactification $\mathrm{X}_{0,\underline{z}}$ of the local Hori-Vafa 3-fold
\[ Y_{\underline{z}} := \{ 1-s(\xi_1^2 x + \xi_2^2 y +\xi_3^2)(1+x^{-1}+y^{-1}) + uv = 0\}\subset (\mathbb{C}^*)^2 \times \mathbb{C}^2 \, , \]
and using the work of Iritani \cite{Iritani} to compare the asymptotic Hodge theory of this B-model to that of the mirror (elliptically fibered) A-model Calabi-Yau $\mathrm{X}^{\circ}$.

The bulk of section~\ref{sec:Bmodel} is concerned with the proof of the isomorphism
\[ H^3_{lim}(\mathrm{X}_{z_0,\underline{z}} )\supseteq \text{ker}(T_0-I)\cong H_3(Y_{\underline{z}})(-3) \]
of mixed Hodge structures (Theorem \ref{MKthm1}), and the explicit
construction of bases for $H^3(\mathrm{X}_{z_0,\underline{z}})$
resp. $H_3(Y_{\underline{z}})$. This allows us to invoke (in
section~\ref{MKsecIG}) results of\footnote{ \label{fnt:DK}
The numbers of section,
conjecture, theorem and equations refer to the published version
of~\cite{DoranKerr}.} \cite[\S 5]{DoranKerr} to compute, in the $z_0 \to 0$ limit, the invariant periods of $\mathrm{X}$ in terms of ``regulator periods'' $R_0^{(i)},R_1$ associated to a family of algebraic $K_2$-classes on the sunset elliptic curve family $E_{\underline{z}}$.

In \S\ref{sec:Amodel}, we compute Iritani's quantum $\mathbb{Z}$-variation of Hodge structure on the even cohomology of the Batyrev mirror $\mathrm{X}^{\circ}$ of $\mathrm{X}$, writing the periods in terms of its Gromov-Witten invariants (section~\ref{MKsecIIB}) and the monodromy transformations in terms of its intersection theory (section~\ref{MKsecIIC}).  (The monodromies $T_i$ are computed in greater detail than we need, as they will be used to provide geometric realizations of certain monodromy cones in the forthcoming work \cite{KPR}.)  Like $\mathrm{X}$, $\mathrm{X}^{\circ}$ is elliptically fibered, over a toric Fano surface $\mathbb{P}_{\Delta^{\circ}}$, which (for the sunset case) is just the del Pezzo of degree 6.  Under the mirror map $\underline{z} \mapsto \underline{q}(\underline{z})=e^{2\pi\ay \underline{\tau}(\underline{z})}$ (computed in section~\ref{MKsecIID}), we have the isomorphism of A- and B-model $\mathbb{Z}$-variation of Hodge structure 
\[ H^3(\mathrm{X}_{z_0,\underline{z}}) \cong H^{even}(\mathrm{X}^{\circ}_{q_0,\underline{q}}) \,,\]
and taking (the invariant part of) limiting mixed Hodge structure on both sides yields the relation
\[ 2\pi \ay R_1 = R_0^{(1)} R_0^{(2)} + R_0^{(2)} R_0^{(3)} +
R_0^{(1)} R_0^{(3)} - \sum_{\ell_1+\ell_2+\ell_3=\ell>0\atop
  (\ell_1,\ell_2,\ell_3)\in\IN^3\backslash (0,0,0)} \ell\,
N_{\ell_1,\ell_2,\ell_3} \prod_{i=1}^3 Q_i^{\ell_i} \]
between regulator periods and local Gromov-Witten numbers of
$\mathbb{P}_{\Delta^{\circ}}$ (Corollary \ref{MKcor2}).  
The expression is done with respect to the local K\"ahler moduli 
$Q_i=e^{R_0^{(i)}}=\xi_i^2\,\hat Q$ for
$i=1,2,3$ with $\hat Q=\exp(\hat R_0)$ and 
where $\hat R_0$ is the logarithmic Mahler measure
\begin{equation}
 \hat R_0= \ay\pi- \int_{|x|=|y|=1}\hspace{-.6cm}
  \log(s^{-1}-(\xi_1^2x+\xi_2^2y+\xi_3^2)(x^{-1}+y^{-1}+1)) \,{d\log x d\log y\over(2\pi\ay)^2}\,.
\end{equation}

 That is, we prove that $R_1$ is the local Gromov-Witten prepotential of $\mathbb{P}_{\Delta^{\circ}}$, which is Conjecture 5.1\footnoteref{fnt:DK} of \cite{DoranKerr}; this puts the observations on asymptotics of the local Gromov-Witten invariants there (Corollary 5.3 of~\cite{DoranKerr}) on a firm foundation at last.

All of what has just been described is carried out, in sections~\ref{sec:Amodel}-\ref{sec:Bmodel}, in a greater level of generality so that the results described apply to other toric families of elliptic curves in addition to the sunset family.

The connection of all this to the Feynman integral \eqref{e:sunsetint} is given in section~\ref{sec:multi}:  writing $\omega_{\underline{z}}$ for a family of holomorphic 1-forms on $E_{\underline{z}}$, and $R|_{E_{\underline{z}}}$ for the family of 1-currents associated to the family of algebraic $K_2$-classes, we have the equality
\[ \cIs(s) = -s\,\int_{E_{\underline{z}}} R|_{E_{\underline{z}}} \wedge \omega_{\underline{z}} . \]
Proposition \ref{MKpropIII1}  shows this leads to the inhomogeneous Picard-Fuchs equation for
$\cIs$  derived explicitly in section~\ref{sec:elliptic-dilogarithm}.

Remarkably we show that the sunset Feynman integral is given by the Legendre
transform of the regulator period $\hat R_1=R_1$  (see~\eqref{eq:Legendre} and~\eqref{e:LegendreE})
\begin{equation}\label{e:Leg}
 \cIs(s)\simeq -s\,2\pi\ay\pi_0\left(
                                                             \frac{\partial\hat
                                                               R_1}{\partial\hat
                                                               R_0}\hat{R}_0
                                                              -
                                                              \hat{R}_1
                                                              \right)       \,,
\end{equation} 
which implies the expansion
of the Feynman integral in terms of Gromov-Witten numbers (see
sections~\ref{sec:localGW} and~\ref{sec:grom-witt-invar-1})
\begin{equation}\label{e:curiousInt}
  \cIs(s)=-s^2{\partial \hat R_0\over\partial s}\,\left( 3 \hat R_0^3  +\sum_{\ell_1+\ell_2+\ell_3=\ell>0\atop (\ell_1,\ell_2,\ell_3)\in\IN^3\backslash(0,0,0)}\ell (1- \ell
  \log\hat Q) 
  N_{\ell_1,\ell_2,\ell_3}\,\prod_{i=1}^3 \hat Q_i^{\ell_i}\right)\,.
\end{equation}
The local Gromov-Witten numbers $N_{\ell_1,\ell_2,\ell_3}$ can be
expressed in terms of the virtual integer number of degree $\ell$
rational curves  by
$$
  N_{\ell_{1},\ell_{2},\ell_{3}}=\sum_{d|\ell_{1},\ell_{2},\ell_{3}}\frac{1}{d^{3}}n_{\frac{\ell_{1}}{d},\frac{\ell_{2}}{d},\frac{\ell_{3}}{d}}\,.
$$
These numbers are tabulated in sections~\ref{sec:localGW} and~\ref{sec:grom-witt-invar-1}.
In the particular case
of the all equal masses case $\xi_1=\xi_2=\xi_3=1$, the mirror map gives
(see section~\ref{sec:grom-witt-invar-1})
\begin{equation}\label{e:Qtoq}
\hat  Q=-q\, \prod_{ n\geq1}  (1-q^n)^{n \delta(n)}; \qquad \delta(n):=(-1)^{n-1} \,
\left(-3\over n \right)\,,
\end{equation}
where $\left(-3\over n\right)=0,1,-1$ for $n\equiv 0,1,2\mod 3$.
The   modularity of the family of sunset elliptic curves allows us to 
relates the sum of elliptic dilogarithms in $q$ of~\cite{Bloch:2013tra} to
the Gromov-Witten expansion in $\hat Q$, and implies the Legendre transform
relation~\eqref{e:Leg}.
Stienstra has already noticed in~\cite{Stienstra,Stienstra:2005wz} the
similarity between the mirror symmetry transformation
in~\eqref{e:Qtoq}  and the ones between A-models of local Calabi-Yau
and dimer models~\cite{Okounkov:2003sp}  for the topological vertex
description of the B-model~\cite{Aganagic:2003db,Aganagic:2003qj}.
Theorem~3.5 of~\cite{Kenyon:2003uj}  shows that the partition
function of the dimer model \emph{is} the Mahler measure of the
Laurent polynomial defining the local Calabi-Yau model.
 In~\cite{Stienstra:2005wz} Stienstra constructed a dimer model associated
to the all equal masses sunset elliptic curve $\xi_1=\xi_2=\xi_3=1$. 
In the case of unequal masses there is no modularity, and it is 
surprising that  an analytic continuation
of a sum of elliptic dilogarithms displays such relation to the local
Gromov-Witten prepotential.

Special type of Feynman integrals for topological strings have been used to compute the local  
Gromov-Witten prepotential~\cite{Hori:2003ic}.
But  our analysis leads to a different kind of result, firstly
because  the sunset Feynman
integral is the Legendre transform~\eqref{e:Leg} of the local
Gromov-Witten prepotential, secondly because this Feynman graph
is not obviously associated to world-sheet graphs of  a topological
string. 
Our results extend to the three-loop banana graph and the
four-loop banana graph, leading to 4-fold and 5-folds
Calabi-Yau  respectively (cf. section~5 of~\cite{DoranKerr}).
The strong similarity of our  analysis with the dimer
models  suggests that one could expect more connection between
Gromov-Witten prepotential and 
(massive) quantum field theory Feynman integrals. 
We expect that this approach to Feynman
integrals can shed some new light on the relation to string theory
along the lines of the results of \cite{ABBF}.

\section{Plan of the paper}

The plan of the paper is the following. In part~\ref{part:One}, we 
analyse the sunset Feynman integral~\eqref{e:sunsetint}. In
 section~\ref{sec:open-elliptic-curve} we describe 
the geometry of the sunset family of elliptic curve and in
section~\ref{sec:PFderivation} derive
the Picard-Fuchs equation following Griffiths's
approach in~\cite{Griffiths} for deriving the Picard-Fuchs equation
from the cohomology of smooth projective hyperspace defined by
rational form in $\IP^2$.  In section~\ref{sec:elliptic-dilogarithm}
we derive the expression~\eqref{e:I1} of the sunset integral in terms
of elliptic dilogarithm. In section~\ref{sec:single-mass-case} we show
how to reproduce the all equal masses result of~\cite{Bloch:2013tra} and
 section~\ref{sec:three-mass} contains numerical verification of the
 three different masses case.  We give a proof of these results using a
 motivic approach in section~\ref{sec:motives}.

Part~\ref{part:Two} of the paper deals with the mirror symmetry
construction. In section~\ref{sec:Bmodel}
we describe the degeneration from a compact Calabi-Yau 3-fold $\mathrm{X}$ to the local Hori-Vafa
model $\mathrm{Y}$, and show in theorem \ref{MKthm1} that  the third
homology of $\mathrm{Y}$ matches the invariant part of the limiting mixed
Hodge structure of $H^{3}(\mathrm{X})$.
In section~\ref{sec:Amodel} we describe the variation of Hodge structure arising
on the A-model obtained by considering the Batyrev mirror of
$\mathrm{X}$.   By comparing the limiting mixed Hodge structures  of the A-model and B-model, we
prove in theorem \ref{MKthm2}  a strong form of local mirror symmetry -- equality of variations of
$\mathbb Q$-mixed Hodge structure. The particular case of
the sunset integral is discussed in section~\ref{sec:multi}.

In the appendix~\ref{sec:theta-functions} we recall the main
properties of Jacobi theta functions, and in the
appendix~\ref{sec:coeff} we give the detailed coefficients entering
the derivation of the Picard-Fuchs equation in section~\ref{sec:PFderivation}.

\medskip
\centerline{\bf Acknowledgements}
\medskip

The research of PV has received funding the ANR
grant reference QST 12 BS05 003 01, and the CNRS grants PICS number
6430. 
PV is partially supported by   a fellowship funded by the French
Government at Churchill College, Cambridge, the European Research
Council under the European Community's Seventh Framework Programme
(FP7/2007-2013) ERC grant agreement no. [247252], and the STFC grant ST/L000385/1.
MK was partially supported by NSF Grant DMS-1361147, and by the Fund for Mathematics.


%

\theoremstyle{plain}
\newtheorem{thm}{\protect\theoremname}[section]
  \theoremstyle{remark}
  \newtheorem{rem}[thm]{\protect\remarkname}
 \newtheorem{remark}[thm]{\protect\remarkname}
  \theoremstyle{plain}
  \newtheorem{cor}[thm]{\protect\corollaryname}
  \theoremstyle{plain}
  \newtheorem{prop}[thm]{\protect\propositionname}
\newtheorem{lem}[thm]{\protect\lemmaname}

 \providecommand{\corollaryname}{Corollary}
  \providecommand{\propositionname}{Proposition}
  \providecommand{\remarkname}{Remark}
\providecommand{\theoremname}{Theorem}
\providecommand{\lemmaname}{Lemma}

\def\Li_#1(#2){\textrm{Li}_{#1}\left(#2\right)}

\def\phis{\phi_\su}
\def\Phis{\Phi_\su}
\def\cYs{\cY_\su}

\part{The elliptic dilogarithm}\label{part:One}
\section{The sunset Feynman integral}
\label{sec:3masses-integral}

The sunset Feynman integral is 
\begin{equation}
  \label{e:IsunsetSym}
  \cIs(s)=  -s\,\int_{\Delta} \Omega_\su(s)
 \,.
\end{equation}
where  the domain of integration is
\begin{equation}\label{e:Ddef}
  \Delta=\{(x,y,z)\in \mathbb P^2| x,y,z\geq0\}  \,,
\end{equation}
and
\begin{equation}
  \label{e:OmegaDef}
  \Omega_\su(s):=  {x dy\wedge dz+y dz\wedge  dx+z dx\wedge dy\over
    xyz \left(1-s\phis\right)  }\,.
\end{equation}
where we have set
\begin{equation}
  \label{e:defphi}
\phis=(\xi_1^2x+\xi_2^2y+\xi_3^2z)(x^{-1}+y^{-1}+z^{-1})\,.
\end{equation}
Where $\xi_i=m_i/\mu$ for $i=1,2,3$ are non-vanishing positive real
numbers given by the ratio of the  internal masses parameters $m_i$ 
and an infrared scale $\mu$.   In this work we assume that  none of the masses vanish. 
As function of $1/s:=K^2/\mu^2$ the integral is a multivalued function on
the complex plane minus a line
$\IC\backslash[(\xi_1+\xi_2+\xi_3)^2,+\infty[$.

In this first part of the paper we show that this integral is an
elliptic dilogarithm. We give to derivations on by a direct
computation and second one based a  motivic analysis.

\subsection{The sunset open elliptic curve}
\label{sec:open-elliptic-curve}

For generic values of the parameters the polar part of
$\Omega_\su(s)$ defines an open with marked points elliptic curve 
\begin{equation}\label{e:Esunset}
  \cEs:=\left\{xyz-s(\xi_1^2x+\xi_2^2y+\xi_3^2z)(xy+xz+yz)=0| (x,y,z)\in \mathbb P^2  \right\}\,.
\end{equation}
The discriminant  is 
 \begin{equation}
   \label{e:Disc}
   \Delta= 16\,s^{-6} M_6^2 \prod_{i=1}^4 (1-s\mu_i^2) \,,
 \end{equation}
and the $J$-invariant is 
 \begin{equation}
\label{e:Jinvariant} J=
-{(\prod_{i=1}^4 (1-s\mu_i^2)+16 s^3 \prod_{i=1}^3 \xi_i^2)^3\over s^6\, 
  M_6^2 \prod_{i=1}^4 (1-s\mu_i^2) }\,,
 \end{equation}
with 
\begin{align}
  \mu_1&:=-\xi_1+\xi_2+\xi_3,\qquad \mu_2:=\xi_1-\xi_2+\xi_3,\nn\\
   \mu_3&:=\xi_1+\xi_2-\xi_3,\qquad \mu_4:=\xi_1+\xi_2+\xi_3\,,
	\label{e:mui}
\end{align}
and
\begin{equation}\label{e:MassDef}
 M_2:=\xi_1^2+\xi_2^2+\xi_3^2,\,
       M_4:=\xi_1^2\xi_2^2+\xi_1^2\xi_3^2+\xi_2^2\xi_3^2,\, M_6:=\xi_1^2\xi_2^2\xi_3^2\,.   
\end{equation}
For generic values of the masses $\xi_1\neq \xi_2\neq \xi_3$ 
there are six singular fibers:   at $s=0$ of type $I_6$, at $s=\infty$
of type $I_2$
and for $1\leq i\leq 4$ at $s^{-1}=\mu_i$  of type~$I_1$. 

We recall that  for the all equal masses case $\xi_1=\xi_2=\xi_3=1$ there are
only four singular fibers of type $I_2$ for $s=\infty$,  $I_3$ for $s=1$, $I_9$
for $s=1/9$ and $I_6$ for $s=0$~\cite{Bloch:2013tra}.

If we introduce the Hauptmodul $u$
\begin{equation}
  \label{e:udef}
  u:={(1-sM_2)^2-4 s^2 M_4\over \sqrt{16s^3M_6}}\,,
\end{equation}
the $J$-invariant takes the form
\begin{equation}
  \label{eq:9}
  J:= 256\, {(3-u^2)^3  \over 4-u^2}\,.
\end{equation}

We introduce
$q=\exp(2\pi\ay\tau)$ with $\tau=\varpi_c/\varpi_r$ the
ratio of the complex period 
 $\varpi_c$ and period $\varpi_r$ is the  real period on the real
axis $s>(\xi_1+\xi_2+\xi_3)^{-2}$.   We assume that $\varpi_c$ has a
positive imaginary part so that $|q|<1$ and $\tau$ is in the upper half-plane.

From the usual parametrization of the $J$-invariant in terms of
theta-functions (see Appendix~\ref{sec:theta-functions}) we deduce
that the Hauptmodul $u$ is given by the three roots
\begin{equation}
  \label{e:uhaupt}
 u_{a,b}\in \left\{u_{3,4}={\theta_3^4+\theta_4^4\over \theta_3^2\theta_4^2},
 u_{2,3}=- {\theta_3^4+\theta_2^4\over \theta_3^2\theta_2^2}, u_{2,4}=i {\theta_2^4-\theta_4^4\over \theta_2^2\theta_4^2} \right\}\,.
\end{equation}
The action of $SL(2,\ZZ)$ leaves invariant the $J$-invariant but
 rotates the three roots.  The subgroup $\Gamma$  of $SL(2,\ZZ)$
 generated by $\tau\to\tau+2$ and $\tau\to \tau/(1-2\tau)$ (see~\cite{chand})
 \begin{equation}
 \Gamma=\left\{\begin{pmatrix}a & b \cr c & d\end{pmatrix} \in
   SL(2,\ZZ) | \begin{pmatrix}a     & b \cr c & d\end{pmatrix}
   = \begin{pmatrix} 1 & 0 \cr 0 & 1\end{pmatrix}  \mod 2\right\}\,,
   \end{equation}
leaves invariant the square of each 
 individual roots $u_{a,b}^2$  for given $a,b$. 

\medskip

For each pair $(a,b)$ labelling the Hauptmodul in~\eqref{e:uhaupt} the  real period $\varpi_r$ is then given  in terms of the
theta constants (see Appendix~\ref{sec:theta-functions} for
definitions and conventions)
\begin{equation}\label{e:wrq}
  \varpi_r= \pi {\theta_a \, \theta_b \over (s^{-1} M_6) ^{1\over4} }\,.
\end{equation}
%

\subsubsection{The points}
\label{sec:points}

The intersection of the elliptic curve and the domain
of integration $\Delta$ are the three points
\begin{equation}\label{e:Ppoints}
    \partial \Delta\cap \mathcal E_\su= \{ P_1=[1,0,0],\ P_2=[0,1,0],\  P_3=[0,0,1]\}\,.
\end{equation}

We will consider as well the other three points 
\begin{equation}
  \label{e:Qpoints}
  Q_1=[0,-\xi_3^2,\xi_2^2],\qquad  Q_2=[-\xi_3^2,0,\xi_1^2],\qquad  Q_3=[-\xi_2^2,\xi_1^2,0]\,,
\end{equation}
arising from the intersection of the sunset elliptic curve and the
lines defining the domain of integration $\Delta$.

In order to map these points to $\mathcal E_\su\simeq\IC^\times/q^\ZZ$
where $C^\times$ is the multiplicative group of non-zero complex numbers, we use the following  Weierstrass model for the sunset elliptic
curve
\begin{equation}\label{e:E3}
\zeta^2\eta=\sigma \, ( s^{-1}M_6\eta^2+u\,\sqrt{s^{-1} \, M_6}\,\sigma\eta+\sigma^2 )\,.
\end{equation}

For any choice of $(a,b,c)=(3,4,2), (2,3,4), (2,4,3)$, a  point on the elliptic curve  with coordinates
$P=[\sigma,\zeta,\eta]$ and $\eta\neq0$ is parametrized by\footnote{We would like to thank Don
  Zagier for explaining how to perform this reduction, and for
  providing the key identities.}  

\begin{align}
  \label{eq:parametrization}
  {\sigma\over\eta}&=\sqrt{s^{-1}M_6}\, (\Lambda_a(x))^2\cr
{\zeta\over\eta}&=(s^{-1}M_6)^{3\over4}\,\Lambda_a(x)
       M_{a,b,c}(x)\,,
\end{align}
where $x\in\IC^\times/q^\ZZ$ and $\Lambda_a(x)$ and $M_{a,b,c}(x)$ are expressed
in terms of the Jacobi theta functions defined in Appendix~\ref{sec:theta-functions}
\begin{align}
  \label{eq:LM}
  \Lambda_a(z)&:={\theta_1(x)\over \theta_a(x)}\cr
M_{a,b,c}(z)&:={\theta_c^2\over\theta_a\theta_b}\,{\theta_a(x)\theta_b(z)\over(\theta_c(x))^2}\,,
\end{align}
that satisfy  the relation
\begin{equation}
(  M_{a,b,c}(x))^2  = (\Lambda_c(x))^4+ u_{a,b}\,(\Lambda_c(x))^2+1
\end{equation}
which is  consequence of the Jacobi relations in~\eqref{e:Jacobi} and 
in~\eqref{eq:thetaRelations}.

\medskip

The differences of $P_{ij}:=P_i-P_j$ are mapped to
\begin{align}
  P_{2,1}&=\left[\xi_1^2\xi_2^2,
  -{\xi_1^2\xi_2^2\over2}(t-\xi_1^2-\xi_2^2+\xi_3^2),1\right]  \\
 P_{3,2}&=\left[\xi_2^2\xi_3^2,
  -{\xi_2^2\xi_3^2\over2}(t+\xi_1^2-\xi_2^2-\xi_3^2),1\right]  \\
 P_{1,3}&=\left[\xi_1^2\xi_3^2,
  -{\xi_1^2\xi_3^2\over2}(t-\xi_1^2+\xi_2^2-\xi_3^2),1\right]  \,,
\end{align}
that implies that for $(i,j,k)$ a permutation of $(1,2,3)$ and $c=2,3,4$
\begin{equation}
  \left(\theta_1(x(P_{ij}))\over \theta_c(x(P_{ij}))\right)^2  =
                                                                   {\xi_k  \sqrt{
                                                                   s^{-1}}\over
                                                                  \xi_i
    \xi_j }
\end{equation}

 The differences $Q_{ij}:=Q_i-P_j$  are mapped to
\begin{align}
  Q_{3,2}&=\left[\xi_1^2t,
  {\xi_1^2t\over2}(s^{-1}+\xi_1^2-\xi_2^2-\xi_3^2),1\right]  \\Q_{1,3}&=\left[\xi_2^2t,
  {\xi_2^2t\over2}(s^{-1}-\xi_1^2+\xi_2^2-\xi_3^2),1\right]  \\Q_{2,1}&=\left[\xi_3^2t,
{\xi_3^2t\over2}(s^{-1}-\xi_1^2-\xi_2^2+\xi_3^2),1\right]  \,.
\end{align}
We  then deduce that for $(i,j,k)$ a permutation of $(1,2,3)$  and $c=2,3,4$
\begin{equation}
 \left(\theta_1(x(Q_{ij}))\over \theta_c(x(Q_{ij}))\right)^2  =
 {\xi_i \xi_j \over   \sqrt{
                                                                   s^{-1}}
                                                                 \xi_k  }\,.
\end{equation}

\medskip

Using that $\theta_1(-x)=\theta_2(x)$ and
$\theta_3(-x)=\theta_4(x)$, we find that  $x(Q_{ij})=-x(P_{ij})$
for $i=1,2,3$. Implying that for $i=1,2,3$ we have  $x(P_i)/x(Q_i)=-1 \in \IC^\times/q^\ZZ$, which
shows that the divisors $Q_i-P_i$ are of torsion two.   This will play an
important role when evaluating the elliptic dilogarithm in section~\ref{sec:elliptic-dilogarithm}.

\subsection{Derivation of the Picard-Fuchs equation}\label{sec:PFderivation}

For completeness we give a short and
explicit derivation of the differential equation satisfied by the
sunset integral
\begin{equation}
  L_\su\left( -{1\over s} \cIs(s)\right)= S_\su(s)  
\end{equation}
where $L_\su$ is the Picard-Fuchs operator (with $\delta_s:=sd/ds$)
\begin{equation}\label{e:LP3mass}
    L_\su= \delta_s^2 + q_1(s)\, \delta_s+ q_0(s) 
\end{equation}
and $S_\su(s)$ is the inhomogeneous term composed by 
 the sum of the Yukawa coupling $\cYs(s)$
and logarithmic contributions in the masses
\begin{equation}
  S_\su(s)= \cYs(s) +\sum_{i=1}^3 c_i(s) \log(\xi_i^2)\,.  
\end{equation}
The logarithms  terms arises from the Kummer
quotient extension of the motive described in
section~\ref{sec:motives} and in proposition~\ref{MKpropIII1}.

This differential equation has already
been derived in~\cite{Remiddi:2013joa,MullerStach:2011ru}.  We follow Griffiths'
approach in~\cite{Griffiths} for deriving the Picard-Fuchs equation
from the cohomology of smooth projective hyperspace defined by
rational form in $\IP^2$.

\medskip

The action of  the Picard-Fuchs operator on $\Omega_\su(s)$ is
\begin{equation}
    L_\su\,\Omega_\su(s)=\left( {2(xyz)^2\over \Phis^3}-
      {(3-q_1(s)) xyz\over \Phis^2}+{1-q_1(s)+q_0(s)\over
        \Phis}\right)\,\Omega
\end{equation}
with  $\Omega=x dy\wedge dz+y dz\wedge dx+z dx\wedge dy$ and where we
have set $\Phis=\linebreak xyz (1-s\phis)$.

For $C_x, C_y, C_z$ homogeneous polynomials  of degree 4 in
$(x,y,z)$ 
the one-form 
\begin{equation}
  \beta_1=  {y C_z-z C_y\over \Phis^2} \, dx+ {z C_x-x C_z\over \Phis^2}\,
  dy+ {x C_y-y C_x\over\Phis^2} dz
\end{equation}
satisfies\footnote{In general if deg$(C_i)=3k-2$  with $i=x,y,z$ the one-form
\begin{equation}
  \beta=  {y C_z-z C_y\over \Phis^k} \, dx+ {z C_x-x C_z\over \Phis^k}\,
  dy+ {x C_y-y C_x\over\Phis^k} dz
\end{equation}
satisfies
\begin{equation}
  d\beta= -k {  (C_x \partial_x+C_y \partial_y+C_z \partial_z )\Phis
    \over \Phis^{k+1}} \, \Omega+ { \partial_x C_x+\partial_y C_y + \partial_z C_z\over\Phis^k}\,\Omega\,.
\end{equation}
}
\begin{equation}
  d\beta_1= -2 {  (C_x \partial_x+C_y \partial_y+C_z \partial_z )\Phis
    \over \Phis^3} \, \Omega+ { \partial_x C_x+\partial_y C_y + \partial_z C_z\over\Phis ^2}\,\Omega\,.
\end{equation}
By choosing the polynomials $C_x$, $C_y$ and $C_z$ such that
\begin{equation}
  (xyz)^2=-(C_x \partial_x+C_y \partial_y+C_z \partial_z )\Phis  
\end{equation}
then
\begin{equation}
  {2(xyz)^2\over \Phis^3}\,\Omega  =- { \partial_x C_x+\partial_y
    C_y+\partial_z C_z\over  \Phis^2}\,\Omega+d\beta_1\,.
\end{equation}
The expressions of the polynomials $C_x,C_y,C_z$ are given in Appendix~\ref{sec:coeff}.
We choose the coefficient $q_1(s)$ so that
\begin{equation}
  \left(\partial_x C_x+\partial_y    C_y+\partial_z C_z\right)+(3-q_1(s)) xyz= (\tilde
  C_x \partial_x+\tilde C_y \partial_y+\tilde C_z\partial_z)\, \Phis
\end{equation}
where $\tilde C_x, \tilde C_y, \tilde C_z$ are at homogeneous polynomial
of degree one in $(x,y,z)$, 
which detailed expressions are given in Appendix~\ref{sec:coeff}.
We find that $q_1(s)$ is given by
\begin{equation}\label{eq:q1}
 q_1(s)=2+\sum_{i=1}^4 {1\over \mu_i^2
               s-1}-{2sM_2-6\over s^2\prod_{i=1}^4 \mu_i-2sM_2+3 }\,.
\end{equation}
The one-form 
\begin{equation}
  \beta_2=  {y \tilde C_z-z \tilde C_y\over \Phis} \, dx+ {z \tilde
    C_x-x \tilde C_z\over \Phis}\,
  dy+ {x \tilde C_y-y \tilde C_x\over\Phis} dz\,,
\end{equation}
satisfies
\begin{align}
	d\beta_2 &=- { \partial_x C_x+\partial_y
    C_y+\partial_z C_z+(3-q_1(s)) xyz\over \Phis^2}\,\Omega \\
\notag	&\quad + {\partial_x
    \tilde C_x+\partial_y
  \tilde  C_y+\partial_z \tilde C_z\over \Phis}\,\Omega\,.
\end{align}
Finally choosing  $q_0(s)$ such that
\begin{equation}
  q_0(s)= -1+q_1(s)+  \partial_x
    \tilde C_x+\partial_y
  \tilde  C_y+\partial_z \tilde C_z 
\end{equation}
leads to
\begin{equation}
      L_\su\,\Omega_\su(s) =d(\beta_1+\beta_2)\,.
\end{equation}

The expression for $q_0(s)$ is (see~\eqref{e:MassDef} for the definitions of $M_2$,
$M_4$ and $M_6$)
\begin{align}\label{eq:q0}
 q_0(s)&=-{n_0\over (s^2\prod_{i=1}^4 \mu_i-2sM_2+3)\prod_{i=1}^4(\mu_i^2s-1)}\cr
n_0&=-\mu_1^3 \mu_2^3   \mu_3^3 \mu_4^3 s^6 \\
	&\quad +s^5 \mu_1\mu_2\mu_3\mu_4\left(-3 M_2^3+12M_2 M_4+12M_6\right)\cr
&\quad +s^4 \left(-18 M_2^4+108 M_2^2
   M_4-120 M_2 M_6-144 M_4^2\right)\cr
&\quad +s^3 \left(26 M_2^3-96 M_2 M_4+324
   M_6\right)\cr
&\quad +s^2 \left(24 M_4-15 M_2^2\right)
+3 M_2 s\nonumber\,.
\end{align}

Acting with the Picard-Fuchs operator on the sunset integral gives
\begin{equation}
S_\su(s)=   \int_{\Delta} L_\su\,\Omega_\su
 =\int_{\Delta} d\beta \,,
\end{equation}
with $\beta=\beta_1+\beta_2=\beta_xdx+\beta_y dy+\beta_z dz$.

For evaluating this integral we consider the blow-up  $\tilde
\Delta$ of  the  domain of integration $\Delta=\{[x:y:z]\in\IP^2 |
x,y,z\geq0\}$, by putting a sphere of radius $\epsilon>0$
around each of the points $[1:0:0]$, $[0:1:0]$ and $[0:0:1]$.

Integration by part gives the boundary contributions 
\begin{align}
	S_\su(s)&=\lim_{\epsilon\to0}\int_{\partial \tilde \Delta|_{x=0}} \, \left( \beta_y dy+
      \beta_z dz\right)\\
     &\quad + \lim_{\epsilon\to0}\int_{\partial \tilde \Delta|_{y=0}} \, \left( \beta_x dx+
      \beta_z dz\right) \cr
	&\quad +\lim_{\epsilon\to0}\int_{\partial \tilde \Delta|_{z=0}} \, \left( \beta_x dx+
      \beta_y dy\right)\,.\nonumber
\end{align}
Where $\partial \tilde \Delta|_{x=0}$ denote the boundary of the
blown-up domain $\tilde\Delta$ restricted to the plane $x=0$.
Setting $\zeta=y/z$ in the first integral, setting $\zeta=z/x$ in the
second integral and $\zeta=x/y$ in the last integral we obtain
\begin{equation}
    S_\su(s)=\lim_{\epsilon\to0}\int_{\epsilon}^{1\over\epsilon} \, \left(
      z\beta_y+  x\beta_z + y\beta_x \right)d\zeta\,.
\end{equation}
With $z\beta_y=(a_1+b_1\zeta)/(\zeta\,(\xi_3^2+\xi_2^2\zeta))$  and
$y\beta_x=-(b_1+b_2\zeta)/(\zeta(\xi_2^2+\xi_1^2\zeta))$
and $x\beta_z=(- \xi_1^2
b_2/2+b_3\zeta-\xi_3^2a_1/2\zeta^2)/(\zeta(\xi_1^2+\xi_3^2\zeta)^2)$
where $a_1,  b_1, b_2, b_3$ are polynomials in $s$ reading
\begin{align}
	a_1 &=4
   (\xi_1^2-\xi_2^2)\xi_3^2s \Big(3 
-3s  \left(3 \xi_1^2+3 \xi_2^2-7
   \xi_3^2\right)\cr
	&\quad + s^2
   \left(9 \xi_1^4-10 \xi_1^2 \xi_2^2-14 \xi_1^2 \xi_3^2+9
   \xi_2^4-14 \xi_2^2 \xi_3^2+5 \xi_3^4\right)\cr
	&\quad -3 s^3 \mu_1\mu_2\mu_3\mu_4
 \left(\xi_1^2+\xi_2^2-\xi_3^2\right)\Big)\,,\nonumber
\end{align}
and $b_1$ is obtained from $a_1$ by exchanging $\xi_2$ and $\xi_3$,
the coefficient $b_2$ is
obtained from $a_1$ by exchanging $\xi_1$ and $\xi_3$, and finally 
\begin{align}
	b_3 &=6 \xi_1^2 \xi_3^2\Big( 9-s\left(13 \xi_1^2+10 \xi_2^2+13
   \xi_3^2\right)\cr
	&\quad + s^2 \left(\xi_1^4+27
   \xi_1^2 \xi_2^2+6 \xi_1^2 \xi_3^2-8 \xi_2^4+27 \xi_2^2
   \xi_3^2+\xi_3^4\right)\cr
	&\quad + s^3 (\xi_1^6+4
   \xi_1^4 \xi_2^2-\xi_1^4 \xi_3^2-15 \xi_1^2 \xi_2^4-24
   \xi_1^2 \xi_2^2 \xi_3^2\cr
	&\qquad\quad -\xi_1^2 \xi_3^4+10 \xi_2^6-15
   \xi_2^4 \xi_3^2
+4 \xi_2^2 \xi_3^4+\xi_3^6) \cr
	&\quad + s^4
  \mu_1\mu_2\mu_3\mu_4\left(2
   \xi_1^4-\xi_1^2 \xi_2^2-4 \xi_1^2 \xi_3^2-\xi_2^4-\xi_2^2
   \xi_3^2+2 \xi_3^4\right)\Big)\,.\nonumber
\end{align}

The integral has a finite limit when $\epsilon\to0$ given by  
\begin{equation}
  \label{e:Y3mass}
  S_\su(s)= \cYs(s) -{2s\,\sum_{i=1}^3 \log(\xi_i^2) \, c_i(s)\over \prod_{i=1}^4 (s\mu_i^2-1)  (s^2
  \prod_{i=1}^4\mu_i+2sM_2-3)} 
\end{equation}
where  the Yukawa coupling is given by\footnote{
By construction the Wronskian of the Picard-Fuchs operator is
$W_\su(s)=s^{-1}\, \cYs$.}
\begin{equation}\label{e:Yu}
    \cYs(s)=2
    {s^2\prod_{i=1}^4\mu_i-2sM_2+3\over \prod_{i=1}^4 (\mu_i^2s-1)}  \,.
\end{equation}
The coefficients satisfy  $c_1(s)+c_2(u)+c_3(s)=0$ and are
given by 
\begin{align}
	\label{e:c1}  c_1(s)&= 
-2\xi_1^2 +\xi_2^2+ \xi_3^2\\
&\quad +s \left(6\xi_1^4-7\xi_1^2\xi_2^2-3\xi_2^4-7\xi_1^2\xi_3^2+14\xi_2^2\xi_3^2-3\xi_3^4\right)\cr
&\quad + s^2 \left(-6\xi_1^6+11\xi_1^4\xi_2^2-8\xi_1^2xi_2^4+3\xi_2^6+11\xi_1^4\xi_3^2-3\xi_2^4\xi_3^2-8\xi_1^2\xi_3^4-3\xi_2^2\xi_3^4+3\xi_3^6\right)\cr
	&\qquad\quad -s^3\mu_1\mu_2\mu_3\mu_4(2\xi_1^4-\xi_1^2\xi_2^2-\xi_2^4-\xi_1^2\xi_3^2+2\xi_2^2\xi_3^2-\xi_3^4)
\nonumber
\end{align}
and $c_2(s)$ is obtained from $c_1(s)$ by exchanging $\xi_1$ and
$\xi_2$
\begin{align}
	\label{e:c2} c_2(s)&=\xi_1^2 -2\xi_2^2+ \xi_3^2\cr
&\quad +s \left(6\xi_2^4-7\xi_1^2\xi_2^2-3\xi_1^4-7\xi_2^2\xi_3^2+14\xi_1^2\xi_3^2-3\xi_3^4\right)\\
&\quad +s^2 \left(-6\xi_2^6+11\xi_2^4\xi_1^2-8\xi_2^2xi_1^4+3\xi_1^6+11\xi_2^4\xi_3^2-3\xi_1^4\xi_3^2-8\xi_2^2\xi_3^4-3\xi_1^2\xi_3^4+3\xi_3^6\right) \cr
	&\qquad\quad -s^3\mu_1\mu_2\mu_3\mu_4(2\xi_2^4-\xi_1^2\xi_2^2-\xi_1^4-\xi_2^2\xi_3^2+2\xi_1^2\xi_3^2-\xi_3^4)\,.\nonumber
\end{align}

\begin{rem}
In the all equal masses case $\xi_1=\xi_2=\xi_3=1$ we immediately have
that $y\beta_z=0$ and $y\beta_x=0$ and 
\begin{equation}
  x\beta_y= {36\over (s-1)(9s-1)(1+\zeta)^2}
\end{equation}
leading to $S_\su(s)=6/( (9s-1)(s-1))$ which is the Yukawa
coupling $\cYs(s)$.  The Picard-Fuchs operator reads (with $\delta_s:=sd/ds$)
\begin{equation}
\label{e:PF1mass}
 L_\su= \delta_s^2+ {2s(9s-5)\over (s-1)(9s-1)} \delta_s+ {3s(3s-1)\over (s-1)(9s-1)}\,.
\end{equation}
The sunset integral satisfies the differential equation
\begin{equation}
  L_\su\left( -{1\over s}\cIs(s)  \right)= {6\over (9s-1)(s-1)}\,,
\end{equation}
which is equivalent to
\begin{equation}
  s(9s-1)(s-1) {d^2\cIs(s)\over ds^2} +(9s^2-1){d\cIs(s)\over ds}+ {1-3s\over s} \cIs(s)= -6\,.  
\end{equation}
This differential equation has been presented in the following matrix
form in~\cite[eq.~(4.13)]{Henn:2014qga}
\begin{equation}
  {d\over ds}\, \vec f(s)= {A_0\over s}  + {A_1\over s-1}+ {A_9\over 9s-1} \,.
\end{equation}
The poles are located at the singular fibers of the
sunset elliptic curves family. The residues  are the monodromy
matrices
$A_0$, $A_1$ and $A_9$ which are independent of $s$. 
These squared matrices have size three which equal the (generic) rank
of the  all equal
masses sunset  motive~\cite{Bloch:2013tra}.  This first order equation
arises from the flat Gau\ss-Manin connection for the coherent analytic
sheaf for which the section $\tilde\sigma$ leads to  sunset Feynman integral
according~\eqref{amp} as proven in lemma~6.21~of~\cite{Bloch:2013tra}
for the all equal masses case.
\end{rem}

\subsection{The elliptic dilogarithm}
\label{sec:elliptic-dilogarithm}

For $s\in ](\xi_1+\xi_2+\xi_3)^{-2},+\infty[$ we provide an expression of the
sunset integral $\cIs(s)$ in~\eqref{e:IsunsetSym} 
in term of the elliptic dilogarithms.  A derivation via motives will
given in section~\ref{sec:motives}.

We start by considering the ratio of the coordinates on the sunset
cubic curve as functions  on $\IC^\times/q^\ZZ$ 
\begin{align}
  {X\over Z}(x)& ={\theta_1(x/x(Q_1))\theta_1(x/x(P_3))\over
  \theta_1(x/x(P_1))\theta_1(x/x(Q_3))}\cr
 {Y\over Z}(x)& ={\theta_1(x/x(Q_2)) \theta_1(x/x(P_3))\over
  \theta_1(x/x(P_2))\theta_1(x/x(Q_3))}\,.
\end{align}
where $x(P)$ is the representation of the point $P$ in $\mathcal
E_\su\simeq \mathbb C^\times/q^\ZZ$ using the map of
section~\ref{sec:points}, and  $\theta_1(x)$  is the Jacobi theta function
\begin{equation}
\theta_1(x)=q^{1\over8} {x^{1/2}-x^{-1/2}\over \ay}\,\prod_{n\geq1} (1-q^n)(1-q^nx)(1-q^n/x)\,.
\end{equation}
We evaluate the integral 
\begin{equation}
  F(x)=- \int^x_{x_0} \log \left({X\over Z}(y)\right) \, d\log y
  \end{equation}
where $x_0$ is an arbitrary origin that will cancel in the final
answer. We find
\begin{align}
	F(x)&=F(x_0)+E_2(x/x(P_1))+E_2(x/x(Q_3))\\
	\notag &\quad - E_2(x/x(P_3))-E_2(x/x(Q_1))
\end{align}
where $E_2(x)$ is the elliptic dilogarithm
\begin{equation}
  E_2(x)
=\sum_{n\geq0} \Li_2(q^nx)-\sum_{n\geq1} \Li_2(q^n/x)+\frac14\,(\log (x))^2-\ay\pi\log (x)
\end{equation}
Using the 2-torsion relations $x(Q_i)=-x(P_i)$ for $i=1,2,3$ we can
rewrite $F(x)$ as 
\begin{equation}\label{e:FtoE}
  F(x)=\hat E_2(x/x(P_1))- \hat
  E_2(x/x(P_3))+{\ay\pi\over2}\,\log\left(x(P_1)\over x(P_3)\right)+F(x_0)\,,
\end{equation}
where
\begin{align}
	\hat E_2(x)&= \sum_{n\geq0} \left(\Li_2(q^nx)-\Li_2(-q^nx)\right)\\
	\notag &\quad -\sum_{n\geq1} \left(\Li_2(q^n/x)-\Li_2(-q^n/x)\right)  \,.
\end{align}

With this we can evaluate on the zero or poles of $Y/Z$
\begin{equation}
    \mathcal L_2\left\{{X\over Z},{Y\over Z}\right\}=F(x(P_3))+F(x(Q_2))-F(x(P_2))-F(x(Q_3))\,.
\end{equation}
The origin of the integral $F(x_0)$  has cancelled in the expression.
Using the expression for $F(x)$ in~\eqref{e:FtoE} one gets
\begin{align}
	\mathcal L_2\left\{{X\over Z},{Y\over Z}\right\}
	&=\hat E_2\left(-x(P_2)\over x(P_1)\right) -\hat E_2\left(x(P_2)\over
  x(P_1)\right)+\hat
E_2\left(x(P_2)\over x(P_3)\right)\\
&\quad -\hat
E_2\left(-x(P_2)\over x(P_3)\right) +\hat E_2\left(x(P_3)\over
  x(P_1)\right) -\hat E_2\left(-x(P_3)\over x(P_1)\right)\cr
&\quad +\hat E_2(-1)-\hat E_2(1)\,.\nonumber
\end{align}
Noticing the following properties of the function $\hat E_2(x)$
\begin{align}
  \hat E_2(-x)&=-\hat   E_2(x)\\
\nn\hat E_2(1/x)&=-\hat E_2(x)+\Li_2(x)-\Li_2(-x)+\Li_2({1\over
                    x})-\Li_2(-{1\over x})
\end{align}
together with  the dilogarithm functional equation
\begin{equation}
  \Li_2(x)+\Li_2(1/x)=-{\pi^2\over6}-\frac12\log(-x)^2  
\end{equation}
we can reduce the expression for $\mathcal L_2\{X/Z,Y/Z\}$ to 
\begin{align}\label{e:L2}
	\mathcal L_2\left\{{X\over Z},{Y\over Z}\right\}&=2\hat E_2\left(x(P_1)\over
  x(P_2)\right)+2\hat E_2\left(x(P_2)\over x(P_3)\right)+2\hat
E_2\left(x(P_3)\over x(P_1)\right) \\
	\nonumber &\quad +{\pi^2\over4}-\ay\pi\log\left(x(P_1)\over x(P_2)\right)\,.
\end{align}
The elliptic dilogarithm  $\hat E_2(x)$ is not invariant under
$q$-translation and transforms according
\begin{align} \label{e:Eh2shift}
  \hat E_2(qx)&= \hat E_2(x)-{\pi^2\over2}+\ay\pi\log(x)\\  
  \hat E_2(x/q)&= \hat E_2(x)+{\pi^2\over2}-\ay\pi\log(x/q)\,.
\end{align}
This is because the Feynman integral we are studying is a multivalued
function.
Shifting the representative $x(P)$ of the point $P$ in
$\IC^\times/q^\ZZ$ changes the expression for $  \mathcal
L_2\left\{{X\over Z},{Y\over Z}\right\}$  modulo $\ay\pi\log q$, $\ay\pi
\log(x(P_1))$,  $\ay\pi
\log(x(P_2))$ or $\ay\pi
\log(x(P_3))$.

In order to fix this ambiguity we symmetrize the computation by
summing other all the other choices to get

\begin{align}
   \mathcal L_2&:=  \mathcal L_2\left\{{X\over Z},{Y\over Z}\right\}+  \mathcal L_2\left\{{X\over
       Y},{Z\over Y}\right\}+  \mathcal L_2\left\{{Y\over X},{Z\over
       X}\right\}\\
&=6\hat E_2\left(x(P_1)\over
  x(P_2)\right)+6\hat E_2\left(x(P_2)\over x(P_3)\right)+6\hat
E_2\left(x(P_3)\over x(P_1)\right) \nonumber
+{3\pi^2\over4}\,.
\end{align}

\subsubsection{The all equal masses case }
\label{sec:single-mass-case}

It was shown in~\cite{Bloch:2013tra}  that the all equal masses case $\xi_1=\xi_2=\xi_3=1$ sunset
integral is given by
\begin{equation}\label{e:sunsetBV}
  \cIs(s_\su(q))={\varpi_r\over\pi} \,\left(\ay \pi^2 \, (1-2\tau) +  \, E_\su(q) \right)
\end{equation}
where  $q=\exp(2\pi\ay \tau)$ with $\tau=\varpi_c/\varpi_r$  the period
ratio and
$s_\su$ is the Hauptmodul
\begin{equation}
\label{e:sHaupt}  s_\su(q)^{-1}= 9+72\, {\eta(q^2)\over\eta(q^3)}\,\left(\eta(q^6)\over\eta(q)\right)^5  
\end{equation}
and $E_\su(q)$ is the elliptic dilogarithm evaluated that the sixth
root of unity $\zeta_6=e^{\ay\pi\over3}$
\begin{align}\label{e:Esunset1mass}
  E_\su(q)&={1\over2\ay} \sum_{n\geq1} \left(\Li_2(q^n
              \zeta_6)+\Li_2(q^n\zeta_6^2)-Li_2(q^n\zeta_6^4)-\Li_2(q^n\zeta_6^5)\right)\\
&\quad +{1\over 4\ay}\left(\Li_2(\zeta_6)+Li_2(\zeta_6^2)-\Li_2(
              \zeta_6^4)-\Li_2(\zeta_6^5)\right)\,.\nonumber
\end{align}
Noticing that
\begin{equation}
 2\ay E_\su(q)= \hat E_2(\zeta_6^2)+\zeta(2)
\end{equation}
and since
when all the masses are equal the image in $\IC^\times/q^\ZZ$ of
the points $x(P_i)=\zeta_6^i$ with $i=1,2,3$, we have
\begin{equation}\label{e:R1mass}
  \mathcal L_2\left\{{X\over Z},{Y\over Z}\right\}
=2\ay E_\su(q) +{11\pi^2\over 3}\,.
\end{equation}
Showing that the all equal masses
 sunset integral is equal to the regulator~\eqref{e:R1mass}
modulo periods of the elliptic curves
\begin{equation}
I_\su\equiv {\varpi_r\over2\pi\ay}  \mathcal L_2\left\{{X\over Z},{Y\over Z}\right\}\mod \textrm{periods}\,.
\end{equation}
%

\subsubsection{Three masses case}
\label{sec:three-mass}

In the three masses case  the sunset integral in~\eqref{e:IsunsetSym} is given by
\begin{equation}\label{e:IsunsetR}
  \cIs(s)\equiv {\ay\varpi_r   \over\pi} \left(\hat E_2\left(x(P_1)\over
      x(P_2)\right)+\hat E_2\left(x(P_2)\over x(P_3)\right)+\hat
    E_2\left(x(P_3)\over x(P_1)\right)\right)
\mod  \textrm{periods}
\end{equation}
An expression in terms of multiple polylogarithms has been presented in \cite{Adams:2014vja,Adams:2015gva,Adams:2015ydq,Adams:2016vdo}. It would be interesting to relate these results.

\medskip 
A proof is given in section~\ref{sec:motives} using
a motivic approach. In this section we present numerical verification
of this expression for the sunset integral.

According~\eqref{e:Eh2shift}  the elliptic dilogarithm $\hat E_2(x)$ is not invariant
under  the  change $x(P_i) \to qx(P_i)$ therefore  the expression in~\eqref{e:IsunsetR} shifts
by $\ay\pi\log q$.  Therefore by changing the representative of
$P_1$, $P_2$ and $P_3$ in $\IC^\times/q^\ZZ$ one can change the
coefficients of the periods of the elliptic curve freely.  

\subsubsection{Numerical checks}
We have made some numerical checks (see table~\ref{tab:3massNum} on page~\pageref{tab:3massNum}) of this relation using {\tt
  PARI/GP}~\cite{PARI}.  For given values of the masses and $s$ we
have searched for an integer linear dependence of the vector
\begin{equation}\label{e:vDef}
  v=\left[\cIs(s)- \ay{\varpi_r\over\pi}\, \mathcal L_2\left\{{X\over
        Z},{Y\over Z}\right\}, \ay\pi\varpi_r,\ay\pi
    \varpi_c\right]  
\end{equation}
using the {\tt lindep} command of {\tt PARI/GP}.
The vector is composed of the sunset integral evaluated using the
Bessel integral representation~\cite{Groote:2005ay,Bailey:2008ib,Groote:2012pa,PVstringmath}
\begin{equation}
  \label{e:Ibessel}
  \cIs(s) = \int_0^\infty 4x I_0(\sqrt {s^{-1}} x)\, \prod_{i=1}^3
  K_0(\xi_ix)\, dx\,,
\end{equation}
the regulator evaluated as
\begin{equation}
     \mathcal L_2\left\{{X\over
        Z},{Y\over Z}\right\}=\hat E_2\left(x_{12}\right)+\hat E_2\left(x_{23}\right)+\hat
    E_2\left(x_{31}\right)
\end{equation}
where $x_{ij}=x(P_i)/x(P_j)$ in $\IC^\times/q^\ZZ$.
Since we can easily follow the change of the expression under a
$q$-translation of the points, we have made some choices such that the relation between the sunset integral and the
regulator is modulo periods of the elliptic curves with simple
rational coefficients, keeping the relation $x_{12}x_{23}x_{31}=1$.
For instance in table~\ref{tab:3massNum} for the case $(\xi_1,\xi_2,\xi_3,s^{-1})=(1,2,3,3)$,
we show how the
$q$-translations $(x_{12},x_{31})\to (q x_{12},x_{31}/q)$ affect the
result modulo periods of the elliptic curve.

\begin{table}[ht]
	\resizebox{\textwidth}{!}
{\begin{tabular}[h]{||c|c|c|c|c|c||}
\hline
    $\xi_1, \xi_2, \xi_3, s^{-1}$&$x_{21}$& $x_{32}$&$q$&{\tt lindep(v)}&prec\\
\hline
    $1,2,8,2$&$-0.00931124 + 0.0160094\ay$&$4.87147 - 5.50124\ay$&$0.136089$& $[4, 3,-8]$&$2\,10^{-36}$\\
    $1,2,8,6$&$-0.00640431 + 0.00671999\ay$&$6.17736 -
                                           8.34052\ay$&$0.0963482$&$[4,                                                                  3,
                                                                  -8]$&$2\,10^{-36}$\\
&&&&&\\
    $1,2,3,2$&$0.0733690 - 0.108597\ay$&$0.797236 -
                                        0.603668\ay$&$-0.131059$&
                                                                 $[4,1,0]$&$6\,10^{-37}$\\
$1,2,3,2$&$-0.00961565 + 0.0142326\ay$&$-6.08304 + 4.60608\ay$&$-0.131059$&$[4, -5, 8]$&$6\,10^{-37}$\\
 $1,2,3,3$&$-0.723282 - 0.690553\ay$&$ -0.145143 -
                                    0.107284\ay$&$-0.180489$& $[4, 3,
                                                            -8]$&$5\,10^{-37}$\\
&&&&&\\
    $1,5,7,3$&$-0.481821 + 0.876270\ay$&$ -0.0416592 + 0.0163910\ay$&$-0.0447678$& $[4,-5,8]$&$10^{-37}$\\
 $1,5,7,7$&$-0.766655 + 0.642059\ay$ &$-7.08429 + 2.58610\ay
                                     $&$-0.132599$& $[4, -5,8]$&$7\,10^{-38}$\\
&&&&&\\
$3,5,7,3$& $0.199999 + 0.979796\ay$&$-6.29720 +
                                    3.35123\ay$&$-0.140185$&$[4, -5,
                                                            8]$&$7\,10^{-38}$\\
$3,5,7,7$& $-0.199528 + 0.979892\ay$&$-5.76891 + 4.08260\ay$&$-0.141495$&$[4, -5, 8]$&$5\,10^{-38}$\\
\hline
\end{tabular}}\vspace{1ex}
  \caption{Results of linear dependence of vector $v$ defined
    in~\eqref{e:vDef} using the {\tt PARI/GP} command {\tt
      lindep(v)}. The last column gives the absolute value for the
    numerical evaluation linear relations. }
  \label{tab:3massNum}
\end{table}



\newcommand{\eq}[2]{\begin{equation}\label{#1}#2 \end{equation}}
\newcommand{\ml}[2]{\begin{multline}\label{#1}#2 \end{multline}}
\newcommand{\ga}[2]{\begin{gather}\label{#1}#2 \end{gather}}
\newcommand{\mypmatrix}[1]{\begin{pmatrix}#1 \end{pmatrix}}

\newcommand{\C}{{\mathbb C}}
\newcommand{\G}{{\mathbb G}}
\newcommand{\Q}{{\mathbb Q}}
\newcommand{\Z}{{\mathbb Z}}
\renewcommand{\P}{{\mathbb P}}
\newcommand{\R}{{\mathbb R}}

\newcommand{\sS}{{\mathcal S}}

\newcommand{\inj}{\hookrightarrow}

\newcommand{\bloch}{{\bf SB }}

\section{Approach via Motives}\label{sec:motives}

The purpose of this section is to prove formula \eqref{e:IsunsetR} for
the sunset Feynman integral in two dimensions with arbitrary masses. This is a beautiful illustration how the theory of motives can yield information about Feynman integrals. With an eye toward future applications, we will permit ourselves to say a bit more than what is strictly necessary for the sunset case.  

We fix masses and external momenta and just write $E$ for the
resulting elliptic curve, which is an element in the family
\eqref{e:Esunset}. We have $E\inj \P^2$ with homogeneous
coordinates $X, Y, Z$, and $E$ meets the coordinate triangle $XYZ=0$
in a set $S$ of $6$ points, $S:= \{P_1, P_2, P_3, Q_1, Q_2,
Q_3\}$. Let $E^0:=E-S$. Following \cite[\S6]{Bloch:2013tra},  let $\rho: P \to \P^2$ be the blowup of the vertices of the coordinate triangle, and let $\frak h \subset P$ be the resulting hexagon. We can lift $E\inj P$ and write $\frak h^0=\frak h-E\cap \frak h$. 

Let $\sigma \subset \P^2(\R)$ be the positive real simplex which is the chain of integration for the Feynman integral. for general values of external momenta, $\sigma$ will not meet $E$, and we can lift to $\widetilde\sigma \subset P-E$. We have $\partial\widetilde\sigma \subset \frak h^0$, so 
\eq{sig}{\widetilde\sigma \in H_2(P-E,\frak h^0;\Q)= H^2(P-E,\frak h^0;\Q)^\vee.
}  
The form $\Omega_\su$, \eqref{e:OmegaDef}, represents a class in
$F^2H^2(P-E,\frak h^0;\C)$, and the Feynman integral 
\eq{amp}{\cIs = \langle \Omega_\su,\widetilde\sigma\rangle.
} 
The idea is to interpret $\cIs$ as a quantity intrinsic to the
Hodge structure $H^2(P-E,\frak h^0;\Q)$ together with the choice of
$\Omega_\su$. That way, whenever we see the Hodge structure (and we
will see it in two other guises below) we can be sure that the sunset Feynman
integral  $\cIs$ is involved. 

To begin, we can invoke \cite{BKV},  lemma 6.1.4 to get
\eq{dual}{H_2(P-E,\frak h^0;\Q)= H^2(P-E,\frak h^0;\Q)^\vee \cong H^2(\G_m^2,E^0;\Q(2)).
}
Here we identify 
\eq{}{\G_m^2 = \P^2-\{XYZ=0\}=P-\frak h.
}
We consider the long-exact sequence of Hodge structures
\begin{align}\label{les}
H^1(\G_m^2,\Q(2)) \xrightarrow{\alpha} H^1(E^0,\Q(2)) &\to H^2(\G_m^2,E^0;\Q(2)) \\
\notag &\to H^2(\G_m^2,\Q(2)) \to 0.
\end{align}
The image of $\alpha$ above is spanned by the logarithmic classes 
$$d\log(X/Z), d\log(Y/Z).$$ 
We can avoid these by replacing $\G_m^2$ by the relative space $(\G_m,\{1\})^2$. One has
\eq{}{H^i((\G_m,\{1\})^n,\Q) = \begin{cases} 0 & i\neq n \\
\Q(-n) & i=n
\end{cases}
}
We now build a diagram \minCDarrowwidth.1cm
\eq{comdiag}{\begin{CD}0 @>>> H^1(E^0,\Q(2)) @>>> H^2((\G_m,\{1\})^2,E^0;\Q(2)) @>>> \Q(0) @>>> 0 \\
@. @VV V @VV a V @| \\
0 @>>> H^1(E^0,\Q(2))/Im(\alpha) @>>> H^2(\G_m^2,E^0;\Q(2)) @>>> \Q(0) @>>> 0 
\end{CD}
}
where the bottom line comes from truncating \eqref{les}. 

We are interested in the extensions of Hodge structures associated to these sequences. Since the sequence on the bottom comes by pushout, it\linebreak will suffice to consider the top line. We consider splittings $s_\Q \in\linebreak H^2((\G_m,\{1\})^2,E^0;\Q(2))$ and $s_F \in F^0H^2((\G_m,\{1\})^2,E^0;\C(2))$ lifting $1\in \Q(0)$. The obstruction to splitting the sequence of Hodge structures \eqref{comdiag} is
\eq{}{s_\Q-s_F \in H^1(E^0,\C(2))/H^1(E^0,\Q(2))
}
We can choose $s_\Q$ so its image in $H^2(\G_m^2,E^0;\Q(2))$ coincides with $\widetilde\sigma$ under the identification \eqref{dual}. Indeed, the boundary $\partial\widetilde\sigma = 1 \in H_1(\frak h^0) \cong \Q(0)$. Also, in \eqref{comdiag} the dual to the map labeled $a$ induces an isomorphism on $F^2$. (This is because $Im(\alpha)^\vee = \Q(-1)^2$ and $F^2\C(-1)=(0)$.) In particular, $\Omega_\su$ lifts canonically to an element 
$$\Omega\in F^2 H^2((\G_m,\{1\})^2,E^0;\C(2))^\vee. $$ 
Note that this element is orthogonal to $F^{-1}H^2((\G_m,\{1\})^2,E^0;\C(2))$ so in particular it kills $s_F\in F^0 \subset F^{-1}$. We conclude
\eq{9c}{\cIs = \langle \Omega_\su,\widetilde\sigma\rangle = \langle \Omega,s_\Q-s_F\rangle. 
}

Our objective now is to reinterpret $I_\su$ in terms of elliptic dilogarithms. Let $\square := \P^1-\{1\}$ and write $\partial\square = \{0,\infty\}$. Poincar\'e duality yields an identification
\eq{}{H^*((\square,\partial\square)^n) \cong H^{2n-*}((\G_m,1)^n)(n)^\vee \cong \begin{cases} \Z(0) & *=n \\ 0 & \text{else}. \end{cases}
}
Let $\Gamma^0 \subset E^0\times \G_m^2$ be the graph of the embedding $E^0 \inj \G_m^2$. Note that $\Gamma^0$ is actually closed in $E^0\times (\P^1)^2$ and we may intersect to get a closed codimension $2$ cycle which we also call $\Gamma^0$ on $X\times (\P^1-\{1\})^2$. This cycle doesn't meet the loci where coordinates $\in \{0,\infty\}$. The Gysin sequence yields \minCDarrowwidth.1cm
\eq{14c}{\begin{small}\begin{CD}0 @>>> H^{3}(E^0\times (\square,\partial\square)^2)(2) @>>> H^{3}(E^0\times (\square,\partial\square)^2-\Gamma^0)(2) @>>> \Z(0) @>>> 0 \\
@. @| \\
@. H^{1}(E^0)(2) \end{CD}
\end{small}}   

\begin{lem}\label{lem1}
 The sequence \eqref{14c} and the top row of \eqref{comdiag} agree as extensions of Hodge structure.
\end{lem}
\begin{proof}Note that we can generalize the top row of \eqref{comdiag} to an extension
\eq{}{0 \to H^{n-1}(X)(n) \to H^n((\G_m,\{1\})^n, X)(n) \to \Q(0) \to 0
}
for any $f: X \to \G_m^n$. (To avoid technicalities, we assume in the sequel that $X$ is smooth). We will construct a commutative diagram \minCDarrowwidth.1cm
\eq{tobeconst}{\begin{small}\begin{CD}0 @>>> H^{n-1}(X)(n) @>>> H^n((\G_m,\{1\})^n, X)(n) @>>> \Q(0) @>>> 0 \\
@. @VV\cong V @VVV @| \\
0 @>>> H^{2n-1}(X\times (\square,\partial\square)^n)(n) @>>> H^{2n-1}(X\times (\square,\partial\square)^n-\Gamma_f)(n) @>>> \Z(0) @>>> 0\end{CD}
\end{small}}

We consider the universal case $X=\G_m^n,\ f=id$. 
 Let $\Xi \subset \G_m^n\times \square^n$ be the corresponding graph. Note $\Xi \cong (\G_m-\{1\})^n$. We want to understand $H^*((\G_m,\{1\})^n\times (\square,\partial\square)^n-\Xi)$. Consider the projection 
\eq{15b}{p: (\G_m,\{1\})^n\times (\square,\partial\square)^n-\Xi \to \G_m^n.
}
The cohomology on the left is calculated by the sheaf $\sS$ which is the constant sheaf with fibre $\Q$ on $(\G_m-\{1\})^n\times (\square-\partial\square)^n$ extended by $0$ to $\G_m^n\times \square^n$ and then restricted to the complement of $\Xi$. 
For $z=(z_1,\dotsc,z_n)\in \G_m^n$ we have
\eq{}{\sS|_{p^{-1}(z)} = \begin{cases} \Q_{\square^n-\{z\}} & \text{no $z_i = 1$} \\
(0) & \text{some $z_i=1$}
\end{cases}
}
The cohomology along the fibres of $p$ is thus 
\eq{}{H^*(\sS|_{p^{-1}(z)}) = \begin{cases} (0) & *\neq n, 2n-1;\text{ or some $z_i=1$} \\
\Z(0) & *=n;\  \text{no $z_i = 1$} \\
\Z(-n) & *=2n-1;\  \text{no $z_i = 1$}.
\end{cases}
}
Using again that $H^*((\G_m,\{1\})^n)=(0)$ for $*\neq n$, we conclude that in the Leray spectral sequence associated to $p$, \eqref{15b}, one has
\eq{}{E_2^{ab} \Rightarrow H^{a+b}((\G_m,\{1\})^n\times (\square,\partial\square)^n-\Xi)
}
and $E_2^{ab}=(0)$ unless $a=n$ and $b=n, 2n-1$. In particular, 
$$H^{2n-1}((\G_m,\{1\})^n\times (\square,\partial\square)^n-\Xi) = (0).$$
The Gysin sequence yields
\begin{align}0 \to H^0(\Xi)(-n) &\xrightarrow{gysin} H^{2n}((\G_m,\{1\})^n\times (\square,\partial\square)^n) \\ 
	\notag &\xrightarrow{restrict}  H^{2n}((\G_m,\{1\})^n\times (\square,\partial\square)^n-\Xi). 
\end{align}
Both domain and target of the map labeled $gysin$ are $\Q(-n)$. Since this map is injective, it is an isomorphism, so the map labeled $restrict$ is zero. 

We now have a diagram (to shorten we write $B=(\square,\partial\square)^n$)
\minCDarrowwidth.1cm
\eq{}{\begin{small} \begin{CD}0 @>>> H^{2n-1}(X\times B) @>>> H^{2n}(((\G_m,\{1\})^n,X)\times B) @>>> H^{2n}((\G_m,\{1\})^n\times B) \\
@. @VVV @V a VV @V 0 VV \\
 0 @>>> H^{2n-1}(X\times B-\Gamma_f) @>>> H^{2n}(((\G_m,\{1\})^n,X)\times B-\Xi) @>>> H^{2n}((\G_m,\{1\})^n\times B - \Xi).
\end{CD}
\end{small}} 
As a consequence of the above calculations, the map on the lower left is injective and the vertical map on the right is zero. It follows that the vertical map labeled $a$ lifts to $\tilde a$ fitting into a diagram \minCDarrowwidth.1cm 
\eq{21c}{\begin{small} \begin{CD}0 @>>> H^{2n-1}(X\times B) @>>> H^{2n}(((\G_m,\{1\})^n,X)\times B) @>>> H^{2n}((\G_m,\{1\})^n\times B) @>>> 0 \\
@. @| @V \tilde a VV @V\cong VV \\
0 @>>> H^{2n-1}(X\times B) @>>> H^{2n-1}(X\times B-\Gamma_f) @>\text{residue}>> \Q(-n) @>>> 0
\end{CD}
\end{small}}
After twisting by $\Z(n)$ we find \eqref{21c} coincides with \eqref{tobeconst}, proving the lemma. 
\end{proof}

We can now compute $\cIs$ using \eqref{14c}. We have, again by lemma 6.1.4 in  \cite{BKV}  (writing $S:= E-E^0$)
\eq{}{H^{3}(E^0\times (\square,\partial\square)^2)(2)^\vee = H^3((E,S)\times (\G_m,\{1\})^2)(1).
}
We fix coordinates $x,y$ on $\G_m^2$ and a holomorphic $1$-form $de$ on $E$. The role of $\Omega$ in \eqref{9c} will be played by 
\eq{}{\eta := de\wedge dx/x\wedge dy/y \in F^2(H^3((E,S)\times (\G_m,\{1\})^2)(1)).
}
Let $S:=E-E^0$. A homological interpretation of the top row of \eqref{comdiag} is rather tricky. We need to define the group $H_3((E,S)\times (\G_m,\{1\})^2,\Gamma;\Q)$ where $\Gamma \cong E$ is the complete curve. To justify this, let 
\eq{}{\G_m -\{1\}\stackrel{\ell}{\inj} \G_m  \stackrel{k}{\inj}\P^1;\quad E^0 \stackrel{j}{\inj} E
} 
be the open immersions. Let $f: E \inj E\times \P^1\times \P^1$ extend the graph $E^0 \inj E^0\times \G_m^2$. The point is that the natural map over $E^0$ extends to
\eq{}{f^*(j_!\Q_{E^0}\boxtimes k_*\ell_!\Q_{\G_m-\{1\}}\boxtimes k_*\ell_!\Q_{\G_m-\{1\}}) \to \Q_E.
}
This is because the points on $E\times \P^1\times \P^1$ where $\Gamma$ meets $(\{0,\infty\}\times \P^1)\cup (\P^1\times \{0,\infty\})$ are contained in $S\times \P^1\times \P^1$ so the stalks of the sheaf $j_!\Q_{E^0}\boxtimes k_*\ell_!\Q_{\G_m-\{1\}}\boxtimes k_*\ell_!\Q_{\G_m-\{1\}}$ are zero.

We will integrate $\eta$ over a relative homology $3$-chain $C$ on $E\times \P^2$. An argument (left to the reader) similar to the above will show that $C$ represents a class in $H_3((E,S)\times (\G_m,\{1\})^2, \Gamma;\Q)$ and that $\partial C=\Gamma$.  Define (cf. \cite{KerrLewisSMS})
 \eq{eq:4.25}{ C := \{(e, (1-v)+vx(e), y(e)) \ | \ e\in E, \ 0\le v\le 1\}}
Cut $E$ and $C$ along the locus $T_y := \{e\ |\ y(e)\le 0\}$. On the cut chain we can write $dy(e)/y(e) = d(\log y(e))$ and apply Stokes theorem. (More precisely, $T_y$ is an infinitely thin strip with two sides. The value of $\log y$ differs by $2\pi \ay$ at corresponding points on the two sides of $T_y$, so we find
\eq{}{\int_C \eta = 2\pi \ay\int_\gamma \log(x)de}
where $\gamma$ is a $1$-chain with $\partial \gamma = (y)$, the divisor of $y=Y/Z$ on $E$.
Using \eqref{9c} and lemma \ref{lem1}, we deduce
\begin{prop} The sunset Feynman integral 
\eq{eq:4.30}{\cIs = \kappa \int_{\gamma}\log(x)\cdot \eta
}
where $\kappa \eta = \Omega$ under the identification 
\eq{}{F^3 H^3((E,S)\times (\G_m,\{1\})^2,\C) = F^2H^3(E\times (\G_m,\{1\})^2,\C(1)).
}
\end{prop}
\begin{remark} Note that the $2$-chain $\widetilde\sigma$ \eqref{sig} defines (after reinterpretation in terms of cohomology as explained above) a splitting of the bottom row of \eqref{comdiag}. This chain does not lift to yield a splitting of the top row. This is because the chain $\sigma$ meets the lines $X=Z$ and $Y=Z$ in $\P^2$ and so does not represent a class in $H_2(\P^2-E-\{X=Z\}-\{Y=Z\})$. The effect of this is to introduce some elementary log terms corresponding to periods of $d(X/Z)/(X/Z)$ in~\eqref{e:L2}.
\end{remark}

Consider one last time the top sequence from \eqref{comdiag}. Writing $M$ for the middle group in this sequence, we see that the weight-graded pieces are
\eq{}{W_iM_\Q = \begin{cases} H^1(E,\Q(2)) & i=-3 \\
\bigoplus_5 \Q(1) & i=-2 \\
\Q(0) & i=0.
\end{cases}
}
$M$ should be viewed as a representation of a sort of generalized
graded Lie algebra with graded pieces the above pure Hodge structures
(or pure motives). The Feynman integral is a period associated to the lower lefthand corner of the representation matrix. In trying to generalize to more complicated Feynman diagrams, two problems arise. Firstly, the pieces one sees combinatorially by shrinking edges on the graph have Hodge structures which are themselves mixed rather than pure. And secondly, it is not possible in general to make the intersection between the polar locus $X$ and the simplex at infinity transverse by simply blowing up faces of the simplex. Presumably, therefore, the analog of the duality used above
\eq{}{H^2(P-E,\frak h^0)^\vee \cong H^2(\G_m^2, E^0)(2)
} 
is not valid in general, which means that the link between the Feynman
integral and polylogarithms is more tenuous. 

\begin{remark}
In section \ref{sec:multi}, formula \eqref{eq:4.30} is rederived in terms of regulator currents, as a byproduct of the $K$-theoretic approach to the inhomogeneous Picard-Fuchs equation.  The relationship between these currents and $C$ (in \eqref{eq:4.25}) is explained in \cite{KerrLewisSMS}.
\end{remark}


\def\aa{\mathtt{a}}
\def\bb{\mathtt{b}}
\def\cc{\mathtt{c}}
\def\ay{\mathbf{i}}
\def\d{\partial}
\def\CC{\mathbb{C}}
\def\cB{\mathcal{B}}
\def\QQ{\mathbb{Q}}
\def\DD{\mathbb{D}}
\def\BB{\mathbb{B}}
\def\br{\mathrm}
\def\PP{\mathbb{P}}
\def\ux{\underline{\xi}}
\def\ua{\underline{a}}
\def\cO{\mathcal{O}}
\def\ul{\underline{\ell}}
\def\vf{\varphi}
\def\cK{\mathcal{K}}
\def\cH{\mathcal{H}}
\def\cL{\mathcal{L}}
\def\cP{\mathcal{P}}
\def\rX{\mathrm{X}}
\def\rY{\mathrm{Y}}
\def\GG{\mathbb{G}}
\def\DDelta{\Delta^{\circ}}
\def\hDelta{\hat{\Delta}}
\def\hDDelta{\hat{\Delta}^{\circ}}
\def\HH{\mathbb{H}}
\def\uv{\underline{v}}
\def\TT{\mathbb{T}}
\def\cT{\mathcal{T}}
\def\cX{\mathcal{X}}
\def\uz{\underline{z}}
\def\rXc{\rX^{\circ}}
\def\uw{\underline{w}}
\def\ukh{\underline{\hat{k}}}
\def\uq{\underline{q}}
\def\VV{\mathbb{V}}
\def\uxi{\underline{\xi}}

\def\MHS{mixed Hodge structure }
\def\LMHS{limiting mixed Hodge structure }
\def\VMHS{variation of mixed Hodge structure }
\def\VHS{variation of Hodge structure }
\def\MPCP{maximal projective crepant partial }
\def\SNCD{strict normal crossing divisors }
\def\NCD{normal crossing divisors }
\def\GW{Gromov-Witten }


\part{The local mirror symmetry}\label{part:Two}

In this part we revisit the approach of \cite{CKYZ} to local mirror symmetry,
by semi-stably degenerating a family of elliptically-fibered Calabi-Yau 3-folds $\mathrm{X}_{z_0,\uz}$ (defined by \eqref{MKeqnI0})
to a singular compactification of the local Hori-Vafa 3-fold
\[ Y_{\uz} := \{ 1-s(\xi_1^2 x + \xi_2^2 y +\xi_3^2)(1+x^{-1}+y^{-1}) + uv = 0\}\subset (\mathbb{C}^*)^2 \times \mathbb{C}^2 \]

\section{B-model}\label{sec:Bmodel}

In this section we describe the
degeneration from a compact Calabi-Yau 3-fold $\rX$ to the local Hori-Vafa
model $\rY$ (which is a noncompact Calabi-Yau 3-fold)~\cite{Hori:2000kt}.
 The main point is that the third
homology of $\rY$ matches the invariant part of the limiting mixed
Hodge structure of $H^{3}(\rX)$ (Theorem~\ref{MKthm1}). Comparing
with the \LMHS of the A-model in the next section will allow us to deduce
a strong form of local mirror symmetry --- equality of variations of
$\QQ$-\MHS --- which implies the conjecture~5.1\footnote{The numbers of section,
conjecture, theorem and equations refer to the published version of~\cite{DoranKerr}.}  from \cite{DoranKerr}, see
Theorem \ref{MKthm2}.

\subsection{Laurent polynomial}

\label{MKsecIA}Choose a reflexive polytope $\Delta\subset\RR^{2}$,
with polar polytope $\DDelta$, and write $r=|\d\Delta\cap\ZZ^{2}|$,
$r^{\circ}=|\d\DDelta\cap\ZZ^{2}|$, and $\nu\,(\leq r,r^{\circ})$
for the common number of edges and vertices of both $\Delta$ and
$\DDelta$. The toric surface associated to $\Delta$ is constructed
from the fan on (the vertices of) $\Delta^{\circ}$, and has canonical
desingularization $\PP_{\Delta}\twoheadrightarrow\check{\PP}_{\Delta}$
arising from the fan on all integer points of $\d\DDelta$. Writing
$\d\Delta\cap\ZZ^{2}=\{\underline{m}^{(j)}\}_{j=1}^{r}$, the general
Laurent polynomial with Newton polytope $\Delta$ is
\[
f_{\underline{a}}(x,y):=a_{0}+\sum_{j=1}^{r}a_{j}x^{m_{1}^{(j)}}y^{m_{2}^{(j)}}
\]
(with $a_{j}\in\CC^{*}$). The compactification of $\{f_{\underline{a}}(x,y)=0\}\subset\GG_{m}^{2}$
in $\PP_{\Delta}$ yields (for general $\{a_{j}\}$) a smooth elliptic
curve $E_{\underline{a}}$.

Jumping up two dimensions, in coordinates $(x_{1},x_{2},x_{3},x_{4})=(x,y,u,v)$
on $\GG_{m}^{4}$, we set \begin{equation}\label{MKeqnI0}F:= \aa + \bb u^2 v^{-1} +\cc u^{-1} v + u^{-1} v^{-1} f_{\underline{a}}(x,y)
\end{equation}(with $\aa,\bb,\cc\in\CC^{*}$). Its Newton polytope
\[
\hDelta:=\Delta(F)=\text{hull}\left\{ (0,0,2,-1),\,(0,0,-1,1),\,\Delta\times(-1,-1)\right\} 
\]
is reflexive since its polar
\[
\hDDelta=\text{hull}\left\{ (0,0,1,0),\,(0,0,0,1),\,6\DDelta\times(-2,-3)\right\} 
\]
is integral. Let $\check{\PP}_{\hDelta}$ be the toric 4-fold associated
to $\hDelta$ (i.e. to the fan on $\hDDelta$), and $\PP_{\hDelta}\twoheadrightarrow\check{\PP}_{\hDelta}$
a maximal projective crepant partial (MPCP) desingularization (arising
from the fan on a maximal triangulation $\text{tr}(\d\hDDelta)$).
For general $\aa,\bb,\cc,\underline{a}$, the compactification 
\[
\mathrm{X}_{F}:=\overline{\{F=0\}}\subset\PP_{\hDelta}
\]
is a smooth Calabi-Yau 3-fold.

To describe the coordinates about the large complex structure limit
in the simplified polynomial moduli space, consider the cone $\cL$
of $\ZZ_{\geq0}$-relations on the $\{\underline{m}^{(j)}\}$. If
it is simplicial and smooth with basis $\{\underline{\ell}^{(i)}\}_{i=1}^{r-2}$,
the coordinates are \begin{equation}\label{MKeqnI1/2}z_0 :=\frac{a_0 \bb^2 \cc^3}{\aa^6}\;\;\text{and}\;\;z_i:=\tfrac{\prod_{j=1}^r a_j ^{\ell_j^{(i)}}}{a_0^{\sum_j \ell_j^{(i)}}} \;(i=1,\ldots,r-2) .
\end{equation}Otherwise, there are more $z_{i}$'s (with relations), though $z_{0}$
remains the same; we will explain how to deal with this complication
for the sunset case at the end of the section.

What follows is a study of the degeneration of $\mathrm{X}_{F}$ as
$z_{0}\to0$.

\subsection{Maximal projective crepant partial desingularization}

\label{MKsecIB}$\PP_{\hDelta}$ is not unique, and for an arbitrary
choice of triangulation $\text{tr}(\d\hDDelta)$ may have isolated
terminal singularities. We shall now describe (and fix) a triangulation
which results in a smooth $\PP_{\hDelta}$.

The integral points on $\hDDelta$ are
\begin{gather*}
(0,0,1,0),\,(0,0,0,1),\,(6\DDelta)_{\ZZ}\times(-2,-3),\,(4\DDelta)_{\ZZ}\times(-1,-2),\\
(3\DDelta)_{\ZZ}\times(-1,-1),\:(2\DDelta)_{\ZZ}\times(0,-1),\,\Delta_{\ZZ}^{\circ}\times(0,0).
\end{gather*}
It has $\nu$ facets of the form
\[
\text{hull}\{6e_{i}^{\circ}\times(-2,-3),\,(0,0,1,0),\,(0,0,0,1)\}=:\mathfrak{f}_{\aa,i}^{\circ}
\]
where $\{e_{i}^{\circ}\}_{i=1}^{\nu}$ are edges of $\DDelta$, and
two of the form
\begin{gather*}
\text{hull}\{6\DDelta\times(-2,-3),\,(0,0,1,0)\}=:\mathfrak{f}_{\bb}^{\circ}\\
\text{hull}\{6\DDelta\times(-2,-3),\,(0,0,0,1)\}=:\mathfrak{f}_{\cc}^{\circ}.
\end{gather*}
The decomposition of these facets into elementary tetrahedra proceeds
in four steps:

\subsubsection*{Step 1:}

For each $\ZZ$-point $\underline{w}\in(\d\DDelta)_{\ZZ}$, draw the
half-space
\[
\HH_{\underline{w}}:=\overrightarrow{0.\underline{w}}\times\CC_{u,v}^{2}
\]
through it. This subdivides the facets.

\subsubsection*{Step 2:}

Up to unimodular transformation, the resulting ``slices'' of $\mathfrak{f}_{\bb}^{\circ}$
resp. $\mathfrak{f}_{\cc}^{\circ}$ are 
\[
\text{hull}\{(0,0,-2,-3),\,(6,0,-2,-3),\,(0,6,-2,-3),\,\small\begin{array}{c}
(0,0,1,0)\\
\text{resp. }(0,0,0,1)
\end{array}\}.
\]
To triangulate the first one (second is similar): first decompose
it into
\[
\mathfrak{f}_{\mathfrak{A}}=\text{hull}\{\begin{array}[t]{c}
\underbrace{(3,0,-1,-1),\,(0,3,-1,-1),\,(0,0,-1,-1),}\,(0,0,0,1)\\
\text{hull}=:\mathfrak{p_{A}}
\end{array}\}
\]
and
\[
\mathfrak{f_{B}}=\text{hull}\{\mathfrak{p_{A}},\begin{array}[t]{c}
\underbrace{(6,0,-2,-3),\,(0,6,-2,-3),\,(0,0,-2,-3)}\\
\text{hull}=:\mathfrak{p_{B}}
\end{array}\};
\]
on $\mathfrak{p_{A}}$ and $\mathfrak{p_{B}}$, draw all the integral
horizontal, vertical, and anti-diagonal ($x+y=k$) lines; then for
$\mathfrak{f_{A}}$, complete the resulting triangles to tetrahedra
with vertex at $(0,0,0,1)$; for $\mathfrak{f_{B}}$, further subdivide
into the 4 tetrahedra
\begin{gather*}
\text{hull}\{(0,0,-1,-1),\,(0,0,-2,-3),\,(3,-,-2,-3),\,(0,3,-2,-3)\}\\
\text{hull}\{(0,3,-1,-1),\,(0,3,-2,-3),\,(3,3,-2,-3),\,(0,6,-2,-3)\}\\
\text{hull}\{(3,0,-1,-1),\,(3,0,-2,-3),\,(6,0,-2,-3),\,(3,3,-2,-3)\}\\
\text{hull}\{(0,0,-1,-1),\,(3,0,-1,-1),\,(0,3,-1,-1),\,(3,3,-2,-3)\}
\end{gather*}
(treating these as with $\mathfrak{f_{A}}$) and the 2 ``skew''
tetrahedra
\begin{gather*}
\text{hull}\{(0,0,-1,-1),\,(0,3,-1,-1),\,(0,3,-2,-3),\,(3,3,-2,-3)\}\\
\text{hull}\{(0,0,-1,-1),\,(3,0,-1,-1),\,(3,0,-2,-3),\,(3,3,-2,-3)\}
\end{gather*}
(which get subdivided into elementary tetrahedra of the form
\[
\text{hull}\{(a,0,-1,-1),\,(a+1,0,-1,-1),\,(3,b,-2,-3),\,(3,b+1,-2,-3)\}).
\]

\subsubsection*{Step 3:}

For the $\mathfrak{f}_{\aa,i}^{\circ}$, it will not matter which
triangulation we choose. Two of the 2-faces of $\mathfrak{f}_{\aa,i}^{\circ}$
already receive a triangulation from Step 2. The other 2 may be star-triangulated
with centers of the form $\underline{v}\times(0,0)$, $\underline{v}\in\DDelta_{\ZZ}$.
Any 3-triangulation completing this will do.

\subsubsection*{Step 4:}

One checks that all of the tetrahedra in this triangulation are regular,
i.e. the determinants of their vertices are $\pm1$. This is not always
possible for a general 4-dimensional reflexive polytope, and shows
that $\PP_{\hDelta}$ \emph{is smooth}.

\subsection{Elliptic fibration}

\label{MKsecIC}Write $\Sigma_{\Delta}$ (resp. $\Sigma_{\hDelta}$)
for the fan on $\d\DDelta_{\ZZ}$ (resp. on the triangulation $\text{tr}(\d\hDDelta)$).
By Step 1 in \S\ref{MKsecIB}, we have a map of fans $\Sigma_{\hDelta}\to\Sigma_{\Delta}$
hence a diagram\[\xymatrix{\check{\PP}_{\hDelta} & \ar@{->>}[l] \PP_{\hDelta} \ar@{->>}[d]^{\mathcal{P}} & (x,y,u,v) \ar@{|->}[d]^{\text{(generically)}} \\ & \PP_{\Delta} & (x,y) }\]with
$\PP$ a morphism. The components of $\DD_{\hDelta}:=\PP_{\hDelta}\backslash\GG_{m}^{4}$
not lying over a component of $\DD_{\Delta}=\PP_{\Delta}\backslash\GG_{m}^{2}$,
are the ones dual to $\ZZ$-points of $\hDDelta$ with first two coordinates
$0$ (except the origin). These are $(0,0)\times$
\[
(1,0),\,(0,1),\,(-2,-3),\,(-1,-2),\,(-1,-1),\,(0,-1),
\]
the fan on which produces the toric surface $\PP_{W}$, which is a
desingularization of $\mathbb{WP}(1,2,3)$. So the ``generic'' fiber
of $\mathcal{P}$ is $\PP_{W}$, the ``correct'' toric surface in
which to compactify the (generalized) Weierstrass elliptic curve $F(x_{0},y_{0},u,v)=0$
(where $x_{0},y_{0}\in\CC^{*}$).

We now describe the induced elliptic fibration
\[
\mathrm{X}_{F}\overset{\rho}{\to}\PP_{\Delta}.
\]
Write \begin{equation}\label{MKeqnI0!}E_{\ua}^* = \{ f_{\ua}(x,y)=0\}\cap\GG_m^2 ,
\end{equation}with compactification $E_{\ua}\subset\PP_{\Delta}$; and \begin{equation}\label{MKeqnI0!!}D_{\ua,z_0}^* = \{ f_{\ua}(x,y)=\tfrac{a_0}{2\cdot 6^3 z_0} \} \cap \GG_m^2 ,
\end{equation}with compactification $D_{\ua,z_{0}}\subset\PP_{\Delta}$. We have
$E_{\ua}\cap D_{\ua,z_{0}}=E_{\ua}\cap\DD_{\Delta}=D_{\ua,z_{0}}\cap\DD_{\Delta}=:\BB_{\Delta,\ua}$
(which consists of $r$ points) for the base locus of the pencil $f_{\ua}(x,y)=\lambda$,
$\lambda\in\PP^{1}(\CC)$. The 1-dimensional fibers of $\rho$ are:
\begin{itemize}
\item over $\GG_{m}^{2}\backslash\{D_{\ua,z_{0}}^{*}\cup E_{\ua}^{*}\}$,
a smooth elliptic curve ``$\cE$'';
\item over $E_{\ua}^{*}$, type $I_{1}$ (nodal rational curve) with node
at $(u,v)=(0,0)$;
\item over $D_{\ua,z_{0}}^{*}$, type $I_{1}$ with node at $(u,v)=\left(\tfrac{-\aa^{3}}{12\bb\cc^{2}},\tfrac{\aa^{2}}{6\bb\cc}\right)$;
and
\item over $\DD_{\Delta}\backslash\{\text{sing}(\DD_{\Delta})\cup\BB_{\Delta}\}$,
type $II^{*}$ ($E_{8}$ configuration).
\end{itemize}
Indeed, the local system on $\GG_{m}^{2}\backslash\{E_{\ua}^{*}\cup D_{\ua,z_{0}}^{*}\}$
is the pullback (by $\lambda=f_{\ua}(x,y)\cdot2\cdot6^{3}\bb^{2}\cc^{3}/\aa^{6}$)
of the fibered wise first homology group $H_{1}$ of the family
\[
\frac{\lambda\aa^{6}}{2\cdot6^{3}\bb^{2}\cc^{3}}+\aa uv+\bb u^{3}+\cc v^{2}=0
\]
in $\PP_{W}$, with singular fibers at $0$ of type $I_{1}$, $1$ of type
$I_{1}$, $\infty$ of type $II^{*}$. On $\CC\backslash\{(-\infty,0]\cup[1,\infty)\}$,
we have a basis $\{\alpha,\beta\}$ of 1-cycles (for the local system)
with monodromies $\tiny\left(\begin{array}{cc}
1 & 1\\
0 & 1
\end{array}\right)$ {[}resp. $\tiny\left(\begin{array}{cc}
1 & 0\\
-1 & 1
\end{array}\right)$, $\tiny\left(\begin{array}{cc}
0 & -1\\
1 & 1
\end{array}\right)${]} about $0$ {[}resp. $1$, $\infty${]}; accordingly, we shall
write $\alpha$ resp. $\beta$ for the vanishing cycles of the pullback
local system at $E_{\ua}^{*}$ resp. $D_{\ua,z_{0}}^{*}$.

Over $\BB_{\Delta,\ua}$ and $\mathbb{S}_{\Delta}:=\text{sing}(\DD_{\Delta})$,
the fibers of $\rho$ have dimension 2. The components are obtained
by taking preimages of solutions to edge- and 2-face-polynomials of
$F$ under the blowups used to produce $\PP_{\hDelta}$ from $\check{\PP}_{\hDelta}$.
These preimages are hypersurfaces in components of $\DD_{\hDelta}$
corresponding to integer points of $\hDDelta$ (not in its interior
or that of its facets) lying in one of the open half-spaces $\HH_{\underline{w}}$
(for $\BB_{\Delta,\ua}$) or in between two of them (for $\mathbb{S}_{\Delta}$).

For $\BB_{\Delta,\ua}$, the sole contribution comes from the points
of the form $\uv\times(0,0)$ where $\uv$ is a vertex of $\DDelta$
(dual to an edge $e_{\uv}$ of $\Delta$). These belong to the interior
of a 2-face $\text{hull}\{6\uv\times(-2,-3),\,(0,0,1,0),\,(0,0,0,1)\}$
which is dual to the edge $e_{\uv}\times(-1,-1)$ of $\hDelta$. When
the corresponding 1-dimensional subspace of $\check{\PP}_{\hDelta}$
is blown up to a 3-dimensional one in $\PP_{\hDelta}$, the equation
is inherited from the edge polynomial of $F$, which cuts out a point
(or points) of $\BB_{\Delta,\ua}$. Since this blowup arises from
the star subdivision of Step 3 from \S\ref{MKsecIB}, we conclude
that the fiber over said point is a copy of $\PP_{W}$. 

For $\mathbb{S}_{\Delta}$, there are contributions from all the points
of the form 
\begin{gather*}
(2\d\DDelta)_{\ZZ}\backslash2\d\DDelta_{\ZZ}\times(0,-1),\qquad (2\d\DDelta)_{\ZZ}\backslash3\d\DDelta_{\ZZ}\times(-1,-1),\\
(2\d\DDelta)_{\ZZ}\backslash4\d\DDelta_{\ZZ}\times(-1,-2),\qquad
(6\DDelta)_{\ZZ}\backslash6\DDelta_{\ZZ}\times(-2,-3).
\end{gather*}	
Unimodular
transformation maps any edge of $\text{tr}(\d\DDelta)$ to $[(0,1),(1,0)]$,
hence for a given point $p\in\mathbb{S}_{\Delta}$ (dual to that edge),
one easily sees that $\rho^{-1}(p)$ consists of $21$ rational surfaces.
\begin{rem}
The fibration $\rho$ has an obvious section, given by the intersection
of $\mathrm{X}_{F}$ with the component of $\DD_{\hDelta}$ indexed
by the point $(0,0,-2,-3)\in\hDDelta_{\ZZ}$. In a generic fiber this
is the usual ``point at $\infty$'' in the Weierstrass elliptic
curve.
\end{rem}

\subsection{Middle homology of $\mathrm{X}_{F}$}

\label{MKsecID}We shall work henceforth under the assumption that
$a_{i>0}$ and $z_{0}$ are sufficiently small. Write $\ell(\theta)$
resp. $\ell^{*}(\theta)$ for the number of integral points in a polytope
$\theta$ resp. its interior. We have
\[
h^{3,0}(\mathrm{X}_{F})=\ell^{*}(\hDelta)=1
\]
since $\hDelta$ is reflexive, with the (unique up to scale) holomorphic
3-form given by
\[
\Omega_{F}=\frac{1}{(2\pi\ay)^{3}}\text{Res}_{\mathrm{X}_{F}}\left(\frac{dx/x\wedge dy/y\wedge du/u\wedge dv/v}{F}\right)\in\Omega^{3}(\mathrm{X}_{F}).
\]
We have also the Batyrev formula~\cite{Batyrev:1994hm} \begin{align*}h^{2,1}(\mathrm{X}_F) & = \ell(\hDelta) - \sum_{\tiny \begin{array}{c}\sigma\text{ facet}\\\text{of }\hDelta\end{array}}\ell^*(\sigma)+\sum_{\tiny \begin{array}{c}\theta\text{ 2-face}\\\text{of }\hDelta\end{array}}\ell^*(\theta)\ell^*(\theta^{\circ}) -5 \\
&= \{ \ell(\Delta)+6\}- {1+2}+0-5 \\
&= \ell(\Delta)-2 \\
&=r-1.
\end{align*}Now we shall use the structure of the elliptic fibration to exhibit
a \emph{basis} of $H_{3}(\mathrm{X}_{F},\QQ)$. (In what follows we
often drop subscripts $\ua$, $z_{0}$, $F$, etc.; moreover, some
steps are only sketched).

The basic observation is that
\[
\cK_{E}:=\ker\{H_{1}(E^{*})\to H_{1}(\GG_{m}^{2})\}=\ker\{H_{1}(E^{*})\to H_{1}(\GG_{m}^{2}\backslash D^{*})\}
\]
and
\[
\cK_{D}:=\ker\{H_{1}(D^{*})\to H_{1}(\GG_{m}^{2})\}=\ker\{H_{1}(D^{*})\to H_{1}(\GG_{m}^{2}\backslash E^{*})\}
\]
are $(r-1)$-dimensional spaces. (For instance, to see the second
equality for $\cK_{E}$, take $z_{0}$ small enough that $D$ lies
inside an $\epsilon$-neighborhood $U$ of $\DD_{\Delta}$, and replace
$E^{*}$ by $E\backslash\{E\cap U\}$.) Moreover, we have two obvious
3-cycles $\mathcal{T}_{\alpha}$ and $\mathcal{T}_{\beta}$ consisting
of parallel translates of $\alpha$ resp. $\beta$ over $\TT:=\{|x|=|y|=1\}$.
We will show that, together with these, certain cycles built from
$\cK_{E}$ and $\cK_{D}$ yield $2r$ independent 3-cycles on $\mathrm{X}_{F}$.

To construct these cycles, let $\{\vf\}=\left\{ \{\vf_{0}^{(i)}\}_{i=1}^{r-2},\vf_{1}\right\} \subset\cK_{E}$
be a basis (with $\{\vf_{0}^{(i)}\}$ all being homologous to one
$\vf_{0}$ on $E$). Choose for each $\varphi$ a 2-chain $\Gamma_{\vf}\subset\GG_{m}^{2}\backslash D^{*}$
with $\d\Gamma_{\vf}=\vf$; and let $\cM_{\alpha}(\vf)$ be a continuous
family of 1-cycles of class $[\alpha]$ over $\Gamma_{\vf}$ collapsing
to a point over $p$. We take $\Phi_{E}\subset H_{3}(\pi^{-1}(\GG_{m}^{2}\backslash D^{*}))$
to be the span of the $\{[\cM_{\alpha}(\vf)]\}$, and similarly $\Phi_{D}=\text{span}\{[\cM_{\beta}(\vf)]\}_{\vf\in\cK_{D}}\subset H_{3}(\rho^{-1}(\GG_{m}^{2}\backslash E^{*}))$. 

We may compute $H_{3}^{D}:=H_{3}(\rho^{-1}(\GG_{m}^{2}\backslash D^{*}))$
via the relative homology sequence 
\begin{align}\label{MKeqnI1}
	\cdots\to H_3(\rho^{-1}E^*)\overset{\psi_E}{\to}H_3^D&\to H_3\left(\rho^{-1}(\GG_m^2\setminus D^*),\pi^{-1}E^*\right)\\
	\notag &\overset{\theta_E}{\to} H_2(\rho^{-1}E^*)\to\cdots ,
\end{align}
in which
\[
\text{im}(\psi_{E})\cong\text{im}\left\{ H_{1}(E^{*})\to H_{1}(\GG_{m}^{2}\backslash D^{*})\right\} \otimes[\cE]\underset{n.c.}{\cong}H_{1}(\GG_{m}^{2}).
\]
The second isomorphism is not canonical.
Writing $\cH_{1}$ for $(R^{1}\rho_{*}\QQ)^{\vee}$, the Leray spectral
sequence yields 
\begin{align}\label{MKeqnI2}
0&\to H_1(\GG_m^2\setminus D^*,E^*)\otimes[\cE]\\
\notag &\to \ker(\theta_E)\to  H_2(\GG_m^2\setminus D^*,E^*;\cH_1 )\overset{\theta_E '}{\to} H_1(E^*,\cH_1/\langle\alpha\rangle )
\end{align}
so that \eqref{MKeqnI1} becomes \begin{equation}\label{MKeqnI3}0\to H_1(\GG_m^2\setminus D^*)\otimes [\cE] \to H_3^D\to \ker(\theta_E ')\to 0 .
\end{equation}
Using the exact sequences
\begin{gather*}
0\to H_{2}(\GG_{m}\backslash D^{*})\to H_{2}(\GG_{m}^{2}\backslash D^{*},E^{*})\to\cK_{E}\to0\\
0\to H_{1}(D^{*})\overset{\text{Tube}}{\text{\ensuremath{\to}}}H_{2}(\GG_{m}^{2}\backslash D^{*})\to H_{2}(\GG_{m}^{2})\to0,
\end{gather*}
and writing $\mathsf{T}:=\QQ\langle[\cT_{\alpha}],[\cT_{\beta}]\rangle\subset H_{3}^{D}$
and $\Psi_{D}:=\text{Tube}(H_{1}(D^{*}))\otimes[\beta]$, one can
then show (with some work) that $\mathsf{T}\oplus\Psi_{D}\oplus\Phi_{E}$
maps isomorphically to $\ker(\theta_{E}')$. Repeat this whole argument
with $D$ and $E$ (and $\alpha$ and $\beta$) swapped to compute
$H_{3}^{E}$.

Next, one explicitly checks that $H_{2}(\rho^{-1}(\GG_{m}^{2}\backslash\{D^{*}\cup E^{*}\}))=:H_{2}^{DE}$
injects into $H_{2}^{D}\oplus H_{2}^{E}$, so that \begin{equation}\label{MKeqnI4}H_3 (\rho^{-1}\GG_m^2 )\cong\frac{H^D_3\oplus H_3^E }{\text{im}\{H_3^{DE}\} } .
\end{equation}
The Leray spectral sequence for $\rho$ yields 
\begin{align}\label{MKeqnI5}
0&\to H_1 (\GG_m^2\setminus \{D^* \cup E^*\} )\otimes [\cE]\\
\notag &\to H_3^{DE} \to H_2(\GG_m^2\setminus \{ D^*\cup E^*\} ;\cH_1 )\to 0 ,
\end{align}
which (comparing with \eqref{MKeqnI3}) breaks the computation of
the quotient \eqref{MKeqnI4} into two pieces: for the ``left-hand''
piece, we have
\[
\frac{H_{1}(\GG_{m}^{2}\backslash D^{*})\oplus H_{1}(\GG_{m}^{2}\backslash E^{*})}{\text{im}\{H_{1}(\GG_{m}^{2}\backslash\{D^{*}\cup E^{*}\})\}}\otimes[\cE]\cong H_{1}(\GG_{m}^{2})\otimes[\cE].
\]
For the ``right-hand'' piece, the quotient of $\ker(\theta_{E}')$
by the right-hand term of \eqref{MKeqnI5}, which is an extension
\[
0\to\Psi_{D}\oplus\Psi_{E}\to H_{2}(\GG_{m}^{2}\backslash\{D^{*}\cup E^{*}\};\cH_{1})\to\mathsf{T}\to0,
\]
is evidently isomorphic to $\Phi_{E}\oplus\Phi_{D}\oplus\mathsf{T}$.

Finally, we consider the cohomology of the normal crossing divisor\linebreak
$\rho^{-1}\DD_{\Delta}=\cup\cR_{i}$; here the $\cR_{i}$ are rational
surfaces (meeting along $\PP^{1}$'s) indexed by $\ZZ$-points of
$\hDDelta$ with $(x,y)\neq(0,0)$ and not in the interior of a facet.
By studying the part of the 2-skeleton of $\text{tr}(\d\hDDelta)$
meeting these points, we compute the spectral sequence converging
to $H^{*}(\rho^{-1}\DD_{\Delta})$, with $E_{1}$ page
\[
\left|\underline{\begin{array}{ccccc}
\oplus H^{4}(\cR_{i})\\
0\\
\oplus H^{2}(\cR_{i}) & \overset{\varepsilon}{\to} & \oplus H^{2}(\cR_{ij})\\
0 &  & 0\\
\oplus H^{0}(\cR_{i}) & \to & \oplus H^{0}(\cR_{ij}) & \to & \oplus H^{0}(\cR_{ijk}).
\end{array}}\right.
\]
The \emph{cohomology} ranks of the bottom row are just the Betti numbers
$1,1,0$ of the $2$-skeleton, so that $H^{1}(\DD)\to H^{1}(\rho^{-1}\DD)\overset{\text{Tube}}{\to}H_{4}(\mathrm{X}\backslash\rho^{-1}\DD_{\Delta})$
are isomorphisms (all of rank 1, with Tube hitting $\TT\otimes[\cE]$).
One also deduces from this that $\ker\{\oplus H_{2}(\cR_{ij})\to\oplus H_{4}(\cR_{i})\}$,
hence $\text{coker}(\varepsilon)$ and $H^{3}(\rho^{-1}\DD_{\Delta})$,
have rank 1; it follows that $H^{3}(\rho^{-1}\DD_{\Delta})\overset{\text{Tube}}{\to}H_{2}(\mathrm{X}\backslash\rho^{-1}\DD_{\Delta})$
is injective, with Tube hitting $\TT\otimes[\text{pt.}]$. Moreover,
we have
\[
\begin{cases}
\text{rank}(\oplus H^{4}(\cR_{i}))  =  30r^{\circ}+\nu\\
\text{rank}(\oplus H^{2}(\cR_{i}))  =  120r^{\circ}+5\nu\\
\text{rank}(\oplus H^{2}(\cR_{ij}))  =  90r^{\circ}+4\nu
\end{cases}
\]
and so conclude that $h^{2}(\pi^{-1}\DD_{\Delta})=120r^{\circ}+5\nu-(90r^{\circ}+4\nu)+1=30r^{\circ}+\nu+1.$
Now by Batyrev~\cite{Batyrev:1994hm} \begin{align*}h^4(\mathrm{X}) &= h^2(\mathrm{X}) = \ell(\hDDelta)-\sum_{\tiny \begin{array}{c}\sigma^{\circ}\text{ facet}\\\text{of }\hDDelta\end{array}} \ell^*(\sigma^{\circ}) + \sum_{\tiny \begin{array}{c}\theta^{\circ}\text{ 2-face}\\\text{of }\hDDelta\end{array}} \ell^*(\theta^{\circ})\ell^*(\theta) - 5 \\
&= (31 r^{\circ} + 7) - (r^{\circ}-\nu +3)-5 \\
&= 30r^{\circ} + \nu -1.
\end{align*}By the exact sequence
\[
0\to H_{4}(\mathrm{X})\to H^{2}(\rho^{-1}\DD_{\Delta})\overset{\text{Tube}}{\to}H_{3}(X\backslash\rho^{-1}\DD_{\Delta})\to H_{3}(\mathrm{X})\to0
\]
we now have $\text{rank}(\text{Tube})=h^{2}(\rho^{-1}\DD_{\Delta})-h^{4}(\mathrm{X})=2$.
Since $H_{1}(\GG_{m}^{2})\otimes[\cE]$ is evidently \emph{in} the
image of Tube, this \emph{is} $\text{im}(\text{Tube})$ and thus \begin{equation}\label{MKeqnI6}H_3(\mathrm{X})=\Phi_D \oplus \Phi_E \oplus \mathsf{T} .
\end{equation}

\subsection{Degeneration as $z_{0}\to0$}

\label{MKsecIE}We shall need to replace $z_{0}$ by $t$, which amounts
to pullback by $t\mapsto t^{5}$. Set $\bb=\cc=t$, $\aa=1$ in \eqref{MKeqnI0}.
Write $\mathrm{X}_{\ua,t}$ (or just $\mathrm{X}_{t}$) for the corresponding
Calabi-Yau 3-fold; fix an $\ua$ and disk $\mathbbold{\Delta}$ about
$0$, such that $\cX_{\ua}\to\mathbbold{\Delta}$ (with fibers $\mathrm{X}_{\ua,t}$)
is smooth away from $\{0\}$. For the singular fibers write
\[
\mathrm{X}_{\ua,0}=\cup_{i\geq0}W_{i},
\]
where
\[
W_{0}=\overline{Y}:=\overline{\left\{ f_{\ua}(x,y)+uv=0\right\} }\subset\PP_{\hDelta}
\]
and $W_{i>0}$ are the components of $\DD_{\hDelta}$ corresponding
to integer points of $\hDDelta$ contained in the interiors of the
2-face $6\DDelta\times(-2,-3)$ and of the facets $\mathfrak{f}_{\bb}^{\circ}$
and $\mathfrak{f}_{\cc}^{\circ}$. Write $I_{\DDelta}$ for the index
(sub)set corresponding \emph{just} to the interior points of $6\DDelta\times(-2,-3)$.

The singular locus of the total space is contained in $\mathrm{X}_{\ua,0}$;
more precisely, it consists of the intersections $W_{i}\cap W_{j}\cap\mathrm{X}_{\ua,t\neq0}\cong\PP^{1}$
with $i,j\in I_{\DDelta}$. Let $\PP_{\hDelta}'\overset{B}{\twoheadrightarrow}\PP_{\hDelta}$
denote the blow-up along the smooth rational surfaces $\{W_{i}\cap\mathrm{X}_{\ua,t\neq0}\}_{i\in I_{\DDelta}}$,
in any order, and $\cX'\twoheadrightarrow\cX$ the proper transform
under $B\times\text{id}_{\mathbbold{\Delta}}$. Note that $\cX'$
is smooth, with fibers over $\mathbbold{\Delta}^{*}$ unchanged, and
$\mathrm{X}_{\ua,0}'=\cup W_{i}'$ having no \emph{additional} irreducible
components. Indeed, the only change is that some irreducible components
of $\mathrm{X}_{0}$ have been blown up along some $\PP^{1}$'s. Write
$\DD_{\hDelta}'$, $\overline{Y}'$, etc. for proper transforms.

Furthermore, $\mathrm{X}_{0}$ and $\mathrm{X}_{0}'$ are smooth normal
crossing divisors in $\PP_{\hDelta}$ and $\PP_{\hDelta}'$ respectively, and
$\mathrm{X}_{0}'$ is a reduced \SNCD in $\cX'$ --- i.e. $\cX'\to\mathbbold{\Delta}$
is a semistable degeneration. Because this is a ``partial'' toric
degeneration (i.e. $\mathrm{X}_{0}\neq\DD_{\hDelta}$), these facts
are not automatic. The \NCD property is checked by computations in
local coordinate systems associated to the individual tetrahedra in
$\text{tr}(\mathfrak{f}_{\bb}^{\circ})$ and $\text{tr}(\mathfrak{f}_{\cc}^{\circ})$;
it holds for the triangulation described in \S\ref{MKsecIB}, but
not for another triangulation we considered. Also visible in these
local coordinates is the fact that the equation of $\cX'$ takes the
form $t=\prod_{i=1}^{M}f_{i}(\underline{X})$ ($M\leq4$) where the
$\vec{\nabla}f_{i}(\underline{X})$ are independent along intersections
of the respective components. (Since these computations are both lengthy
and straightforward, we omit them.) 

We describe the degenerated elliptic fibration $\mathrm{X}_{0}\overset{\rho}{\to}\PP_{\Delta}$,
noting that as $t\to0$ ($z_{0}\to0$), $D_{\ua,z_{0}}\to D_{\ua,0}=\DD_{\Delta}$.
Over $\PP_{\Delta}$, there are 5 components (including $\overline{Y}$),
forming an $I_{5}$ over $\GG_{m}^{2}\backslash E^{*}$ and an $I_{6}$
over $E^{*}$. Over each $\PP^{1}\subset\DD_{\Delta}$, $\mathrm{X}_{0}$
has 11 components; while over each point in $\mathbb{S}_{\Delta}$,
there are 14. $\overline{Y}$ itself is generically a $\PP^{1}$-bundle,
whose fiber breaks into 2 $\PP^{1}$'s (joined at $(u,v)=(0,0)$)
over $E^{*}$ and 8 $\PP^{1}$'s over $\DD_{\Delta}\backslash\{\mathbb{S}_{\Delta}\cup\BB_{\Delta}\}$,
while $\rho^{-1}(\BB_{\Delta})\cap\overline{Y}$ is a configuration
of 5 rational surfaces (in addtion to the 11 which lie over every
point of $\DD_{\Delta}\backslash\mathbb{S}_{\Delta}$). This description
does not change for $\mathrm{X}_{0}'$.

\subsection{Middle cohomology of $\mathrm{X}_{0}$}

\label{MKsecIF}By Clemens-Schmid~\cite{Clemens,Schmid} we have an exact sequence of \MHS
\[
H_{5}(\mathrm{X}_{0}')(-4)\to H^{3}(\mathrm{X}_{0}')\to H_{lim}^{3}(\mathrm{X}_{t})\overset{N}{\to}H_{lim}^{3}(\mathrm{X}_{t})(-1)
\]
since $\cX'\to\mathbbold{\Delta}$ is a semistable degeneration. Using
the combinatorics of $\text{tr}(\hDDelta)$ one shows that $\oplus H^{4}(W_{i}')\twoheadrightarrow\oplus H^{4}(W_{ij}')$,
and clearly $\oplus H^{5}(W_{i}')=H^{5}(\overline{Y}')=H_{1}(\overline{Y}')(-3)$
which we shall show is zero (see below). So $H^{5}(\mathrm{X}_{0}')$,
hence $H_{5}(\mathrm{X}_{0}')$, is zero; writing $H_{inv}^{3}(\mathrm{X}_{t})$
for the \MHS $\ker(T-I)=\ker(N)\subset H_{lim}^{3}(\mathrm{X}_{t})$,
we have 
\[
H^{3}(\mathrm{X}_{0}')\cong H_{inv}^{3}(\mathrm{X}_{t}).
\]

\begin{rem}
The only possible discrepancy between $H^{3}(\mathrm{X}_{0}')$ and
$H^{3}(\mathrm{X}_{0})$ is in $Gr_{2}^{W}$, for which we have the
diagram
	
\[\xymatrix{\oplus H^2(W_i) \ar @{^(->} [d]^{B^*} \ar [r]^{\delta} & \oplus H^2(W_{ij}) \ar @{^(->} [d]^{B^*} \ar [r]^{\delta} & \oplus H^2(W_{ijk}) \ar @{=} [d] \\ \oplus H^2(W_i ') \ar [r]^{\delta} & \oplus H^2(W_{ij} ') \ar [r]^{\delta} & \oplus H^2(W_{ijk} ') . }\]Any
topological cycle $Z$ with $[Z]\in\ker(\delta)\subset\oplus H^{2}(W_{ij})$
can be moved to avoid the blowup points. As a consequence, if $B^{-1}(Z)=\delta\mathfrak{Z}$
then $Z=\delta(B(\mathfrak{Z}))$, showing $H^{3}(\mathrm{X}_{0})\hookrightarrow H^{3}(\mathrm{X}_{0}')$.
\end{rem}
Recalling that $Y_{\ua}'=W_{0}'$, set
\[
Y_{\ua}:=\overline{Y}_{\ua}\cap(\GG_{m}^{2}\times\mathbb{A}^{2})=\overline{Y}_{\ua}'\cap(\GG_{m}^{2}\times\mathbb{A}^{2})\subset\overline{Y}_{\ua}'\backslash\{\overline{Y}_{\ua}'\cap(\oplus_{i\geq1}W_{i}')\},
\]
the solution set of $f_{\ua}(x,y)+uv=0$ with $x,y\in\CC^{*}$ and
$u,v\in\CC$. It is a $\GG_{m}$-bundle over $\GG_{m}^{2}\backslash E_{\ua}$
which degenerates to two affine lines meeting at $(u,v)=\underline{0}$
over $E_{\ua}$. As in \cite[\S 5.1]{DoranKerr} we have exact sequences
of \MHS 
\begin{align}\label{MKeqnI7}
	H_{k-1}(E_{\ua}^*)(1) &\overset{(I_*,0)}{\to} H_{k-1}(\GG_m^2)(1)\oplus H_k(\GG_m^2)\\
	\notag &\to H_k(Y_{\ua})  \to H_{k-2}(E_{\ua}^*)(1)\to H_{k-2}(\GG_m^2)(1) .
\end{align}
Setting $k=1$ gives $H_{1}(Y)=H_{1}(\GG_{m}^{2})\otimes[\text{pt.}]$,
which evidently maps to $0$ in $H_{1}(\overline{Y}')$, so that $H_{1}(\overline{Y}')=\{0\}$.

For $k=3$, \eqref{MKeqnI7} becomes \begin{equation}\label{MKeqnI8}0\to H_2 (\GG_m^2)(1)\overset{\otimes [S^1]}{\to} H_3(Y_{\ua}) \overset{\xi}{\to} \cK_E(1) \to 0 .
\end{equation}The cycles $\{\cM_{\alpha}(\vf)\}_{\vf\in\cK_{E}}$ and $\cT_{\alpha}$
evidently limit (with $t\to0$) to cycles $\{\cM(\vf)\}_{\vf\in\cK_{E}}$
and $\cT$ on $Y_{\ua}$, with the $S^{1}$ on the $\GG_{m}$-fibers
replacing $\alpha$. Clearly $[\cT]=\text{im}(\otimes[S^{1}])$, and
$\text{span}\{[\cM(\vf)]\}$ maps isomorphically to $\cK_{E}(1)$
(cf. the construction of the right-inverse%
\footnote{as morphism of $\QQ$-vector spaces, \emph{not} \MHS%
} ``$M$'' to $\xi$ in\cite[\S 5.1]{DoranKerr} 
So as a $\QQ$-vector space,
$H_{3}(Y_{\ua})\cong\Phi_{E}\oplus\langle\cT\rangle$; as a \MHS\unskip, $H_{3}(Y_{\ua})(-3)$
is an extension of $\cK_{E}(-2)$ (which has type $(2,1)+(1,2)+(1,1)$)
by a $\QQ(0)$ (spanned by $\cT$).

Now consider the composite morphism 
\begin{align*}
	\Theta':\, H_{3}(Y_{\ua})(-3) &\cong H^{3}(\overline{Y}',\overline{Y}'\backslash Y)\to H^{3}(W_{0}',\cup W_{0i}') \\ 
	&=H^{3}(\mathrm{X}_{0}',\cup_{i\geq1}W_{i}')\to H^{3}(\mathrm{X}_{0}')
\end{align*}of \MHS (one defines $\Theta$ similarly). On the level of closed chains,
$\Theta'$ is induced by \begin{align*}
\tilde{\Theta} ' :\, Z^{top}_3(Y) & \rightarrow \ker\left\{ Z_3^{top}(\overline{Y}')_{\#} \to \oplus_i Z_i^{top}(W_{0i} ')\right\} \\
& \hookrightarrow \ker \left\{ \delta:\, \oplus_i Z_{top}^3 (W_i ')_{\#} \to \oplus_{i,j} Z^3_{top} (W_{ij} ') \right\} \\
& \rightarrow \ker \left\{ \DD :\, \oplus_I Z_{top}^{3-|I|} (W_I ')_{\#} \to \oplus_J Z^{4-|J|}_{top} (W_J ') \right\}
\end{align*}where $\#$ denotes intersection conditions and $\DD$ is the total
differential for the complex computing $H^{*}(\mathrm{X}_{0}')$.
The main point here is that the Clemens retraction map $\mathfrak{r}:H^{3}(\mathrm{X}_{0}')\to H^{3}(\mathrm{X}_{t})$
is given (on $\ker(\delta)$, hence $\text{im}(\tilde{\Theta}')$)
by simple preimage under $\mathfrak{r}:\mathrm{X}_{t}\twoheadrightarrow\mathrm{X}_{0}'$.
Since this obviously sends $\cT\mapsto\cT_{\alpha}$ and $\cM(\vf)\mapsto\cM_{\alpha}(\vf)$,
the composite MHS morphism \[\xymatrix{H_3(Y_{\ua})(-3) \ar @/_1pc/ @{-->} [rr]_{\Theta '} \ar [r]^{\Theta} & H^3(\mathrm{X}_0) \ar @{^(->} [r]^{B^*} & H^3(\mathrm{X}_0 ') \ar @/_1pc/ @{-->} [rr]_{\mathfrak{r}^*} \ar [r]^{\cong} & H^3_{inv}(\mathrm{X}_t) \ar @{^(->} [r] & H^3_{lim}(\mathrm{X}_t) }\]sends
$[\cT]\mapsto[\cT_{\alpha}]$ and $[\cM(\vf)]\mapsto[\cM_{\alpha}(\vf)]$.
Since these classes remain independent in $H^{3}(\mathrm{X}_{t})$,%
\footnote{either by the computation of the basis in \S\ref{MKsecID}, or by
the Remark below%
} $\Theta$ and $\Theta'$ are injective. Consequently $H_{lim}^{3}(\mathrm{X}_{t})$
has $I^{0,0}\supseteq I^{0,0}(H_{3}(Y)(-3))$ of rank at least $1$,
$I^{1,1}\supseteq N^{+}(I^{0,0})\oplus I^{1,1}(H_{3}(Y)(-3))$ of
rank at least $1+(r-3)=r-2$, and $I^{2,1}\cong I^{1,2}\supset I^{1,2}(H_{3}(Y)(-3))$
of rank at least $1$. The only possible \LMHS type, given that $H_{lim}^{3}(\mathrm{X}_{t})$
has $Gr_{F}^{i}$ ranks $1$, $r-1$, $r-1$, $1$, is \[\begin{minipage}{0.3\textwidth}  \includegraphics[scale=0.65]{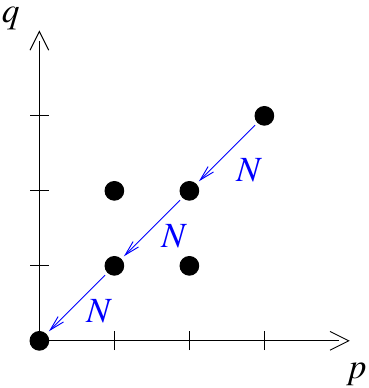} \end{minipage} \hfill \begin{minipage}{0.65\textwidth} \begin{align*} 1=\text{rk}(I^{0,0})=\text{rk}(I^{2,1})=\text{rk}(I^{1,2})=\text{rk}(I^{3,3}) \\ r-2 = \text{rk}(I^{1,1})=\text{rk}(I^{2,2}) \\ \text{(all others zero)} \end{align*} \end{minipage} \]This
implies at once that (the image of) $H_{3}(Y_{\ua})(-3)$ is \emph{all}
of $\ker(N)=H_{inv}^{3}(\mathrm{X}_{t})$, and so $\Theta$, $\Theta'$,
and $B^{*}$ are all isomorphisms:
\begin{thm}\phantomsection\label{MKthm1}
We have isomorphisms of $\QQ$-\VMHS  
\[
H_{3}(Y_{\ua})(-3)\cong H^{3}(\mathrm{X}_{\ua,0})\cong H^{3}(\mathrm{X}_{\ua,0}')\cong H_{inv}^{3}(\mathrm{X}_{\ua,t}).
\]
\end{thm}

\subsection{Monodromy and asymptotics of periods}

\label{MKsecIG}We begin by addressing the nature of the limiting
periods (i.e., by Theorem~\ref{MKthm1}, periods on $Y_{\ua}$). Set\begin{align*}
\eta_{\ua} & := \frac{1}{(2\pi \ay)^3} \text{Res}_{Y_{\ua}} \left( \frac{dx/x\wedge dy/y \wedge du\wedge dv}{f_{\ua}(x,y)+uv}\right)
\\
& = \left. \frac{1}{(2\pi \ay)^3} \frac{dx}{x}\wedge \frac{dy}{y}\wedge \frac{du}{u}\right|_{Y_{\ua}} \in \Omega^3\left( Y_{\ua} \right) .
\end{align*}Write \begin{equation}\label{MKeqnIR}R\{ x,y\} :=\log(x)\frac{dy}{y} - 2\pi \ay \log(y) \delta_{T_x}
\end{equation}for the 1-current on $\GG_{m}^{2}$, where $\log(x)$ is the (discontinuous)
branch with argument in $(-\pi,\pi)$, and $T_{x}=x^{-1}(\RR_{-})$
(with $\RR_{-}$ oriented from $-\infty$ to $0$). For any invariant
3-cycle $\kappa$, 
\[
\lim_{t\to0}\int_{\kappa}\Omega_{\ua,t}=\int_{\kappa}\eta_{\ua}
\]
which for $\kappa=\cT$ is $1$ and for $\kappa=\cM(\vf)$ is (according
to \cite[\S 5.1]{DoranKerr})\begin{equation}\label{MKeqnInew*}\frac{1}{(2\pi\ay)^{2}}\int_{\Gamma_{\vf}}\frac{dx}{x}\wedge\frac{dy}{y}=\frac{1}{(2\pi\ay)^{2}}\int_{\vf}R\{x,y\}|_{E_{\ua}^{*}}=:\frac{1}{2\pi\ay}R_{\vf}(\ua).
\end{equation}These ``regulator periods'' were computed in \cite[\S 5.2]{DoranKerr}
for a specific choice%
\footnote{cf. p. 487 of \cite[\S 5.2]{DoranKerr}, 
 where these ``distinguished vanishing
cycles'' are denoted $\vf_{0}^{[i]}$. Also see \cite[\S 4.1-2]{BKV}
for a brief introduction to regulator currents in the context of Feynman
integrals.%
} of $\nu-2$ $\{\vf_{0}^{(i)}\}$ and with only those $\{a_{j}\}$
attached to vertices nonzero. If $\Delta$ has $\nu=r$ (no interior
integral points on edges), then these $\{\vf_{0}^{(i)}\}$ are ``enough''
(we need $r-2$) and no $a_{j}$ are set equal to $0$.

In that case, and with that choice of basis --- to both of which we
shall restrict for the moment --- we get an alternate proof of the
independence of the $r$ invariant cycles $\left\{ \cM_{\alpha}(\vf_{1}),\{\cM_{\alpha}(\vf_{0}^{(i)})\}_{i=1}^{r-2},\tau_{\alpha}\right\} $
in $H_{(lim)}^{3}(\mathrm{X}_{t})$, which is simpler than constructing
the full basis of $H^{3}(\mathrm{X}_{t})$. Namely, observe
(cf. \cite[\S 5.2]{DoranKerr})
that\begin{equation}\label{MKeqnIp21sharp}\begin{cases} R_{\vf_0^{(i)}}(\ua)  \sim  \log(z_i) \;\;\;\;(i=1,\ldots,r-2) \\ R_{\vf_1}(\ua) \sim  \tfrac{1}{2\pi \ay} Q\left(\log(z_1),\ldots,\log(z_{r-2})\right)\;\;\;\; \left(\tiny \begin{array}{c}Q\text{ quadratic,}\\\text{with }\QQ\text{-coeffs.}\end{array} \right) \end{cases}
\end{equation}are independent functions as $\ua$ varies; therefore, so are the
$\int_{\cM_{\alpha}(\vf)}\Omega_{\ua,t}$ and $\int_{\tau_{\alpha}}\Omega_{\ua,t}$,
and independence of the cycle classes follows.

From \eqref{MKeqnIp21sharp} we also obtain information on the asymptotics
of periods of $H^{3}(\mathrm{X}_{F})$ which will be key for defining
the mirror map in \S\ref{MKsecIID}. Writing $T_{j}$ for the monodromy
about $z_{j}=0$ (and $N_{j}=\log(T_{j})$), $\cT_{\alpha}$ is the
cycle whose ``Tube'' (for all $|z_{i}|\ll\epsilon$) is $\cap_{i=1}^{r-2}\{|z_{i}|=\epsilon\}$,
hence is invariant by all $T_{j}$. (In particular, we have $\lim_{t\to0}\int_{\cT_{\alpha}}\Omega_{F}=1$.)
Putting 
\[
\tilde{\Omega}_{F}:=\frac{\Omega_{F}}{\int_{\tau_{\alpha}}\Omega_{F}},
\]
we define normalized B-model periods by
\[
\Pi_{\gamma}^{\text{B}}(\uz):=\int_{\gamma}\tilde{\Omega}_{F}\;\;\;\;\;\;\;\;(\gamma\in H_{3}(\mathrm{X}_{F},\ZZ)).
\]
Obviously $\Pi_{\cT_{\alpha}}^{\text{B}}$ is identically $1$ whilst (using \eqref{MKeqnInew*}
and \eqref{MKeqnIp21sharp}) \begin{equation}\label{MKeqnIp22ss}\Pi_{\cM^{\text{B}}_{\alpha}(\vf_0^{(i)})}\sim \frac{1}{2\pi \ay}\log(z_i)\;\;\;\;\;\;\;\;\; (i>0).
\end{equation}

Now when we take all $\{a_{i}\}_{i>0}$, hence all $\{z_{i}\}_{i>0}$,
to zero, the equation for $\mathrm{X}$ becomes
\[
\bb u^{3}+\cc v^{2}+\aa uv+a_{0}=0.
\]
That is, we are left with an isotrivial family of elliptic curves 
$\cong:\cE_{z_{0}}$ over $\GG_{m}^{2}$ (as $\rho^{-1}E$ and $\rho^{-1}D$
have both collapsed to $\rho^{-1}\DD$), with $\rho^{-1}\DD$ consisting
of all 3-fold components of $\DD_{\hDelta}$ associated to points
of $\hDDelta_{\ZZ}$ dual to an edge or vertex of $\Delta\times(-1,-1)$.
Clearly $\cT_{\alpha}$ and $\cT_{\beta}$ remain in the cohomology
of this singular 3-fold, so $\cT_{\beta}$ is invariant under all
$\{T_{j}\}_{j>0}$. Taking $\varpi_{z_{0}}\in\Omega^{1}(\cE_{z_{0}})$
to be normalized so that $\int_{\alpha}\varpi_{z_{0}}=1$, the limiting
period 
\[
\Pi_{\cT_{\beta}}^{\text{B}}(z_{0},\underline{0})=\frac{1}{(2\pi\ay)^{2}}\int_{\cT_{\beta}}\frac{dx}{x}\wedge\frac{dy}{y}\wedge\varpi_{z_{0}}=\int_{\beta}\varpi_{z_{0}}
\]
is asymptotic to $\frac{5}{2\pi\ay}\log t$ (since $\cE_{z_{0}}$
limits to an $I_{5}$ at $t=0$); therefore\begin{equation}\label{MKeqnInewp16}\Pi_{\cT_{\beta}}^{\text{B}}\sim\frac{1}{2\pi \ay}\log(z_{0}).
\end{equation}

Remark that by \eqref{MKeqnIp22ss}, the $\{\cM_{\alpha}(\vf_{0}^{(i)})\}_{i\neq j}$
are invariant under $T_{j}$ for $j>0$; equivalently, the membrane
in $\GG_{m}^{2}\backslash D^{*}$ bounding on each $\vf_{0}^{(i)}\subset E^{*}$
($i\neq j$) behaves well in the $z_{j}\to0$ limit (under which $E$
and $D$ become nodal rational curves in the same linear system).
Symmetrically, there are $'\vf_{0}^{(i)}\subset D^{*}$ with the same
properties. So for $j>0$, 
\[
\ker(T_{j}-I)=\ker(N_{j})=\langle\cT_{\alpha},\cT_{\beta},\{\cM_{\alpha}(\vf_{0}^{(i)})\}_{i\neq j},\{\cM_{\beta}({}'\vf_{0}^{(i)})\}_{i\neq j}\rangle,
\]
while as previously remarked $\ker(N_{0})=\langle\cT_{\alpha},\{\cM_{\alpha}(\vf_{0}^{(i)})\}_{i=1}^{r-2},\cM_{\alpha}(\vf_{1})\rangle.$
Combining this with the fact that $\cM_{\beta}(\vf_{1})$ is the only
cycle in our basis pairing nontrivially with $\cT_{\alpha}$ (e.g.,
consider the above $z_{1},\ldots,z_{r-2}\to0$ limit), we can compute
a basis for $W_{\bullet}^{0}:=W(N_{0})_{\bullet}$. Writing $\{\psi_{I}\}_{i=1}^{r-3}$
for a basis of $\ker\{\cK_{E}\to H_{1}(E)\}$, so that $\cK_{E}=\langle\{\psi_{i}\},\vf_{0}^{(1)}\rangle$,
we have
\[
W_{0}^{0}=\langle\cT_{\alpha}\rangle,\;\;\;\;\; W_{2}^{0}=W_{0}^{0}+\langle\{\cM_{\alpha}(\psi_{i})\}_{i=1}^{r-3},\cT_{\beta}\rangle,
\]
\[
W_{3}^{0}=W_{2}^{0}+\langle\cM_{\alpha}(\vf_{0}^{(1)}),\cM_{\alpha}(\vf_{1})\rangle,
\]
\[
W_{4}^{0}=W_{3}^{0}+\langle\{\cM_{\beta}(\vf_{0}^{(i)})\}_{i=1}^{r-2}\rangle,\;\;\;\;\; W_{6}^{0}=W_{4}^{0}+\langle\cM_{\beta}(\vf_{1})\rangle.
\]
For the (Hodge-Tate) limit at $\uz\!=\!\underline{0}$, we have (for $W_{\bullet}\!=\!W(N_{0}\!+\!\cdots\!+\!N_{r-2})_{\bullet}$)
\[
W_{0}=\langle\cT_{\alpha}\rangle,\;\;\;\;\; W_{2}=W_{0}+\langle\cT_{\beta},\{\cM_{\alpha}(\vf_{0}^{(i)})\}_{i=1}^{r-2}\rangle,
\]
\[
W_{4}=W_{2}+\langle\{\cM_{\beta}({}'\vf_{0}^{(i)})\}_{i=1}^{r-2},\cM_{\alpha}(\vf_{1})\rangle,\;\;\;\;\; W_{6}=W_{4}+\langle\cM_{\beta}({}'\vf_{1})\rangle.
\]
The other $W(N_{j})_{\bullet}$'s are more difficult and will be computed
via the A-model in \S\ref{MKsecIIC}.

We will say more about the specialization to the sunset case, where
\[
f_{\ua}=a_{0}+a_{1}x+a_{2}y+a_{3}x^{-1}y+a_{4}x^{-1}+a_{5}y^{-1}+a_{6}xy^{-1},
\]
in the next section, but some preliminary remarks are in order. Writing
$\hat{\vf}_{0}^{(i)}$ for the vanishing cycle in $H_{1}(E_{\ua}^{*})$
for $a_{i}\to0$, we have (with indices modulo 6) $\vf_{0}^{(j)}=-\hat{\vf}_{0}^{(j)}+\hat{\vf}_{0}^{(j-1)}+\hat{\vf}_{0}^{(j+1)},$
which all map to the same cycle $\vf_{0}\in H_{1}(E_{\ua})$, and
\[
z_{j}=\frac{a_{j+1}a_{j-1}}{a_{j}a_{0}},\;\;\;\;\; j=1,\ldots,6.
\]
The apparent obstacle here is that although $r\!=\!r^{\circ}\!=\!\nu\!=\!6$,
the 4-dimensional cone spanned by the vectors $\ul^{(j)}$ of \cite{DoranKerr}
($\ell_{j}^{(j)}=-1$, $\ell_{j-1}^{(j)}=\ell_{j+1}^{(j)}=1$, $\ell_{\text{other}}^{(j)}=0$)
is not simplicial. Hence $z_{1}$ thru $z_{4}$ do \emph{not} suffice
to parametrize the resulting \emph{singular} local parameter 4-space,
and $\vf_{1},\{\vf_{0}^{(j)}\}_{j=1}^{4}$ span $\cK_{E}$ rationally
but \emph{not} integrally. (We need all 6 $\{z_{i}\}$ and all 6 $\{\vf_{0}^{(j)}\}$,
and the relations on the $\{z_{i}\}$ produce the singularity.) Writing
$R(m)\subseteq\ZZ^{6}$ for the set of ``relations vectors'' $(\ell_{1},\ldots,\ell_{6})$
with $\sum\ell_{j}=m\in\mathbb{N}$, $\sum\ell_{j}\um^{(j)}=\underline{0}$,
we have ($j=1,\ldots,6$) \begin{align*} \lim_{z_0 \to 0} 2\pi \ay\cdot \Pi_{\cM_{\alpha}(\vf_0^{(j)}}^{\text{B}} &= R_{\vf_0^{(i)}}(\ua)=\log(z_j)+H(\ua) \\ &= \log(z_j) +\sum_{m\geq 1} \frac{1}{m} \sum_{\ul\in R(m)} \frac{m!}{\ell_1 ! \cdots \ell_6 !}\cdot  \frac{a_1^{\ell_1}\cdots a_6^{\ell_6}}{(-a_0)^m} \end{align*}for
the limiting periods (cf. \cite[eqn (5.4)]{DoranKerr}).

Upon specializing to the ``Feynman locus'' where \begin{equation}\label{MKeqnIp25*}f_{\ua} = 1-s\phis,\;\;\;\;\; \phis:=(-\xi_1^2 x - \xi_2^2 y + \xi_3^2)(1-x^{-1}-y^{-1}),
\end{equation}a small miracle occurs. The resulting substitutions $\ua=(\xi_1,\xi_2,\xi_3,s)$
yield \begin{equation}\label{MKeqnIp25!}z_1 = -\frac{\xi_2^2s}{a_0} = z_4,\quad z_2 = -\frac{\xi_1^2s}{a_0} = z_5 , \quad z_3 = -\frac{\xi_3^2s}{a_0} = z_6,
\end{equation}where $a_{0}=1-s\sum \xi_{i}^{2}$, and $R_{0}^{(j)}:=R_{\vf_{0}^{(i)}}(\ua)=$\begin{equation}\label{MKeqnIp25!!}\log(z_j) + \sum_{m\geq 1}\frac{1}{m} \sum_{\ul\in R(\um)} \frac{(-1)^{\ell_3 + \ell_6} m!}{\ell_1 ! \cdots \ell_6 !} z_1^{\ell_2 + \ell_3} z_2^{\ell_1+\ell_6} z_3^{\ell_4 + \ell_5} ,
\end{equation}which will be considered as a function of $(z_{1},z_{2},z_{3})$.
In particular, we have\begin{equation}\label{MKeqnIp25**}R_0^{(1)}=R_0^{(4)},\quad R_0^{(2)}=R_0^{(5)},\quad R_0^{(3)}=R_0^{(6)};
\end{equation}in effect, the specialization has replaced a singular 4-fold local
parameter space by a smooth 3-dimensional slice. From the standpoint
of periods (of $R\{x,y\}$ on $E^{*}$, or $\eta$ on $Y$), the class\begin{equation}\label{MKeqnIp25!*}-\vf_0^{(1)} + \vf_0^{(4)} = \vf_0^{(2)} - \vf_0^{(5)} = -\vf_0^{(3)} + \vf_0^{(4)}
\end{equation}in $\cK_{E}$ is now ``trivial'', and the quotient $\overline{\cK_{E}}$
of $\cK_{E}$ by \eqref{MKeqnIp25!*} is integrally spanned by $\vf_{0}^{(1)},\vf_{0}^{(2)},\vf_{0}^{(3)}$.
Recalling that $\ker(N_{0})\subset H_{lim}^{3}(\mathrm{X})$ is an
extension 
\[
\QQ(0)\to\ker(N_{0})\to\cK_{E}(-2),
\]
the immediate consequence is that the $\QQ(-1)\subset\cK_{E}(-2)$
spanned by \eqref{MKeqnIp25!*} lifts to a $\QQ(-1)\subset\ker(N_{0})$.
It is the quotient $\overline{\ker(N_{0})}$ by this constant sub-\VMHS
which we will be interested in when comparing with the A-model.

It will be useful in the sequel to denote by $\cV_{\mathrm{B}}$ the
$\ZZ$-\VHS $H^{3}(\mathrm{X}_{F})$ considered over a product of punctured
disks with parameters $z_{0},\ldots,z_{r-2}$.

\section{A-model}\label{sec:Amodel}

We now turn to the (integral) variation of Hodge structure arising
from the quantum product on $H^{even}(\rXc)$, where $\rXc\subset\PP_{\hDDelta}$
is the Batyrev mirror of $\rX$. Its equivalence to the B-model \VHS
$H^{3}(\rX)$ allows us to compute all monodromies of the latter (about
the hyperplanes $\{z_{i}=0\}$) and relate its $z_{0}\to0$ \LMHS to
\emph{local} Gromov-Witten data for $\PP_{\DDelta}$. In order to
make use of the computations in \cite[\S 5]{DoranKerr}, we shall
work under the assumption that $\nu=r$ (so that $(\d\Delta)_{\ZZ}$
consists of vertices).

\subsection{Elliptic fibration and even cohomology}

As in the B-model case, triangulating $\d\hDelta$ produces a resolution
of singularities $\PP_{\hDDelta}\twoheadrightarrow\check{\PP}_{\hDDelta}$.
The desired triangulation is achieved by:
\begin{itemize}
\item inserting the $\tfrac{1}{2}$-planes $\HH_{\uw}$ ($\uw\in(\d\Delta)_{\ZZ}$)
as in Step 1 of \S\ref{MKsecIB}, which subdivides the 2-face $\Delta\times(-1,-1)$
and each of the facets 
		\begin{align*}
			\mathfrak{f}_{1}&=\text{hull}\{\Delta\times(-1,-1),\,(0,0,-1,1)\},\\
			\mathfrak{f}_{2}&=\text{hull}\{\Delta\times(-1,-1),\,(0,0,2,-1)\}
		\end{align*}
 into $r$ pieces; and
\item further subdividing the facet-pieces by inserting 2-planes through
the edges of $\Delta\times(-1,-1)$ and $(0,0,1,-1)$, $(0,0,0,-1)$,
resp. $(0,0,-1,0)$.
\end{itemize}
The first step guarantees a morphism $\PP_{\hDDelta}\overset{\mathcal{P}^{\circ}}{\twoheadrightarrow}\PP_{\DDelta}$,
with generic fiber $\PP_{W}$. Its restriction to an anticanonical
(Calabi-Yau) hypersurface $\mathrm{X}^{\circ}\overset{\imath}{\subset}\PP_{\hDDelta}$,
cut out by a generic Laurent polynomial (with Newton polytope $\hDDelta$),
produces a Weierstrass elliptic fibration $\rho^{\circ}:\,\rXc\twoheadrightarrow\PP_{\DDelta}$.
The discriminant locus of $\rho^{\circ}$ (over which the fiber is
$I_{1}$) is a higher-genus curve meeting $\DD_{\DDelta}$ properly;
in particular, $(\rho^{\circ})^{-1}$ of components of $\DD_{\DDelta}$
(and of their intersections) are smooth.

Let $D_{0}\cong\PP_{\DDelta}$ denote the (zero-)section of $\rho^{\circ}$
given by intersecting $\rXc$ with the component of $\DD_{\hDDelta}$
dual to $(0,0,-1,-1)\in(\d\hDelta)_{\ZZ}$. Writing $\DD_{\DDelta}=\cup_{i=1}^{r}C_{i}$
(with the counterclockwise ordering), the divisors $D_{i}:=(\rho^{\circ})^{-1}(C_{i})$
are the intersections with $\rXc$ of components dual to the $\{\uv_{i}\times(-1,-1)\}$
(where $\{\uv_{i}\}$ are the vertices of $\Delta$). \[ \includegraphics[scale=0.50]{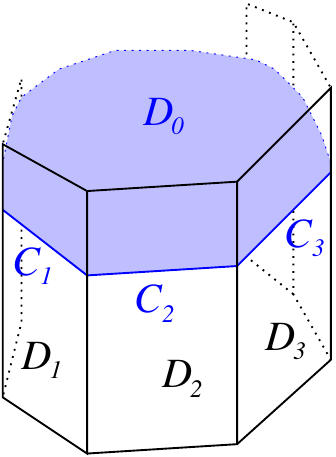} \](In
the sequel, $C_{i}$ will mean $D_{0}\cap D_{i}\subset\rXc$.) There
are five more components of $\DD_{\hDDelta}$: those dual to $(0,0,1,-1)$,
$(0,0,0,-1)$, and $(0,0,-1,0)$ do not meet $\rXc$; and we denote
by $D',D''$ the intersections with $\rXc$ of those dual to $(0,0,2,-1)$
resp. $(0,0,-1,1)$.

The divisors of the toric coordinates $\{X_{i}\}_{i=1}^{4}$ restricted
to $\rXc$ are then given by \begin{equation}\label{MKeqnII1}\begin{cases} (X_1)=\sum_{i=1}^r v_i^{(1)} D_i,\;\;\; (X_2)=\sum_{i=1}^r v_i^{(2)} D_i , \\ (X_3) = 2D' - D'' - \sum_{i=0}^r D_i ,\;\;\; (X_4) = D''-D'-\sum_{i=0}^r D_i \end{cases}
\end{equation}so that in $CH^{1}(\rXc)\cong H^{2}(\rXc,\ZZ)$ we have $D'\equiv2\sum_{i=0}^{r}D_{i}$,
$D''\equiv3\sum_{i=0}^{r}D_{i}$, and $D_{r-1},D_{r}\in\text{span}\langle\{D_{i}\}_{i=1}^{r-2}\rangle$.
Now $D'\cap D_{0}$ and $D''\cap D_{0}$ are empty (as the corresponding
faces of $\hDDelta$ meet in vertices), and so in $CH^{2}(\rXc)$
(hence $H^{4}(\rXc,\ZZ)$)%
\footnote{We shall use $\cdot$ and often nothing
as cup product on $H^{even}$.%
}

\begin{equation}\label{MKeqnII2}D_0 \cdot D_0 \equiv -\sum_{i=1}^r D_0 \cdot D_i = -\sum_{i=1}^r C_i \equiv -E^{\circ}
\end{equation}where $E^{\circ}$ is a general anticanonical (elliptic) curve in
$D_{0}\cong\PP_{\DDelta}$. Writing $d_{i}$ for $\ell(i^{\text{th}}\text{ edge of }\d\DDelta)$,
so that $r^{\circ}=\sum_{i=1}^{r}d_{i}$, we therefore have \begin{equation}\label{MKeqnII3}(D_0 \cdot C_i ) = (D_0 \cdot D_0 \cdot D_i ) = -(E^{\circ}\cdot D_i ) = -(E^{\circ}\cdot C_i )_{D_0} = -d_i
\end{equation}for $i=1,\ldots,r$.

From the first line of \eqref{MKeqnII1}, we also have $C_{r-1},C_{r}\in\text{span}\langle\{C_{i}\}_{i=1}^{r-2}\rangle$
so that $\{C_{i}\}_{i=1}^{r-2}$ span $H^{2}(D_{0})$, and $[(C_{i}\cdot C_{j})_{D_{0}}]_{i,j=1}^{r-2}=[(C_{i}\cdot D_{j})_{\rXc}]_{i,j=1}^{r-2}$
is nondegenerate. Since a general fiber $C_{0}$ of $\rho^{\circ}$
satisfies $(C_{0}\cdot D_{0})=1$ and $(C_{0}\cdot D_{i})=0$ ($i>0$),
$[(C_{i}\cdot D_{j})_{\rXc}]_{i,j=0}^{r-2}$ is in fact nondegenerate.
Using that $H^{1,1}(\rXc)=h^{2,1}(\rX)=r-1$, it follows that a basis
for 
\[
V=H^{even}(\rXc,\CC)=\oplus_{k=0}^{3}H^{k,k}(\rXc)
\]
is given by $\{\rXc;D_{0},\ldots,D_{r-2};C_{0},\ldots,C_{r-2};p\}$
where $p\in\rXc$ is a point. Write\begin{equation}\label{MKeqnII3.5}J_j = \sum_{k=0}^{r-2} \alpha_j^k D_k  \;\;\;\;\; (\alpha_j^k \in \QQ ;\; j=0,\ldots,r-2 )
\end{equation}for the basis of $H^{2}(\rXc,\QQ)$ Poincar\'e-dual to the $\{C_{j}\}_{j=0}^{r-2}\subset H^{4}(\rXc,\QQ)$. 

Clearly all the $\alpha_{i}^{0}=0$ for $i>0$, so using \eqref{MKeqnII3}
we find that\begin{equation}\label{MKeqnII4}\begin{cases} J_0 = D_0 + (\pi^{\circ})^{-1} E^{\circ} = D_0 + D_1 + \cdots + D_r = D_0 + \sum_{i=1}^{r-2} \alpha_0^i D_i \\ D_0 = J_0 - \sum_{i=1}^{r-2} d_i J_i \end{cases}
\end{equation}hence (by \eqref{MKeqnII2})\begin{equation}\label{MKeqnII5}J_0^2 = r^{\circ} C_0 + C_1 + \cdots + C_r = r^{\circ} C_0 + \sum_{i=1}^{r-2} \alpha_0^i C_i .
\end{equation}For the triple-products, we evidently have $J_{0}^{3}=r^{\circ}$
(dropping the class of the point``$p$''), and $J_{i}J_{j}J_{k}=0$ if $i,j,k>0$.
For $j>0$ we find (by \eqref{MKeqnII5}\begin{equation}\label{MKeqnII6}J_0^2 J_j = \sum_{i=1}^{r-2} \alpha_0^i (C_i \cdot J_j) = \alpha_0^j ,
\end{equation}while for $i,j>0$\begin{equation}\label{MKeqnII7}J_0 J_i J_j = \sum_{k=1}^{r-2} \alpha_j^k J_0 J_i D_k = \sum_k \alpha_j^k D_0 J_i D_k =\sum_k \alpha_j^k (J_i\cdot C_k)= \alpha_j^i .
\end{equation}In particular, $[\alpha_{j}^{i}]_{i,j=1}^{r-2}$ is symmetric, which
reflects the fact that it is the inverse of $[(C_{i}\cdot C_{j})_{D_{0}}]_{i,j=1}^{r-2}$
(which can be computed from the $\ell$-vectors of \cite[\S 5]{DoranKerr}).

From \eqref{MKeqnII6} and \eqref{MKeqnII7} we also have\begin{equation}\label{MKeqnII8}J_i J_j = \alpha_j^i C_0 \;\;\; \text{and}\;\;\; J_0J_j =\sum_{i=0}^{r-2} \alpha_i^j C_i .
\end{equation}Since $D_{0}\cdot J_{0}=-E^{\circ}+E^{\circ}=0$ by \eqref{MKeqnII2}
and \eqref{MKeqnII4}, using \eqref{MKeqnII3} and \eqref{MKeqnII8}
to evaluate $(0=)\, D_{0}\cdot J_{0}\cdot J_{j}$ yields the intriguing
relations \begin{equation}\label{MKeqnII9}\begin{cases} \alpha_0^j = \sum_{i=1}^{r-2} d_i \alpha_j^i \;\;\; (j=1,\ldots,r-2) \\ r^{\circ} = \sum_{i=1}^{r-2} d_i \alpha_0^i . \end{cases}
\end{equation}For the sequel we set $\tilde{\alpha}_{j}^{i}=J_{0}J_{i}J_{j}$, which
allows us to rewrite \eqref{MKeqnII9} as $\tilde{\alpha}_{j}^{0}=\sum_{i=1}^{r-2}d_{i}\tilde{\alpha}_{j}^{i}$
for $j=0,\ldots,r-2$.

\subsection{The quantum $\ZZ$-variation of Hodge structure}

\label{MKsecIIB}Following \cite[\S 8]{CoxKatz} and \cite[\S 5]{Iritani},
we now introduce a weight 3 \VHS on $V_{\cO}=V\otimes\cO((\Delta^{*})^{r-1})$,
where the $\Delta^{*}$ are punctured disks with coordinates $q_{j}=e^{2\pi\ay\tau_{j}}$,
$j=0,\ldots,r-2$. (Write $\kappa:\,\mathfrak{H}^{r-1}\to(\Delta^{*})^{r-1}$
for the obvious map sending $\underline{\tau}\mapsto\underline{q}$.)
The Hodge filtration is straightforward, given by
\[
F^{p}V_{\cO}:=\oplus_{a\geq p}H^{3-a,3-a}(\rXc)_{\cO},
\]
so that $1_{\rXc}=[\rXc](\otimes1)$ generates $F^{3}$. The polarization
is just $(\mathtt{A},\mathtt{B}):=\linebreak (-1)^{\frac{1}{2}\deg\mathtt{A}}\int_{\rXc}\mathtt{A}\cdot\mathtt{B}$.

Let $\tilde{N}_{\ukh}$ denote the genus-zero Gromov-Witten invariant
for the class $C_{\ukh}:=\sum_{\ell=0}^{r-2}k_{\ell}C_{\ell}\in H_{2}(\rXc,\ZZ)$,
for any $\ukh=(k_{0},\underline{k})\in\ZZ_{\geq0}^{r-1}$. Using the
\GW prepotential \begin{equation}\label{MKeqnII10}\Phi := \frac{(2\pi\ay)^3}{6} \int_{\rXc} \left(\sum_{j=0}^{r-2}\tau_j J_j \right)^3 + \sum_{\underline{k}\neq\underline{0}} \tilde{N}_{\ukh} \uq^{\ukh} ,
\end{equation}we define the quantum product ``$*$''  on $V_{\cO}$ to be cup product on the
last subsection's basis ($\otimes1$) except for
\[
J_{i}*J_{j}:=\frac{1}{(2\pi\ay)^{3}}\sum_{\ell=0}^{r-2}\Phi_{ij\ell}'''C_{\ell}=J_{i}\cdot J_{j}+\text{h.o.t.}(\uq).
\]
where $\text{h.o.t.}(\uq)$ denote higher order term in the $q$ expansion.
Here $\Phi_{ij\ell}'''=\d_{i}\d_{j}\d_{\ell}\Phi$, where $\d_{i}:=\frac{\d}{\d q_{i}}$.

The ($\CC$-)local system $\VV_{\CC}\subset V_{\cO}$ is then given
by the kernel of 
\[
\nabla:=\text{id}_{V}\otimes d+\sum_{i=0}^{r-2}(J_{i}*)\otimes d\tau_{i},
\]
with monodromy logarithms\begin{equation}\label{MKeqnII11}N_i =
  \log(T_i) = -2\pi\ay {Res}_{q_i = 0} (\nabla) = -J_i \cdot (\,)
\end{equation}about the coordinate axes. A basis of $\nabla$-flat
sections is given by
the map
\[
\sigma:\, B\to\Gamma(\mathfrak{H}^{r-1},\kappa^{*}\VV_{\CC})
\]
sending\begin{equation}\label{MKeqnII12}\begin{cases} p\mapsto p \\ C_i \mapsto C_i - \tau_i p \\ J_i \mapsto J_i - \frac{1}{(2\pi \ay)^3}\sum_j \Phi_{ij} '' C_j + \frac{1}{(2\pi \ay)^3} \Phi_i ' p \\ \rXc \mapsto \rXc - \sum_i \tau_i J_i -\frac{1}{(2\pi \ay)^3} \sum_i \left( \Phi_i ' - \sum_j \tau_j \Phi_{ij} '' \right) C_i \\ \mspace{250mu}+ \frac{1}{(2\pi \ay)^3} \left( 2\Phi - \sum_j \tau_j \Phi_j ' \right) p . \end{cases}
\end{equation}

To define Iritani's integral local system $\VV_{\ZZ}\subset\VV_{\CC}$,
we will need his ``square root of the Todd class''
\[
\hat{\Gamma}(\rXc):=\exp\left(-\frac{1}{24}{ch}_{2}(\rXc)-\frac{2\zeta(3)}{(2\pi\ay)^{3}}{ch}_{3}(\rXc)\right)\in K_{0}(\rXc).
\]
In general, for a toric variety $\PP_{\Sigma}$ defined by a simplicial
fan $\Sigma$,
\[
c_{i}(\PP_{\Sigma})=\sum_{\tau\in\Sigma^{(i)}}Z_{\tau}
\]
where $\Sigma^{(i)}$ denotes the $i$-dimensional cones of $\Sigma$
and $Z_{\tau}$ is the (codimension-$i$) intersection of the divisors
of $\PP_{\Sigma}$ dual to the generators of $\tau$. Applying this
to $\PP_{\hDelta}$ and pulling back to $\rXc$,
\begin{align*}
	\imath^{*}c(\PP_{\hDelta})&=1+\left(D'+D''+\sum_{i=0}^{r}D_{i}\right)\\
	&\quad +\left((11r^{\circ}+r)C_{0}+12\sum_{i=1}^{r}C_{i}\right)+6(r+r^{\circ})p
\end{align*}
while 
\[
\imath^{*}c(\cO(\rXc)^{-1})=1-6\sum_{i=0}^{r}D_{i}+36\left(r^{\circ}C_{0}+\sum_{i=1}^{r}C_{i}\right)-216r^{\circ}p.
\]
(Remark that $\sum_{i=1}^{r}C_{i}=\sum_{i=1}^{r-2}\alpha_{0}^{i}C_{i}$
by \eqref{MKeqnII5}.) This yields 
\[
c(\rXc)=\imath^{*}c(\PP_{\hDelta})\cdot c(\cO(\rXc))^{-1}=1+\left((11r^{\circ}+r)C_{0}+12\sum_{i=1}^{r}C_{i}\right)-60r^{\circ}p,
\]
hence \begin{equation}\label{MKeqnII13}\begin{cases} ch(\rXc) = 3 - \left( 12 \sum_{i=1}^r C_i + (11r^{\circ}+r)C_0 \right) - 30 r^{\circ} p \\ td(\rXc) = 1+\left( \sum_{i=1}^r C_i + \frac{1}{12} (11 r^{\circ} + r) C_0 \right)  \\  \hat{\Gamma}(\rXc) =1 + \left( \frac{1}{2} \sum_{i=1}^r C_i + \frac{1}{24} (11 r^{\circ} + r) C_0 \right) + \frac{60 \zeta(3)}{(2\pi\ay)^3} r^{\circ} p . \end{cases}
\end{equation}

The $\ZZ$-local system (or rather its $\kappa^{*}$-pullback) is
then defined by the image of \begin{equation}\label{MKeqnII14}\begin{array}{cccc}\gamma: & K_{0}^{num}(\rXc) & \to & \Gamma(\mathfrak{H}^{r-1},\kappa^{*}\VV_{\CC})\\ & \xi & \mapsto & \sigma(\hat{\Gamma}(\rXc)\cdot ch(\xi)) .\end{array}
\end{equation}The role of the $\hat{\Gamma}$-class is tied to the Mukai pairing
$\langle\;,\;\rangle:\, K_{0}^{num}(\rXc)\times K_{0}^{num}(\rXc)\to\ZZ$,
defined (on the level of vector bundles) by \begin{equation}\label{MKeqnII15}\langle \xi ,\xi ' \rangle := \int_{\rX} ch(\xi^{\vee}\otimes \xi ' )\cdot td (\rXc) .
\end{equation}Iritani's result that \begin{equation}\label{MKeqnII16}(\gamma (\xi),\gamma(\xi ') )=\langle \xi ,\xi '\rangle 
\end{equation}implies the integrality of $(\;,\;)$ on $\VV_{\ZZ}$, and the integrality
of monodromy follows from \begin{equation}\label{MKeqnII17}T_i (\gamma (\xi )) = \gamma (\cO (-J_i)\otimes \xi ) .
\end{equation}The ``period'' of the $(3,0)$-section $[\rXc]\otimes1$ against
the integral class $\gamma(\xi)$ is \begin{equation}\label{MKeqnII!!!8}\Pi^{\text{A}}_{\xi} (\uq) := \langle 1_{\rXc} , \gamma (\xi )\rangle = \,\text{coefficient of }[p]\text{ in }\gamma(\xi).
\end{equation}

To compute $\gamma$ or the Mukai pairing, we first find the Chern
characters of various skyscraper sheaves by resolving them by vector
bundles, e.g.
\[
\cO_{J_{i}}=\cO-\cO(-J_{i}),\;\;\;\cO_{C_{i}}=\cO-\cO(-D_{i})-\cO(-D_{0})+\cO(-(D_{0}+D_{i})),
\]
and using $ch(\cO(D))=e^{D}$ for any divisor $D$. Writing $\ell_{j}^{i}:=(C_{i}\cdot D_{j})$
($i,j=0,\ldots,r$),%
\footnote{note $\ell_{0}^{i}=d_{i}$%
} this gives\begin{align*}
& ch(\cO_p) = p ,\;\;\; ch(\cO_{\rXc})=\rXc,
\\
& ch(\cO_{C_0})=C_0 ,\;\;\; ch(\cO_{C_j})=C_j + \tfrac{1}{2}(d_j - \ell_j^j)p,
\\
& ch(\cO_{J_0}) = J_0 - \tfrac{1}{2}(r^{\circ}C_0 + \Sigma_{i=1}^r C_i ) + \tfrac{r^{\circ}}{6} p,\;\;\; ch(\cO_{J_j}) = J_j - \tfrac{1}{2}\alpha_j^j C_0,
\\
& ch(\cO_{D_0}) = D_0 + \tfrac{1}{2}\Sigma_{i=1}^r C_i + p , \;\;\; ch(\cO_{D_j}) = D_j + \tfrac{1}{2} C_0 .
\end{align*}In particular, we find that a Mukai-symplecitic basis $\{\xi_{k}\}_{k=1}^{2r}$
of $K_{0}^{num}(\rXc)_{\QQ}$ (hence, applying $\gamma$, a symplectic
basis of $\VV_{\QQ}$), with%
\footnote{in the matrix, $i$ and $j$ run from $0$ to $r-2$%
} \begin{equation}\label{MKeqnII!!!9}\langle \xi_k,\xi_{k'} \rangle = \tiny\left(\begin{array}{cc}\scalebox{2}0 & \begin{array}{cc} & 1\\\delta_{ij}\end{array}\\\begin{array}{cc} & -\delta_{ij}\\-1\end{array} & \scalebox{2}0\end{array}\right)\normalsize  ,
\end{equation}is given by \begin{equation}\label{MKeqnII!!!10}\begin{cases}
\xi_1  = \cO_{\rXc}
\\
\xi_2  = \cO_{J_0} + \frac{1}{4} \sum_{j=1}^{r-2} (\alpha_0^j - \alpha_j^j ) \cO_{C_j} 
\\
 \mspace{70mu} - \left( \frac{13 r^{\circ} + r}{12} + \frac{1}{2}\sum_{j=1}^{r-2} (\alpha_0^j -\alpha_j^j )(d_j - \ell_j^j)\right) \cO_p
\\
\xi_{2+j} = \cO_{J_j} + \frac{1}{4} (\alpha_j^j - \alpha_0^j) \cO_{C_0} - \alpha_0^j \cO_p \;\;\;\;\; (j=1,\ldots ,r-2 )
\\
\xi_{r+1} = -\cO_{C_0}
\\
\xi_{r+j+1} = -\cO_{C_j} + \frac{1}{2}(d_j - \ell_j^j ) \cO_p \;\;\;\;\;\;\; (j=1,\ldots, r-2)
\\
\xi_{2r} = \cO_p .
\end{cases}
\end{equation}However, the basis given by the skyscraper sheaves themselves is adequate
for purposes of analyzing monodromy. We shall also have use for the
partial basis \begin{equation}\label{MKeqnIIp10!+}\begin{cases}
\hat{\xi}_{D_0} =\cO_{D_0} + \frac{1}{2} \sum_{j=1}^{r-2} \alpha_0^j C_j + \left( -\frac{15r^{\circ}+ r}{24} + \frac{1}{4} \sum_{j=1}^{r-2} \alpha_0^j (\ell_j^j +d_j )\right) \cO_p
\\
-\hat{\xi}_{C_j} = \xi_{r+j+1} \;\;\;\;\;\; (j=0,\ldots, r=2)
\\
\hat{\xi}_p =\cO_p 
\end{cases}\hspace{-1.5em}
\end{equation}later on. It satisfies
\[
ch(\hat{\xi}_{D_{0}})\cdot\hat{\Gamma}=D_{0},\;\;\; ch(\hat{\xi}_{C_{j}})\cdot\hat{\Gamma}=-C_{j},\;\;\; ch(\hat{\xi}_{p})\cdot\hat{\Gamma}=p,
\]
which implies\begin{equation}\label{MKeqnIIp10sharp}\gamma(\hat{\xi}_{D_0}) = \sigma (D_0) ,\;\;\; \gamma (\hat{\xi}_{C_j}) = -\sigma (C_j) ,\;\;\; \gamma (\hat{\xi}_p) = \sigma (p) .
\end{equation}In particular, we have\begin{equation}\label{MKeqnIIp10sharp!}\Pi^{\text{A}}_{\hat{\xi}_{C_j}} = \tau_j \;\;\;\text{and} \;\;\; \Pi^{\text{A}}_{\hat{\xi}_p} \equiv 1.
\end{equation}

In the sequel the $\ZZ$-\VHS $(\VV_{\ZZ},V_{\cO},F^{\bullet})$ constructed
above will be denoted $\cV_{\text{A}}$.

\subsection{Monodromy types}

\label{MKsecIIC}We shall compute monodromy directly on the level
of $K_{0}^{num}(\rXc)_{\QQ}$, by applying $\cO(-J_{j})\otimes$ to
the basis
\[
\cO_{\rXc};\,\cO_{J_{0}},\ldots,\cO_{J_{r-2}};\,\cO_{C_{0}},\ldots,\cO_{C_{r-2}};\,\cO_{p}.
\]
Writing $i$ resp. $k$ for the rows resp. columns of the various
blocks, this gives
\[
T_{j}=\left(\begin{array}{cccc}
1 & 0 & 0 & 0\\
-\delta_{i}^{j} & \delta_{i}^{k} & 0 & 0\\
0 & -J_{i}J_{j}J_{k} & \delta_{i}^{k} & 0\\
0 & 0 & -\delta_{j}^{k} & 1
\end{array}\right),
\]
where we note that $J_{i}J_{j}J_{k}$ is
\[
\tilde{\alpha}_{i}^{k}=\left(\begin{array}{cc}
r^{\circ} & \alpha_{0}^{k}\\
\alpha_{0}^{i} & \alpha_{i}^{k}
\end{array}\right)\;\;\;\text{resp.}\;\;\;\left(\begin{array}{cc}
\alpha_{0}^{j} & \alpha_{k}^{j}\\
\alpha_{i}^{j} & 0
\end{array}\right)
\]
if $j=0$ resp. $1,\ldots,r-2$. So
\begin{align*}
	N_{j}&=(T_{j}-I)-\tfrac{1}{2}(T_{j}-I)^{2}+\tfrac{1}{3}(T_{j}-I)^{3}\\
	&=\left(\begin{array}{cccc}
0 & 0 & 0 & 0\\
-\delta_{i}^{j} & 0 & 0 & 0\\
-\frac{1}{2}J_{j}^{2}J_{i} & -J_{i}J_{j}J_{k} & 0 & 0\\
-\frac{1}{3}J_{j}^{3} & -\frac{1}{2}J_{j}^{2}J_{k} & -\delta_{j}^{k} & 0
\end{array}\right)
\end{align*}
and the ensuing monodromy weight filtrations $W(N_{j})_{\bullet}$
are rather different in these two cases,%
\footnote{We remind the reader that if $j>0$, $J_{j}^{2}J_{i}=\alpha_{j}^{j}\delta_{i}^{0}$
and $J_{j}^{3}=0$, while if $j=0$ then $J_{0}^{2}J_{i}=\alpha_{0}^{i}$
and $J_{0}^{3}=r^{\circ}$.%
} which we denote type ``I'' resp. ``II''.

For $W_{\bullet}=W(N_{0})_{\bullet}$ (type I), we determine the following
generators for the $Gr_{\ell}^{W}$:\begin{align*}
W_0 &= \langle \cO_p \rangle \\
W_2 &= W_0 + \langle \{ \Sigma_{k=1}^{r-2} (\alpha_0^i \alpha_{i+1}^k -\alpha_0^{i+1}\alpha_i^k )\cO_{C_k} \}_{i=1}^{r-3} ,\, r^{\circ}\cO_{C_0} + \Sigma_{k=1}^{r-2} \alpha_0^k \cO_{C_k} \rangle \\
W_3 &= W_2 + \langle \cO_{D_0} ,\, \Sigma_{k,i=1}^{r-2} d_i \alpha_k^i \cO_{C_k} \rangle \\
W_4 &= W_3 + \langle \{ \alpha_0^{i+1} \cO_{J_i} - \alpha_0^i \cO_{J_{i+1}} \}_{i=1}^{r-3} ,\, \cO_{J_0} \rangle \\
W_6 &= W_4 + \langle \cO_{\rXc} \rangle .
\end{align*}The $T_{0}$-invariants $\ker(T_{0}-I)=\ker(N_{0})$ are spanned by\begin{equation}\label{MKeqnIIp12*1}\cO_p ,\, \{\cO_{C_k}\}_{k=1}^{r-2} ,\,\text{and}\;\cO_{D_0} .
\end{equation}A key point here is that because the ``$\zeta(3)$'' in $\hat{\Gamma}(\rXc)$
only appears in $\gamma(\cO_{\rXc})$, it does not appear in any $T_{0}$-invariant
A-model periods.

For $W_{\bullet}=W(N_{j})_{\bullet}$ ($j>0$), the situation bifurcates
according to whether $\alpha_{j}^{j}\neq0$ (type IIa) or $\alpha_{j}^{j}=0$
(type IIb). If $\alpha_{j}^{j}\neq0$ then we have $W_{0}=\{0\}$,\begin{align*}
W_1 & = \langle \cO_p,\, \cO_{C_0}\rangle , \\
W_3 & = W_1 + \langle \Sigma_{k=0}^{r-2} \alpha_k^j \cO_{C_k}, \, \cO_{J_j},\, \{\cO_{C_i}\}_{i\neq j,0} ,\, \{ \alpha_{i+1}^j \cO_{J_i} - \alpha_i^j \cO_{J_{i+1}} \}_{i=1}^{r-3} \rangle , \\
W_5 & = W_3 + \langle \cO_{J_0},\, \cO_{\rXc}\rangle .
\end{align*}In particular, $N_{j}$ sends $\cO_{J_{0}}\mapsto-\sum\alpha_{k}^{j}\cO_{C_{k}}\mapsto\alpha_{j}^{j}\cO_{p}$
and $\cO_{\rXc}\mapsto-\cO_{J_{j}}\mapsto\alpha_{j}^{j}\cO_{C_{0}}$.
A basis for the $T_{j}$-invariants in this case (type IIa) is%
\footnote{$\{\cO_{C_{i}}\}_{i\neq j}$ includes $\cO_{C_{0}}$%
}\begin{equation}\label{MKeqnIIp12*2}\cO_p,\, \{ \cO_{C_i}\}_{i\neq j} ,\, \text{and}\; \{ \alpha_{i+1}^j \cO_{J_i} - \alpha_i^j \cO_{J_{i+1}} \}_{i=1}^{r-3} .
\end{equation}If $\alpha_{j}^{j}=0$, let $j'\neq j,0$ be such that $\alpha_{j'}^{j}\neq0$;
this exists because $\{\alpha_{\ell}^{k}\}_{k,\ell=1}^{r-2}$ is nondegenerate.
The type IIb weight filtration is then $W_{1}=\{0\}$,\begin{align*}
W_2 &= \langle \cO_p,\, \cO_{C_0},\,\Sigma_{k=1}^{r-2} \alpha_k^j \cO_{C_k},\, \cO_{J_j}\rangle \\
W_3 &= W_2 + \langle \{\cO_{C_i}\}_{i\neq 0,j,j'},\, \{ \alpha_i^j \cO_{J_{j'}} - \alpha_{j'}^j \cO_{J_i}\}_{i\neq 0,j,j'} \rangle \\
W_4 &= W_3 + \langle \cO_{\rXc}, \,\cO_{J_0},\,\cO_{J_{j'}},\,\cO_{C_j}\rangle ,
\end{align*}with $T_{j}$-invariants\begin{equation}\label{MKeqnIIp13*3}\cO_p,\, \{ \cO_{C_i}\}_{i\neq j},\, \cO_{J_j},\, \{\alpha_i^j \cO_{J_{j'}} - \alpha_{j'}^j \cO_{J_i} \}_{i\neq 0,j,j'} .
\end{equation}

The three types of limiting \MHS $\psi_{q_{j}}\cV_{\text{A}}$ (arising
along the hyperplanes $\{q_{j}=0\}$) can be displayed pictorially
by placing a bullet in the $(p,q)$ spot if $(\psi_{q_{j}}\cV_{\text{A}})^{p,q}\neq\{0\}$
(and indicating its rank). Arrows denote the action of $N_{j}$, with
ranks of these maps indicated:\[ \includegraphics[scale=0.55]{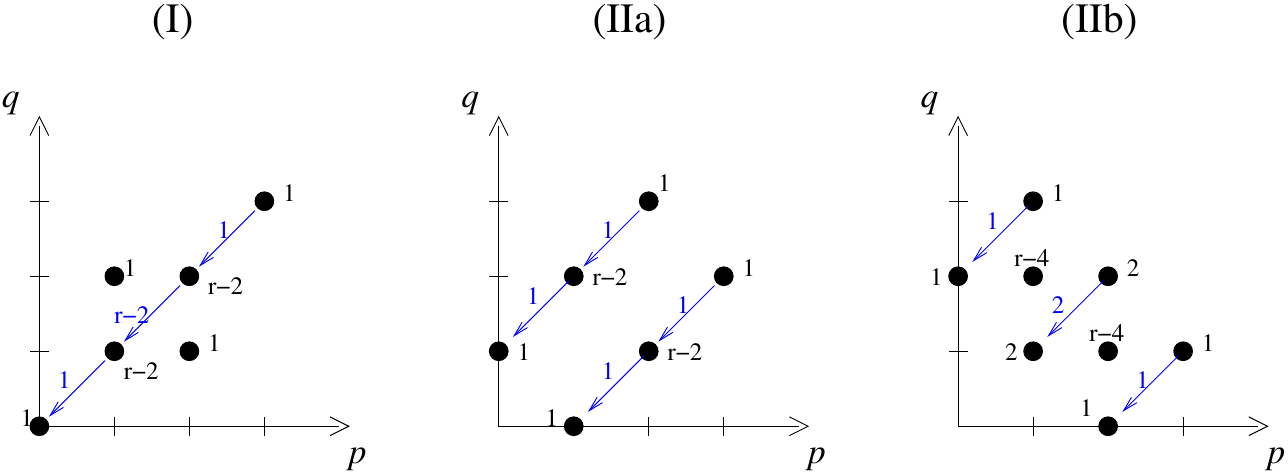} \]Note
that the space ofinvariants (i.e. $\ker(N)$) has rank $r$ for type
I but rank $2r-4$ for both types IIa and IIb. Bases forthese invariant
spaces may obviously be obtained by applying $\gamma$ to \eqref{MKeqnIIp12*1}--\eqref{MKeqnIIp13*3},
and changing basis where convenient. For type I, we find immediately
that \eqref{MKeqnIIp10sharp} is a basis for ($\kappa^{*}$ of) $\ker(N_{0})\subset\VV_{\QQ}$.
For both types II, one deduces that\begin{multline*} \sigma(p),\,\{\sigma(C_{i})\}_{i\neq j}\;(\text{incl. }\sigma(C_{0})),\;\text{and} \\ r-3\;\;\QQ\text{-linear combinations of the }\{\sigma(J_{i})\}_{i=1}^{r-2} \end{multline*}span
$\ker(N_{j})$.

\subsection{Mirror map}\label{sec:mirrormap}

\label{MKsecIID}Let $\square\subset\RR^{n}$ be a reflexive polytope,
$\mathfrak{F}=\sum_{i=0}^{m+n}b_{i}\underline{{\tt x}}^{\underline{{\tt v}}^{(i)}}$
a general $\square$-regular Laurent polynomial (with $\underline{{\tt v}}^{(0)}=\underline{0}$),
and assume none of the $\{\underline{{\tt v}}^{(i)}\}_{i=0}^{m+n}=\square\cap\ZZ^{n}$
lie on the relative interior of a \emph{facet} of $\square$. If ${\tt V}$
denotes the $\QQ$-vector space formally generated by the $\{\underline{{\tt v}}^{(i)}\}_{i=1}^{m+n}$,
let ${\tt R}:=\ker\{{\tt V}\to\QQ^{n}\}$ be the relations subspace,
with $\QQ$-basis $\{\underline{{\tt r}}^{(j)}=\sum_{i=1}^{m+n}{\tt r}_{i}^{(j)}\underline{{\tt v}}^{(i)}\}_{j=1}^{m}$,
and set (for $j=1,\ldots,m$)
\[
w_{j}:=b_{0}^{-\sum_{i}r_{i}^{(j)}}\prod_{i}b_{i}^{{\tt r}_{i}^{(j)}}=\prod_{i}\left(\tfrac{b_{i}}{b_{0}}\right)^{{\tt r}_{i}^{(j)}}.
\]

Write $\cX\subset\PP_{\square}$ for the zero locus of $\mathfrak{F}$
and $\cX^{\circ}\subset\PP_{\square^{\circ}}$ for a general anticanonical
hypersurface. We have the exact sequence
\[
0\to(\QQ^{n})^{\vee}\to{\tt V}^{\vee}\to{\tt R}^{\vee}\to0
\]
where ${\tt R}^{\vee}\cong H^{2}(\cX^{\circ},\QQ)$. A basis of ${\tt V}^{\vee}$
is given by the divisors $\mathcal{D}_{i}\subset\cX^{\circ}$ dual
to the $\underline{{\tt v}}^{(i)}$. Choose $\{\beta_{\ell}^{k}\}\in\QQ^{m(m+n)}$
such that $\sum_{k=1}^{m+n}\beta_{\ell}^{k}{\tt r}_{k}^{(j)}=\delta_{\ell}^{j}$
($\ell,j=1,\ldots,m$), and put
\[
\mathcal{J}_{\ell}:=\sum_{k=1}^{m+n}\beta_{\ell}^{k}[\mathcal{D}_{k}]\in H^{2}(\cX^{\circ}).
\]
This gives a basis dual to $\{\underline{{\tt r}}^{(j)}\}$, since
\[
\underline{{\tt r}}^{(j)}(\mathcal{J}_{\ell})=\sum_{i,k}\beta_{\ell}^{k}{\tt r}_{i}^{(j)}\underline{{\tt v}}^{(i)}(\mathcal{D}_{k})=\sum_{i,k}\beta_{\ell}^{k}{\tt r}_{i}^{(j)}\delta_{k}^{i}=\sum_{k}\beta_{\ell}^{k}{\tt r}_{k}^{(j)}=\delta_{\ell}^{j}.
\]

Now the mirror map sends the complex structure parameter $\underline{b}$
of $\cX$ to a K\"ahler parameter in $H^{2}(\cX^{\circ},\CC)$, of
the form $\tau(\underline{w})=$
\[
=\sum_{j=1}^{m}\tau_{j}(\uw)\mathcal{J}_{j}=\frac{1}{2\pi \ay}\sum_{i=1}^{m+n}\log\left(\tfrac{b_{i}}{b_{0}}\right)[\mathcal{D}_{i}]+\cO\left(\left\{ \tfrac{b_{i}}{b_{0}}\right\} \right),
\]
where $\tau_{j}(\underline{b})$ are (B-model) periods. We compute
\begin{align*}
	\underline{{\tt r}}^{(j)}\left(\sum_{i}\log\left(\tfrac{b_{i}}{b_{0}}\right)[\mathcal{D}_{i}]\right)&=\sum_{i,k}{\tt r}_{k}^{(j)}\log\left(\tfrac{b_{i}}{b_{0}}\right)\underline{{\tt v}}^{(k)}(\mathcal{D}_{i})\\
	&=\sum_{i}{\tt r}_{i}^{(j)}\log\left(\tfrac{b_{i}}{b_{0}}\right)=\log(w_{j}),
\end{align*}
which shows $\tau_{j}(\uw)\sim\frac{1}{2\pi\ay}\log(w_{j})$.

Applying this to our situation ($n=4$, $m=r-1$), with $\alpha,\beta,\gamma,\{a_{i}\}$
replacing the $\{b_{i}\}$, with $D',D'',D_{0},\ldots,D_{r}$ replacing
the $\{\mathcal{D}_{i}\}$, and with $z_{0},\ldots,\allowbreak z_{r-2}$ replacing
$w_{1},\ldots,w_{m}$, we recover \eqref{MKeqnI1/2} and \eqref{MKeqnII3.5},
and find that the coefficients $\{\tau_{j}(\uz)\}$ of the $\{J_{j}\}$
in $\tau(\uz)$ are asymptotic to $\frac{1}{2\pi\ay}\log(z_{j})$.
By \S\ref{MKsecIG} (especially \eqref{MKeqnIp22ss}--\eqref{MKeqnInewp16}),
the mirror map is therefore exactly\begin{equation}\label{MKeqnIIp15!!}\tau(\uz) = \sum_{j=0}^{r-2} \tau_j (\uz) J_j = \Pi_{\tau_{\beta}}^{\text{B}}(\uz) J_0 + \sum_{j=1}^{r-2} \Pi_{\cM_{\alpha}(\vf_0^{(j)})}^{\text{B}}(\uz) J_j .
\end{equation}Writing $\mathcal{Q}(z_{0},\ldots,z_{r-2}):=(q_{0}(\uz),\ldots,q_{r-2}(\uz))$,
we note that the B-model coordinate axes $z_{j}=0$ map to the A-model
axes $q_{j}=0$.

Now \cite[Thm. 5.9]{Iritani} provides an isomorphism $\Theta:\,\mathcal{Q}^{*}\cV_{\text{A}}\overset{\cong}{\to}\cV_{\text{B}}$
of $\ZZ$-\linebreak \VHS sending $1_{\rXc}\mapsto[\tilde{\Omega}]$. Since \eqref{MKeqnIIp10sharp!}
and \eqref{MKeqnIIp15!!} identify the periods 
\[
\Pi_{\cT_{\beta}}^{\text{B}}(\uz)\equiv\Pi_{\hat{\xi}_{C_{0}}}^{\text{A}}(\mathcal{Q}(\uz))\;\;\;\text{and}\;\;\;\Pi_{\cM_{\alpha}(\vf_{0}^{(j)})}^{\text{B}}(\uz)\equiv\Pi_{\hat{\xi}_{C_{j}}}^{\text{A}}(\mathcal{Q}(\uz))
\]
 modulo $\ZZ$, and obviously $\Pi_{\cT_{\alpha}}^{\text{B}}(\uz)=1=\Pi_{\hat{\xi}_{p}}^{\text{A}}(\mathcal{Q}(\uz))$,
we deduce that (up to changing $\cT_{\beta}$ and $\cM_{\alpha}(\vf_{0}^{(j)})$
by integer multiples of $\cT_{\alpha}$)
\[
\Theta(\sigma(p))=\cT_{\alpha},\;\;\;\Theta(-\sigma(C_{0}))=\cT_{\beta},\;\;\;\Theta(-\sigma(C_{j}))=\cM_{\alpha}(\vf_{0}^{(j)}).
\]
By considering $W(N_{0})_{\bullet}$ on $\ker(N_{0})$ on the A and
B sides (cf. \S\ref{MKsecIF} and \S\ref{MKsecIIC}), we find in addition
that (after modifying $\vf_{1}$ by $\ZZ\langle\{\vf_{0}^{(i)}\}\rangle$
and $\cM_{\alpha}(\vf_{1})$ by $\ZZ\langle\cT_{\alpha}\rangle$ if
necessary)\begin{equation}\label{MKeqnIIp16}\Theta (\sigma (D_0)) = \Theta (\gamma (\hat{\xi}_{D_0})) = \cM_{\alpha}(\vf_1) .
\end{equation}(More precisely, if we look at $W(N_{0})_{3}\cap\ker(T_{0}-I)$ in
$H^{3}(\rX,\ZZ)\,(\subset\cV_{\text{B}})$ resp. $\VV_{\ZZ}\,(\subset\cV_{\text{A}})$,
this is generated by $\cM_{\alpha}(\vf_{1})$ mod $\ZZ\langle\cT_{\alpha},\{\cM_{\alpha}(\vf_{0}^{(i)})\}\rangle$
resp. $\gamma(\hat{\xi}_{D_{0}})$ mod $\ZZ\langle\gamma(\hat{\xi}_{p}),\{\gamma(\hat{\xi}_{C_{i}})\}_{i=1}^{r-2}\rangle$.)
Heuristically, we obviously have some matching as well between the
$\{\cM_{\beta}(\vf_{0}^{(i)})\}$ and $\{\gamma(\cO_{J_{j}})\}$,
and between $\cM_{\beta}(\vf_{1})$ and $\gamma(\cO_{\rXc})$; but
we will not dissect this further, as \eqref{MKeqnIIp16} shall now
yield the local mirror symmetry identity we seek.

Recalling that $\hat{k}=(k_{0},\underline{k})$, write $\tilde{N}_{\underline{k}}$
for $\tilde{N}_{\ukh}$ when $k_{0}=0$; and referring to \S\ref{MKsecIG},
write $R_{1}:=R_{\vf_{1}}(\uz)$ resp. $R_{0}^{(i)}:=R_{\vf_{0}^{(i)}}(\uz)$,
where $\uz=(z_{1},\ldots,z_{r-2})$ omits $z_{0}$. Accordingly, we
shall change notation for $(z_{0},\ldots,z_{r-2})$ to $\hat{\uz}=(z_{0},\uz)$.
Define \emph{local} K\"ahler parameters $Q_{i}:=e^{R_{0}^{(i)}}$
(for $i=1,\ldots,r-2)$.
\begin{thm}\phantomsection\label{MKthm2}
On the universal cover of $(\Delta^{*})^{r-2}$ we
have \begin{equation}\label{MKeqnIIp17*!!} (2\pi\ay) R_1 = \frac{1}{2}\sum_{i,j=1}^{r-2}\alpha^i_j R_0^{(i)} R_0^{(j)} - \sum_{\underline{k} \neq \underline{0}} (\Sigma_{i=1}^{r-2} d_i k_i ) \tilde{N}_{\underline{k}} \underline{Q}^{\underline{k}} . \end{equation}\end{thm}
\begin{rem}
This is Conjecture 5.1 in \cite{DoranKerr}; also see \cite{CKYZ,Hosono}.\end{rem}
\begin{proof}
Taking periods of \eqref{MKeqnIIp16} on both sides (with respect
to $\tilde{\Omega}$ resp. $1_{\rXc}$) yields
\[
\Pi_{\cM_{\alpha}(\vf_{1})}^{\text{B}}(\hat{\uz})=\Pi_{\hat{\xi}_{D_{0}}}^{\text{A}}(\mathcal{Q}(\hat{\uz}))=\langle1_{\rXc},\sigma(D_{0})\rangle(\mathcal{Q}(\hat{\uz})).
\]
By \eqref{MKeqnII4}, 
\[
\sigma(D_{0})=\sigma(J_{0})-\sum_{i=1}^{r-2}d_{i}\sigma(J_{i})
\]
which by \eqref{MKeqnII12}
\[
=D_{0}-\frac{1}{(2\pi\ay)^{3}}\sum_{j}\left(\Phi_{0j}''-\Sigma_{i}d_{i}\Phi_{ij}''\right)C_{j}+\frac{1}{(2\pi\ay)^{3}}\left(\Phi_{0}'-\Sigma_{i}d_{i}\Phi_{i}'\right)p.
\]
Writing $\d_{D_{0}}:=\d_{0}-\sum d_{i}\d_{i}$, the A-model period
is then 
\[
\langle1_{\rXc},\sigma(D_{0})\rangle=\frac{1}{(2\pi\ay)^{3}}\d_{D_{0}}\Phi=\d_{D_{0}}\left(\Sigma_{j=0}^{r-2}\tau_{j}J_{j}\right)^{3}+\frac{1}{(2\pi\ay)^{3}}\d_{D_{0}}\sum_{\ukh\neq0}\tilde{N}_{\ukh}\tilde{\underline{Q}}^{\ukh}.
\]
For the first term, $\d_{D_{0}}$ of 
\[
\frac{r^{\circ}}{6}\tau_{0}^{3}+\frac{1}{2}\tau_{0}^{2}\sum_{j=1}^{r-2}\alpha_{0}^{j}\tau_{j}+\frac{1}{2}\tau_{0}\sum_{i,j=1}^{r-2}\alpha_{j}^{i}\tau_{i}\tau_{j}
\]
is
\[
\frac{1}{2}\left(r^{\circ}-\Sigma_{i}d_{i}\alpha_{0}^{i}\right)\tau_{0}^{2}+\tau_{0}\sum_{j=1}^{r-2}\left(\alpha_{0}^{j}-\Sigma_{i}d_{i}\alpha_{j}^{i}\right)\tau_{j}+\frac{1}{2}\sum_{i,j=1}^{r-2}\alpha_{j}^{i}\tau_{i}\tau_{j}=\frac{1}{2}\sum_{i,j}\alpha_{j}^{i}\tau_{i}\tau_{j}
\]
by \eqref{MKeqnII9}; for the second we have $\frac{1}{(2\pi\ay)^{2}}\sum_{\ukh\neq0}(k_{0}-\sum d_{i}k_{i})\tilde{N}_{\ukh}\underline{q}^{\ukh}.$
Pulling back by $\mathcal{Q}$ therefore gives (as multivalued functions
of $\hat{\underline{z}}$)\begin{equation}\label{MKeqnIIp18!*}\Pi^{\text{B}}_{\cM_{\alpha}(\vf_1)} =\frac{1}{2} \sum_{i,j} \alpha_j^i \Pi^{\text{B}}_{\cM_{\alpha}(\vf_0^{(i)})} \Pi^{\text{B}}_{\cM_{\alpha}(\vf_0^{(j)})} +\frac{1}{(2\pi \ay)^2} \sum_{\ukh\neq \underline{0}} (k_0 - \Sigma d_i k_i ) \tilde{N}_{\ukh} \underline{q}(\hat{\underline{z}})^{\ukh} ,
\end{equation}where $q_{j}(\hat{\underline{z}})=e^{2\pi\ay\prod_{\cM_{\alpha}(\vf_{0}^{(i)})}^{\text{B}}}$
and $q_{0}(\hat{\underline{z}})=e^{2\pi\ay\prod_{\cT_{\beta}}^{\text{B}}}\sim z_{0}$.

Now we pass to the limit $z_{0}\to0$, where \eqref{MKeqnIIp18!*}
essentially becomes an equality of extension classes of A- and B-model
\LMHS\unskip. (In particular, the limit on both sides is finite since these
are periods of $T_{0}$-invariant cycles; this is also clear from
the absence of $\tau_{0}=\Pi_{\cT_{\beta}}^{\text{B}}$ in any term.)
Since $\lim_{z_{0}\to0}q_{0}(\hat{\underline{z}})=0$, the $\sum_{\ukh}$
becomes a $\sum_{\underline{k}}$, while by \S\ref{MKsecIG}
\[
\lim_{z_{0}\to0}\Pi_{\cM_{\alpha}(\vf_{1})}^{\text{B}}(z_{0},\underline{z})=\frac{1}{2\pi\ay}R_{1}(\underline{z})\,,\;\;\;\lim_{z_{0}\to0}\Pi_{\cM_{\alpha}(\vf_{0}^{(i)})}^{\text{B}}(z_{0},\underline{z})=\frac{1}{2\pi\ay}R_{0}^{(i)}(\underline{z})\,,
\]
hence $\lim_{z_{0}\to0}q_{i}(z_{0},\underline{z})=Q_{i}(\underline{z})$.
So $(2\pi\ay)^2 \cdot \eqref{MKeqnIIp18!*}|_{z_0=0} $ indeed yields
our main result \eqref{MKeqnIIp17*!!}.
\end{proof}
The \GW invariants $\tilde{N}_{\underline{k}}$ ``counting''%
\footnote{These are rational and possibly negative numbers, so only ``count''
anything in the sense of excess intersection number.%
} genus-0 curves of class\linebreak $\sum_{i=1}^{r-2}k_{i}[C_{i}]$ on $\rXc$,
may also be interpreted as local \GW invariants of $D_{0}\cong\PP_{\DDelta}$,
or equivalently as (usual) \GW invariants of the 3-fold $\PP(\cO\oplus K_{\PP_{\DDelta}})$.
With this interpretation, the right-hand-side of \eqref{MKeqnIIp17*!!} (perhaps
replacing $R_{0}^{(i)}$ by $(2\pi\ay)\tau_{i}$) is the local \GW
prepotential $\Phi_{loc}$ of $\PP_{\DDelta}$.

\subsection{The sunset case}

\label{MKsecIIE}Specializing to the diagram \[\includegraphics[scale=0.50]{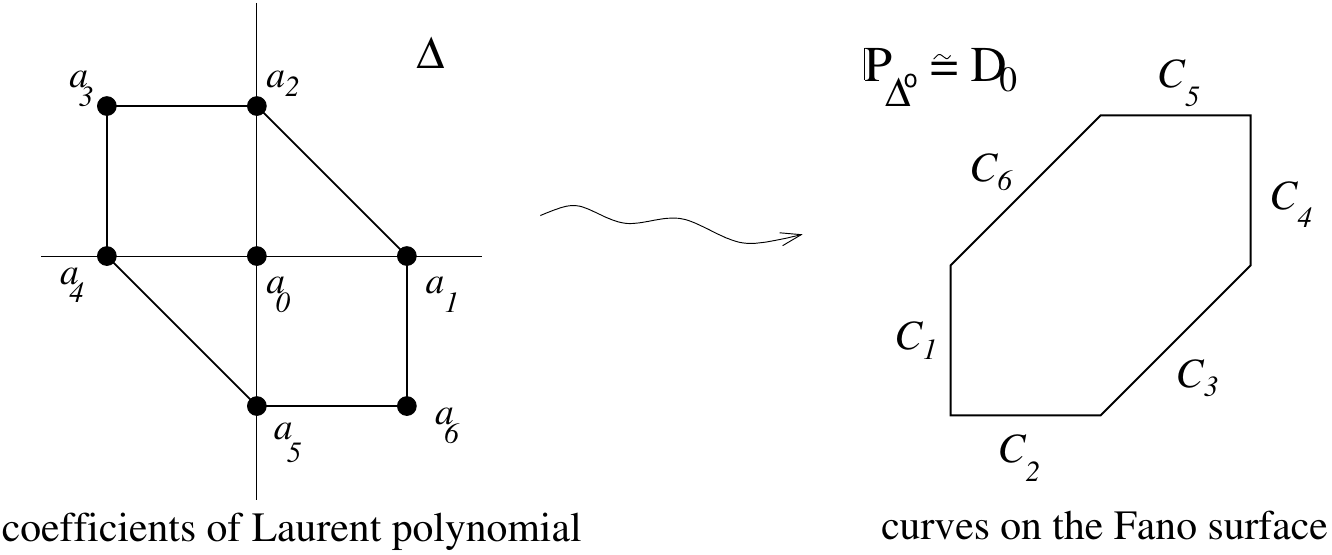}\]we
have $r=r^{\circ}=\nu=6$, $d_{i}=1$, and
\[
\ell_{j}^{i}= \begin{cases}
-1, & i=j\\
1, & i\underset{(6)}{\equiv}j\pm1\\
0, & \text{otherwise}.
\end{cases}
\]
From this we deduce that
\begin{align*}
	J_{1}&=-D_{1}+D_{3}+D_{4}\,(=D_{6})\,,\quad J_{2}=D_{3}+D_{4}\,,\\
	J_{3}&=D_{1}+D_{2}\,,\quad  J_{4}=D_{1}+D_{2}-D_{4}\,(=D_{5})\,,
\end{align*}
and (using \eqref{MKeqnII9}) that \[ \tilde{\alpha}_j^i = \small \left( \begin{array}{c|cccc} 6 & 1 & 2 & 2 & 1 \\ \hline 1 & -1 & 0 & 1 & 1 \\ 2 & 0 & 0 & 1 & 1 \\ 2 & 1 & 1 & 0 & 0 \\ 1 & 1 & 1 & 0 & -1 \end{array} \right) \normalsize . \]In
\eqref{MKeqnII13} and \eqref{MKeqnIIp10!+} we have for example\begin{align*}
\hat{\Gamma}(\rXc)&=1+\frac{1}{2}\sum C_{i}+3C_{0}+\frac{360\zeta(3)}{(2\pi\ay)^{3}}p\,,
\\
\hat{\xi}_{D_{0}}&=\cO_{D_{0}}+\tfrac{1}{2}\cO_{C_{1}}+\cO_{C_{2}}+\cO_{C_{3}}+\tfrac{1}{2}\cO_{C_{4}}-4\cO_{p}\,,
\\
-\hat{\xi}_{C_{j}}&=\cO_{C_{j}}-\cO_{p}=\cO_{C_{j}}(-1)\,;
\end{align*}notice that $\hat{\xi}_{D_{0}}$ is not quite integral. One easily
reads off the $N_{j}$ from $\tilde{\alpha}_{j}^{i}$: \[
\begin{tabular}{c|c|c}
& type & invariants \\
\hline
$N_0$ & I & $\cO_p,\, \{\cO_{C_k}\}_{k\neq 0} ,\, \cO_{D_0} $ \\
$N_1$ & IIa & $\cO_p ,\, \{\cO_{C_k}\}_{k\neq 1},\, \cO_{J_1}+\cO_{J_3},\,\cO_{J_1}+\cO_{J_4},\, \cO_{J_2}$ \\
$N_2$ & IIb & $\cO_p,\, \{ \cO_{C_k}\}_{k\neq 2},\, \cO_{J_1},\,\cO_{J_2},\,\cO_{J_3} - \cO_{J_4}$ \\
$N_3$ & IIb & $\cO_p,\, \{ \cO_{C_k}\}_{k\neq 3},\, \cO_{J_1}-\cO_{J_2},\,\cO_{J_3},\,\cO_{J_4}$ \\
$N_4$ & IIa & $\cO_p,\, \{ \cO_{C_k}\}_{k\neq 4},\, \cO_{J_1}+\cO_{J_4},\,\cO_{J_2}+\cO_{J_4},\,\cO_{J_3}$ \\
\end{tabular}
\]where $k$ runs from $0$ to $4$.

This is in some sense incomplete, as the nonsimplicial nature of the
Mori cone $\RR_{\geq0}\langle C_{1},\ldots,C_{6}\rangle\subset H_{2}(\PP_{\DDelta},\RR)$
(and the dual ``K\"ahler'' cone in $H^{2}(\PP_{\DDelta},\RR)$)
forces us to use all $6$ $\{z_{i}\}$ to parametrize the singular
4-\linebreak dimensional domain of the B-model VHS, as described in \S\ref{MKsecIF}.
But this will not matter as we presently restrict to the Feynman locus,
where $z_{i}=z_{i+3}$ ($i=1,2,3$) and $R_{0}^{(1)}=R_{0}^{(4)}$
(cf. \eqref{MKeqnIp25!}--\eqref{MKeqnIp25*}), so that the mirror
map zends $\uz\mapsto R_{0}^{(1)}(J_{1}+J_{4})+R_{0}^{(2)}J_{2}+R_{0}^{(3)}J_{3}$.
This specialization therefore replaces K\"ahler by the 3-dimensional
simplicial ``slice'' $\RR_{\geq0}\langle J_{1}+J_{4},\, J_{2},\, J_{3}\rangle=\RR_{\geq0}\langle D_{2}+D_{3},\, D_{3}+D_{4},\, D_{1}+D_{2}\rangle,$
and Mori by the 3-dimensional simplicial quotient $\RR_{\geq0}\langle\overline{C_{1}},\,\overline{C_{2}},\,\overline{C_{3}}\rangle$
in $\overline{H_{2}}:=H_{2}(\PP_{\DDelta})/\langle C_{1}-C_{4}\rangle$.
(Note that $\overline{C_{1}}\equiv\overline{C_{4}}$ $\implies$ $\overline{C_{2}}=\overline{C_{5}}$
and $\overline{C_{3}}=\overline{C_{6}}$; and that working modulo
this equivalence, $\gamma(\hat{\xi}_{D_{0}})$ becomes integral.)
It also replaces $N_{1}$ and $N_{4}$ in the table by their sum $N_{1}+N_{4}$,
which we compute to be (like $N_{2}$ and $N_{3}$) of type IIb, with
invariants
\[
\cO_{p},\,\cO_{C_{0}},\,\cO_{C_{1}}-\cO_{C_{4}},\,\cO_{C_{2}},\,\cO_{C_{3}},\,\cO_{J_{1}},\,\cO_{J_{2}}-\cO_{J_{3}},\,\cO_{J_{4}}.
\]

We shall also have to define local \GW invariants for classes $\ell_{1}\overline{C_{1}}+\ell_{2}\overline{C_{2}}+\ell_{3}\overline{C_{3}}\in\overline{H_{2}}$,
writing\begin{equation}\label{MKeqnIIp21a}N_{\ul}:= \sum_{k_1 + k_4 =\ell_1} \tilde{N}_{k_1,\ell_2,\ell_3,k_4} \in \QQ .
\end{equation}Now the statement of Theorem \ref{MKthm2} for the sunset reads $(2\pi\ay)R_{1}=$
\[
=\left(R_{0}^{(1)}+R_{0}^{(2)}\right)\left(R_{0}^{(3)}+R_{0}^{(4)}\right)-\frac{1}{2}\left(R_{0}^{(1)}\right)^{2}-\frac{1}{2}\left(R_{0}^{(4)}\right)^{2}-\sum_{\underline{k}\neq\underline{0}}|\underline{k}|\tilde{N}_{\underline{k}}\underline{Q}^{\underline{k}},
\]
where $|\underline{k}|:=\sum_{i=1}^{4}k_{i}$. The Feynman specialization
gives $R_{0}^{(1)}=R_{0}^{(4)}$ and $Q_{1}=Q_{4}$, and so writing
$\underline{Q}^{\ul}=Q_{1}^{\ell_{1}}Q_{2}^{\ell_{2}}Q^{\ell_{3}}$
and $|\ul|=\sum_{i=1}^{3}\ell_{i}$, we have the
\begin{cor} \phantomsection\label{MKcor2}
On the Feynman locus $(\cong(\Delta^{*})^{3})$ parametrizing
the general-mass sunset family, we
have\begin{equation}\label{MKeqnIIp21b}(2\pi\ay)R_1 =
  R_0^{(1)}R_0^{(2)} + R_0^{(2)}R_0^{(3)} + R_0^{(1)}R_0^{(3)} -
  \sum_{\ul \in \IN^3\backslash \underline{0}} |\ul| N_{\ul} \underline{Q}^{\ul} .
\end{equation}
A computation of the local Gromov-Witten invariant is given in section~\ref{sec:localGW}.
\end{cor}

\section{The multiparameter sunset integral}\label{sec:multi}

In this section we use regulators (see \S\ref{MKsecIG}) to derive
the inhomogeneous Picard-Fuchs equation (Prop. \ref{MKpropIII1}) for the sunset
integral, and also to relate it to the elliptic dilogarithm (Remark
\ref{MKremIII2}). This analysis complements the derivation of
  the Picard-Fuchs equation given in section~\ref{sec:PFderivation}
  and the evaluation of the sunset integral in section~\ref{sec:elliptic-dilogarithm}. Using Corollary \ref{MKcor2}, we are able to derive
an expression for the integral in terms of the local \GW numbers, and
to compute these numbers (Prop. \ref{MKpropIII2}ff).

\subsection{Degeneration of the Yukawa coupling}

Let $\cB$ denote the symplectic basis for the B-model $\QQ$-local
system given by applying $\Theta\circ\gamma$ to \eqref{MKeqnII!!!10}.
According to $\S\S$\ref{MKsecIIB},\ref{MKsecIID} (esp. \eqref{MKeqnII12})
we find 
\begin{align}\label{MKeqnIII1}{}^{t}[\Omega]_{\cB} =  \Bigg(&1,\tau_{0},\ldots,\tau_{r-2},\frac{\Phi_{0}'}{(2\pi \ay)^{3}}+\mathcal{O}(\underline{\tau}),\ldots,\frac{\Phi_{r-2}'}{(2\pi\ay)^{3}}+\mathcal{O}(\underline{\tau}),\\ &\frac{1}{(2\pi\ay)^{3}}\left\{ 2\Phi-\sum_{\ell=0}^{r-2}\tau_{\ell}\Phi_{\ell}'\right\} +\mathcal{O}(\underline{\tau})\Bigg)\notag
\end{align}There are two ways to define Yukawa coupling: while (with $\delta_z:=z\partial_z$)\begin{equation}\label{MKeqnIII2}\tilde{Y}_{ijk}:=\int_{\br{X}}\tilde{\Omega}\wedge\nabla^3_{\delta_{z_i}\delta_{z_j}\delta_{z_k}}\tilde{\Omega},
\end{equation} makes sense ``globally'' (in $z_{0,},\ldots,z_{k})$, we consider
instead (referring to~\eqref{MKeqnIIp15!!} for $\tau(z)$) 
\begin{equation}\label{MKeqnIII3}Y_{ijk} := \int_{\br{X}}\Omega \wedge \nabla^3_{\d_{\tau_i}\d_{\tau_j}\d_{\tau_k}}\Omega ,
\end{equation}
which is defined ``locally'' about the large complex structure limit
(in $q_{0},\ldots,\allowbreak q_{k}$). Since ${[Q]}_{\cB}$ is given by \eqref{MKeqnII!!!9},
\eqref{MKeqnIII3} is easily computed to be\begin{equation}\label{MKeqnIII4}= {}^t[\Omega]_{\cB}[Q]_{\cB}[\Omega]_{\cB} = \tfrac{1}{(2\pi\ay)^3} \Phi^{(3)}_{ijk}.
\end{equation}Motivated by the fact that the unique combination of \emph{first}
derivatives of $\Phi$ remaining finite in the $q_{0}\to0$ ($z_{0}\to0$)
limit is $\Phi_{0}'-\sum_{i=1}^{r-2}d_{i}\Phi_{i}'$ (see the proof
of Theorem \ref{MKthm2}), we look at \begin{align*} Y_{jk}^{\text{loc}} & := \lim_{q_{0}\to0}\left(Y_{0jk}-\sum_{i=1}^{r-2}d_{i}Y_{ijk}\right) \\&=\alpha_{k}^{j}-\sum_{\underline{\kappa}\neq0}\tilde{N}_{\underline{\kappa}}\left(\sum_{i=1}^{r-2}d_{i}\kappa_{i}\right)\kappa_{j}\kappa_{k}\underline{Q}^{\underline{\kappa}}\\&=\tfrac{1}{(2\pi\ay)^{2}}\Phi_{\text{loc},jk}''\;=\;\tfrac{1}{2\pi\ay}\partial_{R_{0}^{(j)}R_{0}^{(k)}}^{2}R_{1}. \end{align*}To
relate these to a Yukawa coupling on the elliptic curve family $\{E_{\underline{a}}\}$,
write (cf. \eqref{MKeqnI0},\eqref{MKeqnI0!})\begin{equation}\label{MKeqnIII5}\omega_{\underline{a}} := \frac{1}{2\pi\ay} Res_{E_{\underline{a}}}\left( \frac{\tfrac{dx}{x}\wedge \tfrac{dy}{y}}{f_{\underline{a}}(x,y)}\right) \in \Omega^1(E_{\underline{a}}) ,
\end{equation}and $\pi_{0}=\int_{\varphi_{0}}\omega$, $\pi_{1}=\int_{\varphi_{1}}\omega$.
Now pass to the ``diagonal slice'' subfamily of \cite[\S5.4]{DoranKerr},
specializing $f_{\underline{a}}$ to $1-s\phis$ where $\phis(x,y)$
is a specific tempered $\Delta$-regular Laurent polynomial; by \cite[\S5.4]{DoranKerr} we have $z_{i}(s)/s^{d_{i}}$ a root of $1$ ($\forall i$)
and $R_{0}^{(1)}/d_{1}\equiv\cdots\equiv R_{0}^{(r-2)}/d_{r-2}\equiv:R_{0}$
mod $\QQ(1)$. Moreover, one has $\delta_{s}R_{i}=\pi_{i}$ ($i=1,2$),
and an easy computation reveals that
\[
2\pi\ay\sum_{j,k}d_{j}d_{k}Y_{jk}^{\text{loc}}|_{\Delta}=\partial_{R_{0}}^{2}R_{1}=\frac{\cY^{E}}{\pi_{0}^{3}},
\]
where 
\[
\cY^{E}(s):=\int_{E_{s}}\omega_{s}\wedge\nabla_{\delta_{s}}\omega_{s}=\pi_{0}\delta_{s}\pi_{1}-\pi_{1}\delta_{s}\pi_{0}
\]
is the Yukawa coupling for $\{E_{s}\}:=\{E_{\{\underline{a}(s)}\}$.
\begin{rem}
In general, if $\br{X}$ is replaced by a family of elliptically-fibered
Calabi-Yau $(n+1)$-folds, and $E$ by a family of $(n-1)$-dimensional Calabi-Yaus
$W$ with rank $n$ Picard-Fuchs equation along $\Delta$, a heuristic
Hodge-theoretic argument shows that a ($z_{0}\to0$) limit of Yukawa
couplings for $\br{X}$ yields $\cY^{W}/\pi_{0}^{n+1}$ along $\Delta$.
\end{rem}
For the rest of this section, we specialize to the sunset case. However,
to treat the three-mass situation, we shall need to consider ``arbitrary
slices'' of the Feynman locus, given by (the vanishing of)
\begin{align*}
	f_{\underline{a}(s;\underline{\xi})}(x,y)&:=f_s^{\su}(x,y):=1-s\phis(x,y),\\
	\phis(x,y)&:=(1-x^{-1}-y^{-1})(\xi_{3}^{2}-\xi_{2}^{2}x-\xi_{1}^{2}y).
\end{align*}
(Note that $\phis$ is no longer tempered.) 
We write $\cE^{\ux}\overset{\varepsilon}{\to}\PP_{s}^{1}$
for the family with fibers 
\[
E_{s}^{\su}=\overline{\{f_s^{\su}(x,y)=0\}}\subset\PP_{\Delta},
\]
and $\omega_{s}:=\omega_{\ua}(s;\ux)$ (cf. \eqref{MKeqnIII5}) for
the section of $\varepsilon_{*}\omega_{\cE/\PP^{1}}\cong\cO(1)$ with
a simple zero at $s=\infty$. Note that this family is semistable.

The Yukawa coupling (with $\delta_s:=s\partial_s$) \begin{equation}\label{MKeqnIII5.5}Y_\su(s):=
  2\pi\ay \int_{E_s^{\su}} \omega_s \wedge \nabla_{\delta_s} \omega_s
  \in \CC(\PP^1)^* \cong \CC(s)^* 
\end{equation}can be determined up to scale by the properties:
\begin{itemize}
\item $Y_{\su}$ has a double zero at $\infty$;
\item $Y_{\su}(0)\in\CC^{*}$;
\item at other singular fibers, $Y_{\su}(s)$ has a simple pole;
\item $Y_{\su}(s)$ has a zero of order $m-1$ at branch points of order $m$
for the $J$-invariant $J(s)$;
and
\item $Y_{\su}(0)=6$, by \eqref{MKeqnIII20} below.
\end{itemize}
This yields the function \begin{equation}\label{MKeqnIII6}Y_{\su}(s)=\frac{2 \mu_1 \mu_2 \mu_3 \mu_4 s^2 - 4(\xi_1^2 + \xi_2^2 + \xi_3^2 )s + 6}{\prod_{i=1}^4 (1-\mu_i^2 s)} ,
\end{equation}where $\mu_{1}=-\xi_{1}+\xi_{2}+\xi_{3}$, $\mu_{2}=\xi_{1}-\xi_{2}+\xi_{3}$,
$\mu_{3}=\xi_{1}+\xi_{2}-\xi_{3}$, $\mu_{4}=\xi_{1}+\linebreak \xi_{2}+\xi_{3}$.
This of course reproduces the expression for the Yukawa coupling
in~\eqref{e:Yu} derived in section~\ref{sec:PFderivation}.

We shall use this below to compute the local \GW invariants $N_{\ul}$,
for simplicity of notation suppressing most ``$\su$'' subscripts
in what follows.

\subsection{Inhomogeneous equation for the sunset integral}\label{sec:sunsetBmodel}

Continuing an analysis of the $1$-parameter family $\cE_\su\overset{\varepsilon}{\to}\PP^{1}$,
we write as usual $\{\{\vf_{0}^{(i)}\}_{i=1}^{6},\vf_{1}\}\subset\cK_{E}$
(cf. $\S$\ref{MKsecID}), and recall that on the Feynman locus, $\{\{\vf_{0}^{(i)}\}_{i=1}^{3},\vf_{1}\}$
furnish a basis for $\bar{\cK}_{E}$ (cf. $\S$\ref{MKsecIF}). For
the holomorphic period (about $s=0$), the usual residue computation
yields\begin{align} \label{MKeqnIII7}  \pi_{0}  &=
    \int_{\vf_{0}}\omega=\int_{\vf_{0}^{(i)}}\omega\;\;(i=1,2,3)\\\notag  &=
    \sum_{m\geq0}s^{m}\left(\sum_{|\underline{b}|=m}\ux^{\underline{b}}{m
        \choose \underline{b}}^{2}\right)  =:
    \sum_{m\geq0}s^{m}\beta_{m}, \end{align}
where ${m
        \choose \underline{b}}= {m!
        \over b_1!b_2!b_3!}$ and the coefficients $\beta_m$ are
      generalized Ap\'ery numbers.
Writing
$R=\tfrac{1}{2\pi\ay}R\{x,y\}=\tfrac{1}{2\pi\ay}\log(x)\tfrac{dy}{y}-\log(y)\delta_{T_{x}}$
for the regulator current on $E_{s}^{*}$, \eqref{MKeqnIp25!!} gives
for $i=1,2,3$ \begin{equation}\label{MKeqnIII8}R_0^{(i)} = \int_{\vf_0^{(i)}} R = \log \left( \tfrac{-\xi_i^2 s}{1-s\sum\xi_i^2}\right) + \cH(s),
\end{equation}where $\cH$ is holomorphic (about $s=0$) and vanishes at $s=0$.
Write $\cL_{i}:=2\log(\xi_{i})$.

Interpreted as a 1-current on $\cE\backslash E_{0}$, $R$ has coboundary
\begin{equation}\label{MKeqnIII9}d[R] = \tfrac{1}{2\pi\ay} \tfrac{dx}{x} \wedge \tfrac{dy}{y} - (2\pi \ay)\delta_{T_x\cap T_y} - \sum_{i=1}^3 \log\left(\tfrac{m_i^2}{m_{i-1}^2}\right)\delta_{q_i \times \PP^1\setminus \{ 0\}} ,
\end{equation}where ${\tt q}_{1},{\tt q}_{2},{\tt q}_{3},{\tt p}_{1},{\tt p}_{2},{\tt p}_{3}$
constitute the base locus of $\{E_{s}\}$:\[\includegraphics[scale=0.6]{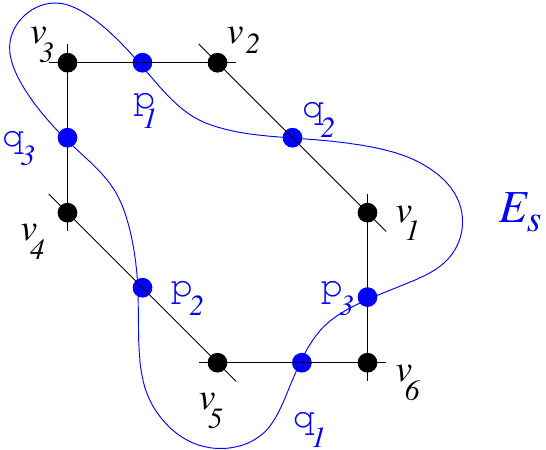}\]So
locally over any small disk $U\subset\PP^{1}$ avoiding the discriminant
locus of $\varepsilon$, writing $\cP^{ij}\to U$ for the $3$-chain
with boundary ${\tt q}_{j}\times U-{\tt q}_{i}\times U$ (and fibers
$P^{ij}$), we may construct the $1$-current \begin{equation}\label{MKeqnIII10}\hat{R} := R - \{ \cL_1 \delta_{\cP^{23}} + \cL_2 \delta_{\cP^{12}} + \cL_3 \delta_{\cP^{31}} \} - (2\pi\ay) \delta_{\partial^{-1}(T_x\cap T_y)} ,
\end{equation}which has $d[\hat{R}]=(2\pi\ay)^{-1}dx/x\wedge dy/y$. Notice that
its restriction to fibers $E_{s}$ is closed.%
\footnote{The resulting family of classes in $H^{1}(E_{s},\CC)$ are lifts of
regulator classes in $H^{1}(E_{s},\CC/\ZZ(2))$ for an element of
$CH^{2}(E_{s},2)$ precisely when the ratios $\xi_{i}/\xi_{i-1}$
are roots of unity, but we will not need this.%
}

For $\xi_{i}$ all $1$ (equal masses) and $s\notin[0,\tfrac{1}{9}]$,
we have $T_{x}\cap T_{y}\cap E_{s}=\emptyset$; moving the $\xi_{i}$
in a small neighborhood of $\underline{1}$, the ``bad set'' $[0,\tfrac{1}{9}]$
thickens slightly. Taking $U$ in the complement, we may ignore the
last term of \eqref{MKeqnIII10} for purposes of integrating over
$\vf_{0}^{(i)}$. Recall from \cite{DoranKerr} that if $\vf_{0}^{i}$
are the cycles that (at $s=0$) get pinched to $v_{i}$, \[\includegraphics[scale=0.6]{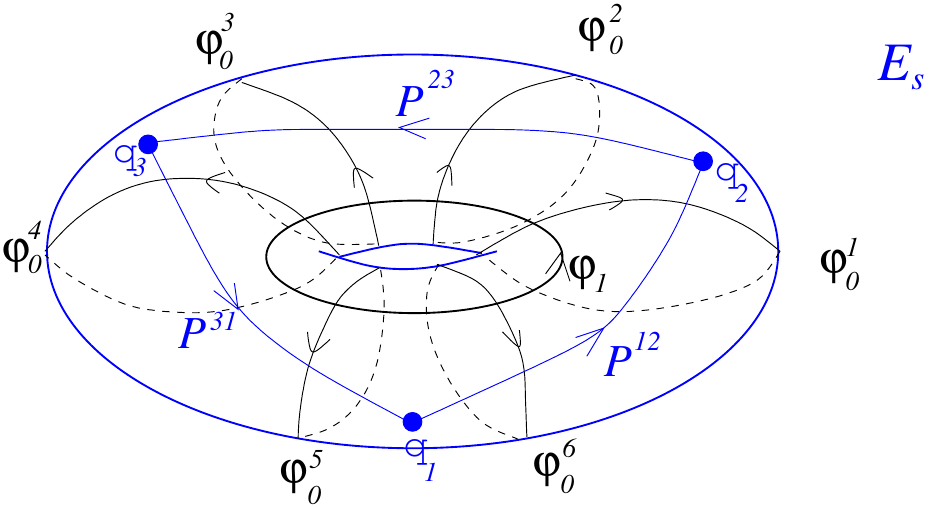}\]then
\[
\vf_{0}^{(i)}=-\vf_{0}^{i}+\vf_{0}^{i+1}+\vf_{0}^{i-1}.
\]
Together with \eqref{MKeqnIII10}, the resulting intersection numbers%
\footnote{Here we are pairing $H^{1}(E^{*})\cong H_{1}(E,\{{\tt p}_{i},{\tt q}_{i}\}_{i=1}^{3})$
and $H_{1}(E^{*})$.%
} $\vf_{0}^{(1)}\cdot P^{23}=\vf_{0}^{(2)}\cdot P^{12}=\vf_{0}^{(3)}\cdot P^{31}=1$
(all others zero) yield \begin{equation}\label{MKeqnIII11}\int_{\vf_0^{(i)}}\hat{R} = R_0^{(i)} - \cL_i =: \hat{R}_0 ,
\end{equation}which according to \eqref{MKeqnIII8} (or the closedness of $\hat{R}|_{E_{s}}$)
is independent of $i$. As suggested by the picture, we can also choose
the $P^{ij}$ to avoid $\vf_{1}$, and so \begin{equation}\label{MKeqnIII11.5}\hat{R}_1 := \int_{\vf_1}\hat{R} = R_1 .
\end{equation}

Next consider the interior product of $d[\hat{R}]$ with a lift of
$s\tfrac{d}{ds}$: working over $U$,\begin{align*} 2\pi\ay\cdot d[\hat{R}]\lrcorner \widetilde{s\tfrac{d}{ds}} & = \tfrac{dx}{x} \wedge \tfrac{dy}{y} \lrcorner \widetilde{s\tfrac{d}{ds}} \\ & = - Res_{\cE}\left( \tfrac{dx}{x} \wedge \tfrac{dy}{y} \wedge \text{dlog} (s^{-1}-\phis) \right) \lrcorner \widetilde{s\tfrac{d}{ds}} \end{align*} restricts
on fibers to\begin{equation*} -Res_{E_s} \left(\tfrac{dx}{x} \wedge\tfrac{dy}{y} \wedge \tfrac{-ds/s^2}{s^{-1}-\phis} \lrcorner \widetilde{s\tfrac{d}{ds}} \right) \\ = Res_{E_s} \left( \frac{\tfrac{dx}{x} \wedge \tfrac{dy}{y}}{1-s\phis} \right) = (2\pi\ay)\omega_s . \end{equation*}It
follows that \begin{equation}\label{MKeqnIII12}\nabla_{\delta_s} [\hat{R}|_{E_s}] = [\omega_s] ,
\end{equation}which along with \eqref{MKeqnIII11} implies that \begin{equation}\label{MKeqnIII13}R_0^{(i)} - \cL_i \, (=\hat{R}_0 ) = \log(-s) + \sum_{m>0} \tfrac{s^m}{m} \beta_m
\end{equation}up to an additive constant. This constant is obviously zero by \eqref{MKeqnIII8}.

Now recall that the Feynman integral is given by $\cIs(s):=-s
V_\su(s)$ \begin{equation}\label{MKeqnIII14}V_\su(s) = \int_{T_x\cap T_y} \frac{\tfrac{dx}{x} \wedge \tfrac{dy}{y}}{1-s\phis} =: \int_{T_x \cap T_y} \hat{\omega}_s .
\end{equation}Writing $\imath^{s}:\, E_{s}\hookrightarrow\PP_{\Delta}$, and $R=\tfrac{1}{2\pi\ay}R\{x,y\}$
as above, we note that $d[\hat{\omega}_{s}]=(2\pi\ay)^{2}\imath_{*}^{s}\omega_{s}$
as a current, and that (on $\PP_{\Delta}$) 
\[
d[\tfrac{1}{2\pi\ay}R]=\tfrac{1}{(2\pi\ay)^{2}}\tfrac{dx}{x}\wedge\tfrac{dy}{y}-\delta_{T_{x}\cap T_{y}}+\{\text{residue terms supported on }\mathbb{D}_{\Delta}\}.
\]
Using integration by parts (for currents), we get that \eqref{MKeqnIII14}
becomes \begin{equation}\label{MKeqnIII14.5}\tfrac{1}{2\pi\ay} \int_{\PP_{\Delta}} R\wedge d[\hat{\omega}_s] = 2\pi\ay\int_{\PP_{\Delta}} R \wedge \imath_*^s \omega_s = 2\pi \ay \int_{E_s} R|_{E_s} \wedge \omega_s .
\end{equation}(Note that \eqref{MKeqnIII14.5} is \emph{not} a truncated higher
normal function in the sense of \cite{DoranKerr}, and neither is
$\langle\hat{R}|_{E_{s}},\omega_{s}\rangle$ in \eqref{MKeqnIII15}
below.) Since $\partial^{-1}(T_{x}\cap T_{y})$ in \eqref{MKeqnIII10}
avoids $\varepsilon^{-1}(U)$ (and $s\in U$), we conclude that%
\footnote{Of course $\langle\hat{R},\omega\rangle$ means $\int_{E_{s}}\hat{R}\wedge\omega$;
we write it this way to emphasize that fact that two cohomology classes
are being paired.%
} \begin{equation}\label{MKeqnIII15}V_\su(s)=\langle \hat{R}|_{E_s} , \tilde{\omega}_s\rangle + \sum_{i=1}^3 \cL \tilde{\pi}^{(i)}_1 (s;\ux ) ,
\end{equation}where $\tilde{\omega}_{s}:=(2\pi\ay)\omega_{s}$, and \begin{equation}\label{MKeqnIII15.5}\tilde{\pi}_1^{(1)} := \int_{\tt{q}_2}^{\tt{q}_3} \tilde{\omega}_s , \tilde{\pi}_1^{(2)} := \int_{\tt{q}_1}^{\tt{q}_2} \tilde{\omega}_s , \tilde{\pi}_1^{(3)} := \int_{\tt{q}_3}^{\tt{q}_1} \tilde{\omega}_s .
\end{equation}Note that the $\{{\tt q}_{i}\}$ depend on $\ux$, and that $\sum_{j=1}^{3}\tilde{\pi}_{1}^{(j)}=\tilde{\pi}_{1}(=2\pi\ay\pi_{1})$.

Let $\theta=\delta_{s}^{2}+q_1(s)\delta_{s}+q_0(s)$ be the Picard-Fuchs
operator associated to $\{\omega_{s}\}$, so that $\nabla_{\delta_{s}}^{2}+f(s)\nabla_{\delta_{s}}+g(s)$
kills $[\omega_{s}]$. Using \eqref{MKeqnIII12} and \eqref{MKeqnIII5.5},
we find $\delta_{s}\langle\hat{R},\omega\rangle=\langle\hat{R},\nabla_{\delta_{s}}\omega\rangle$
and $\delta_{s}^{2}\langle\hat{R},\omega\rangle=(2\pi \ay)^{-1}Y_{\su}(s)+\langle\hat{R},\nabla_{\delta_{s}}^{2}\omega\rangle$,
which leads at once to the inhomogeneous Picard-Fuchs equation:
\begin{prop} \phantomsection\label{MKpropIII1}
We have \begin{equation}\label{MKeqnIII16}
\theta \left(V_\su(s)\right)=\theta\left(-{1\over s} \cIs(s)\right) = Y_{\su}(s) + \sum_{j=1}^3 \log(\xi_j^2) \nu_j(s) ,
\end{equation} where \begin{equation}\label{MKeqnIII17}\nu_i (s) :=\theta \left( \tilde{\pi}^{(i)}_1 (s;\ux ) \right)
\end{equation}satisfy $\sum_{i=1}^{3}\nu_{i}=0$. \end{prop}
\begin{rem}
(i) The functions in \eqref{MKeqnIII17} belong to $\bar{\QQ}(s)^{*}$,
since the partial elliptic integrals in \eqref{MKeqnIII15.5} are
the normal functions associated to globally well-defined algebraic
$0$-cycles $[{\tt q}_{j+1}]-[{\tt q}_{j}]$ on the family $\{E_{s}\}$,
and the section $\{\tilde{\omega}_{s}\}$ of the relative canonical
bundle is defined over $\bar{\QQ}$.

(ii) The right-hand-side of \eqref{MKeqnIII16} only depends on $s$ and the mass
ratios, since this is true for $\nu_{j}$ and $Y_{\su}$; and
we have $\sum_{j=1}^{3}\cL_{j}\nu_{j}=\log(m_{2}^{2}/m_{3}^{2})\nu_{1}+\log(m_{1}^{2}/m_{3}^{2})\nu_{2}$. 
\end{rem}

\begin{rem}The coefficients $q_1(s)$ and $q_0(s)$ are respectively
  given in~\eqref{eq:q1} and~\eqref{eq:q0}. An explicit expression for the $\nu_i(s)$ in some
  coordinate system is given in 
  section~\eqref{e:Y3mass}. In particular $\prod_{i=1}^4 (s\mu_i^2-1)  (s^2
  \prod_{i=1}^4\mu_i-2s(\xi_1^2+\xi_2^2+\xi_3^2)+3) \nu_i(s)= 12s\,c_i(s) $ with $c_1(s)$ given in~\eqref{e:c1} and $c_2(s)$ given\linebreak in~\eqref{e:c2}.
\end{rem}

\begin{rem}
One can also relate \eqref{MKeqnIII14.5} directly
to the elliptic dilogarithm. Noting that up to coboundary $-(2\pi\ay)R|_{E_{s}}\equiv\log(y)\tfrac{dx}{x}-(2\pi\ay)\log(x)\delta_{T_{y}},$
we get \begin{equation}\label{MKeqnIII17.5}\cIs(s)=-s\int_{T_{y}(\cap E_{s})}\log(x)\tilde{\omega}_{s}.
\end{equation} Recalling that $\partial T_{y}=(y)$, this connects at
once to the expression for the sunset integral in \eqref{eq:4.30} hence~\eqref{e:IsunsetR}.
\label{MKremIII2}\end{rem}

\subsection{On the local Gromov-Witten numbers}\label{sec:localGW}

Turning to the numbers $N_{\ul}=N_{\ell_{1},\ell_{2},\ell_{3}}$ (cf.
\eqref{MKeqnIIp21a}), note first that symmetries of $\PP_{\Delta}$
immediately imply that for any $\sigma\in S_{3}$,
\[
N_{\ul}=N_{\sigma(\ul)}.
\]
We also know that 
\[
N_{100}+N_{010}+N_{001}=6,
\]
as this is the number of ``anticanonical-degree-one'' rational curves
on $\PP_{\Delta}$ (the six toric boundary components).

The symmetries also force the prepotential $\Phi_{loc}=(2\pi\ay)R_{1}$
to be symmetric in the $\tau_{i}=(2\pi\ay)R_{0}^{(i)}$ ($i=1,2,3$).
Indeed, this is already recorded in \eqref{MKeqnIIp21b}, which combined
with \eqref{MKeqnIII11} and \eqref{MKeqnIII11.5} becomes \begin{equation}\label{MKeqnIII18}(2\pi\ay) \hat{R}_1 = 3\hat{R}_0^2 + 2\left(\sum \cL_i \right)\hat{R}_0 - \sum_{\ell >0} \ell N_{\ell}\hat{Q}^{\ell} ,
\end{equation}where $\hat{Q}=e^{\hat{R}_{0}}$, and \begin{equation}\label{MKeqnIII19}N_{\ell}:=\sum_{|\ul |=\ell} N_{\ul} \ux^{2\ul} .
\end{equation}But since $\nabla_{\delta_{s}}[\hat{R}]=[\omega]$, we have immediately
$\delta_{s}\hat{R}_{1}=\pi_{1}$ and $\delta_{s}\hat{R}_{0}=\pi_{0}$,
so that \begin{equation}\label{MKeqnIII20}(2\pi\ay){\partial^2
    \hat{R}_1\over\partial \hat{R}_0^2}  =
  (2\pi\ay){\partial\over \partial \hat{R}_0} \frac{\pi_1}{\pi_0} = \frac{Y_{\su}}{\pi_0^3} .
\end{equation}Putting together the expressions of the Yukawa
coupling~\eqref{MKeqnIII6}, the period $\pi_0$ and the coefficients
$\beta_m$ in~\eqref{MKeqnIII7},  for $\hat R_0$ in~\eqref{MKeqnIII13},
the expansion of $\hat R_1$ in~\eqref{MKeqnIII18} and \eqref{MKeqnIII20} now yields the
\begin{prop}\phantomsection\label{MKpropIII2}
In a neighborhood of $s=0$ ($\hat{Q}=0$), we
have
\[
6-\sum_{\ell>0}\ell^{3}N_{\ell}\hat{Q}^{\ell}=\frac{6-4(\xi_{1}^{2}+\xi_{2}^{2}+\xi_{3}^{2})s+2\mu_{1}\mu_{2}\mu_{3}\mu_{4}s^{2}}{\left(1+\sum_{m>0}\beta_{m}s^{m}\right)^3\prod_{i=1}^{4}(1-\mu_{i}^{2}s)},
\]
where $\hat{Q}=-s\exp\left\{ \sum_{m>0}{\beta_{m} s^{m}\over
m}\right\}$.
\end{prop}
We may use Proposition \ref{MKpropIII2} to recover the $N_{\ul}$,
as well as the local ``instanton numbers'' $n_{\ul}$ defined by
the Aspinwall-Morrison formula~\cite{Aspinwall:1991ce,Voisin}
\[
N_{\ell_{1},\ell_{2},\ell_{3}}=\sum_{d|\ell_{1},\ell_{2},\ell_{3}}\frac{1}{d^{3}}n_{\frac{\ell_{1}}{d},\frac{\ell_{2}}{d},\frac{\ell_{3}}{d}}.
\]
As far as we computed, the latter are integers:\[%
\begin{tabular}{|c||c|c|c|c|c|c|c|c|c|}
\hline 
$\ul$ & $(100)$ & $\overset{k>0}{(k00)}$ & $(110)$ & $(210)$ & $(111)$ & $(310)$ & $(220)$ & $(211)$ & $(221)$\tabularnewline
\hline 
\hline 
$N_{\ul}$ & $2$ & $2/k^{3}$ & $-2$ & $0$ & $6$ & $0$ & $-1/4$ & $-4$ & $10$\tabularnewline
\hline 
$n_{\ul}$ & $2$ & $0$ & $-2$ & $0$ & $6$ & $0$ & 0 & $-4$ & $10$\tabularnewline
\hline 
\end{tabular}\]\[\small%
\begin{tabular}{|c||c|c|c|c|c|c|c|c|c|}
\hline 
$\ul$ & $(410)$ & $(320)$ & $(311)$ & $(510)$ & $(420)$ & $(411)$ & $(330)$ & $(321)$ & $(222)$\tabularnewline
\hline 
\hline 
$N_{\ul}$ & $0$ & $0$ & $0$ & $0$ & $0$ & $0$ & $-2/27$ & $-1$ & $-189/4$\tabularnewline
\hline 
$n_{\ul}$ & $0$ & $0$ & $0$ & $0$ & $0$ & $0$ & $0$ & $-1$ & $-48$\tabularnewline
\hline 
\end{tabular}\vspace{1em}\]
Finally, we note that the \GW invariants appear directly in the
Feynman integral, as follows. Write $-s^{-1}\cIs=-s^{-1}\tilde{\cI}_\su+\sum_{i=1}^{3}\cL_{i}\tilde{\pi}_{1}^{(i)},$
and apply $\partial_{\hat{R}_{0}}$ to \eqref{MKeqnIII18} to have
\[
(2\pi\ay)\frac{\pi_{1}}{\pi_{0}}=6\hat{R}_{0}+2\sum_{i=1}^{3}\cL_{i}-\sum_{\ell>0}\ell^{2}N_{\ell}\hat{Q}^{\ell}.
\]
The contribution $\tilde{\cI}_\su$ to the Feynman integral read 
\begin{align}
                                   -s^{-1}  \tilde{\cI}_\su(s) &=
                                                        2\pi\ay\langle
                                                        \hat{R},\omega\rangle
                                                        =
                                                        2\pi\ay(\pi_1\hat{R}_0
                                                        - \pi_0
                                                        \hat{R}_1 )
                                                        \\  &= \pi_0
                                                              \left(
                                                              2\pi\ay\frac{\pi_1}{\pi_0}\hat{R}_0
                                                              -
                                                              2\pi\ay\hat{R}_1
                                                              \right)                \,,                                           
                                                              \nonumber
                                     \end{align}
which using $\pi_1/\pi_0= \delta_s \hat R_1/\delta_s \hat R_0=\partial \hat R_1/\partial R_0$ leads to the
expression as a Legendre transform of $\hat R_1$
\begin{equation}
  \label{eq:Legendre}
\tilde{\cI}_\su(s)= -s\, 2\pi\ay\pi_0\left(
                                                             \frac{\partial\hat
                                                               R_1}{\partial\hat
                                                               R_0}\hat{R}_0
                                                              -
                                                              \hat{R}_1
                                                              \right)       \,.
\end{equation}
This expression  has the expansion
\begin{equation}\label{e:curious}
                             \tilde{\cI}_\su(s)= -s\,\pi_0 \left\{ 3\hat{R}_0^2 + \sum_{\ell>0} \ell(1-\ell\hat{R}_0 )N_{\ell}\hat{Q}^{\ell} \right\} , \end{equation}
The occurrence of the \GW numbers in this Feynman integral seems to be novel.



\subsection{The local Gromov-Witten numbers in the\\ all equal masses case}
\label{sec:grom-witt-invar-1}

In this subsection we compute the local \GW invariants for the all equal
masses case. The family of elliptic curves
$\cEs:=\{xyz-s (x+y+z)(xy+xz+yz)=0|(x,y,z)\in\mathbb P^2\}$ defines a
pencil of elliptic curves in $\mathbb P^2$ corresponding to a modular
family of elliptic curves $f : \cE_\su \to X_1(6) = \{\tau \in \mathbb C|
\Imm(\tau) > 0\}/\Gamma_1(6)$ (see~\cite{Bloch:2013tra}).

\subsubsection{The local \GW numbers}
\hspace{-1ex}In this case Proposition~\ref{MKpropIII2} applied to the case
$\xi_1=\xi_2=\xi_3=1$ implies that
\begin{equation}\label{e:Nk}
  6-\sum_{\ell\geq1} \ell^3 N_\ell \hat Q^\ell=  {6\over (9s-1)(s-1)\, \pi_0^3}\,,
\end{equation}
where the holomorphic period (about $s=0$) of~\eqref{MKeqnIII7} reads
\begin{equation}
\pi_0= \sum_{\ell\geq0} s^\ell \, \sum_{p_1+p_2+p_3=\ell} \left(\ell!\over p_1!p_2!p_3!\right)^2\,.
\end{equation}
and $\hat Q= \exp(\hat R_0)$ where $\hat  R_0$ in~\eqref{MKeqnIII11}
satisfies $sdR_0/ds=\pi_0$  and reads
\begin{equation}
  R_0=\ay\pi+\log s+  \sum_{\ell>0} {s^\ell\over\ell}  \sum_{p_1+p_2+p_3=\ell} \left(\ell!\over p_1!p_2!p_3!\right)^2\,.
\end{equation}
Taking for $s$ the  Hauptmodul used in~\cite{Bloch:2013tra}
\begin{equation}
  s_\su(q)^{-1}= 9+72\, {\eta(q^2)\over\eta(q^3)}\,\left(\eta(q^6)\over\eta(q)\right)^5  
\end{equation}
 we have 
\begin{equation}
  \pi_0(q) = \frac14\, {\theta_2^3(q)\over   \theta_2(q^3)}
\end{equation}
and
\begin{equation}
 \hat R_0(q)= \ay\pi+ \log q   -\sum_{n\geq1} (-1)^{n-1} \left(-3\over n\right)\,n\,\Li_1(q^n)\,,
\end{equation}
where $\left(-3\over n\right)=0,1,-1$ for $n\equiv 0,1,2\mod 3$.

From~\eqref{e:Nk} we compute the local \GW numbers  $N_\ell$ 
\begin{align}
	N_\ell/6 &=1,-\frac{7}{8},\frac{28}{27},-\frac{135}{64},\frac{626}{125},-\frac{751}
   {54},\frac{14407}{343},-\frac{69767}{512},\frac{339013}{729},-\frac{827191}{
   500},\\
   \notag &\quad\ \frac{8096474}{1331}, -\frac{367837}{16},\frac{195328680}{2197},-\frac{1
   37447647}{392},\frac{4746482528}{3375},\\
	\notag	&\quad -\frac{23447146631}{4096},\frac{11596
   2310342}{4913},-\frac{574107546859}{5832},\frac{2844914597656}{6859},\\
\notag &\quad -\frac{
   1410921149451}{800},\frac{10003681368433}{1323},\dots
\end{align}
or introducing $n_\ell$  the virtual number of degree $\ell$
rational curves using the Aspinwall-Morrison multiple cover formula~\cite{Aspinwall:1991ce,Voisin}

\begin{equation}
   N_\ell = \sum_{d|\ell} {1\over d^3} n_{\ell\over d}
\end{equation}
we have
\begin{align}
	n_k/6 &=1,-1,1,-2,5,-14,42,-136,465,-1655,6083,-22988,\\
	\notag &\quad\ 88907,-350637,1406365, -5724384,
   23603157,-98440995,\\
	\notag &\quad\ 414771045,-1763651230,7561361577,-32661478080,\\
	\notag &\quad\ 14204649044
   1,-621629198960,2736004885450,\\
	\notag &\quad\ -12105740577346,53824690388016,\dots
\end{align}

\subsubsection{Comparing the two expansions}
\label{sec:comparing-expansions}

We will show how to relate the $q$ and $Q$  expansions using  a
$\Gamma_1(6)$ modular transformation.
In the all equal masses case the sunset integral was  given by~\cite{Bloch:2013tra}
\begin{equation}
  \cIs(s) \equiv  {\varpi_r\over \pi} E_\su(q) \mod \textrm{period}  \,,
\end{equation}
with $E_\su(q)$ given in~\eqref{e:Esunset1mass}.
The expression is modulo periods of the elliptic curve, and $\varpi_r$
is the real period on the real axis $s>(\xi_1+\xi_2+\xi_3)^{-2}$ given
in~\eqref{e:wrq}.

The all equal masses case the sunset integral is equal to
$\tilde{\mathcal I}_\su$ in~\eqref{e:curious}
\begin{equation}
  \cIs(s) \equiv -s\left(\pi_0 \hat R_1- \pi_1 \hat R_0\right)\mod \textrm{period}  
\end{equation}
where $\pi_0$ is the holomorphic period around $s=0$  and $\pi_1$
is the other non-holomorphic period in~\eqref{MKeqnIII15.5}, and $\hat
R_1$
is such that $sd\hat R_1/ds=\pi_1$
of~\eqref{MKeqnIII18}.
The modular transformations $\tau\to-1/(6\tau)$ maps  the periods as
\begin{equation}
\begin{aligned}
	\varpi_r(-1/(6\tau))&=-6s_\su(\tau) (2\ay\pi\tau)\, \pi_0(\tau);\\
	\pi_1(-1/6\tau)&={3\tau-1\over 6}\,   s_\su(\tau)^{-1} \, \varpi_r(\tau)\,.
\end{aligned}
\end{equation}
The same modular transformation applied to the sunset integral leads
to the relation between the elliptic dilogarithm $E_\su(q)$ and the
regulator period
\begin{equation}\label{e:LegendreE}
36 \ay\tau\, E_\su (-1/(6\tau))=\pi^2+3\ay\pi\log(-q)+3\left( \hat R_1(\tau) - {\partial
  \hat R_1\over \partial \hat R_0} \, \hat R_0\right)\,.
\end{equation}
This shows that  $E_\su(q)$ is the Legendre transform of $\hat R_1(q)$ as
expected from the general different masses case in~\eqref{eq:Legendre}.
Using the $q$-expansion given above and using that $\partial
\hat R_1/\partial \hat R_0=\log(-q)$ we have
\begin{align}
	\hat   R_1(q)- {\partial \hat R_1\over\partial \hat R_0}\, \hat R_0 &= -\frac12\,
  \log(-q)^2\\
	\notag &\quad +\sum_{n\geq1} \, \left( \sum_{d|n} (-1)^d d^2 \left(-3\over d\right)\right) \Li_2(q^n)  \,.
\end{align}


%
\part{Appendices}
\appendix
\section{Theta functions}
 \label{sec:theta-functions}

In this appendix we recall standard results on Jacobi theta functions that
are used in the text. We use the notation $q=e^{2\pi\ay \tau}$ with
$\tau$ the period ratio chosen to lie in the upper-half-plane, and $x\in\IC^\times/q^\ZZ$
\begin{equation}\label{e:ThetaI}
\theta_1(x)
 :={q^{1\over 8}\over \ay} \left(x^{1\over2}-x^{-{1\over2}}\right)\,\prod_{n=1}^{\infty} (1-q^n)(1-q^nx)(1-q^n/x)\,,
\end{equation}
and
\begin{equation}\label{e:ThetaII}
\theta_2(x) 
:=q^{1\over 8}\left(x^{1\over2}+x^{-{1\over2}}\right)\,\prod_{n=1}^{\infty} (1-q^n)(1+q^nx)(1+q^n/x)\,,
\end{equation}
and
\begin{equation}\label{e:ThetaIII}
\theta_3(x)
:=\prod_{n=1}^{\infty} (1-q^n)(1+q^{n-1/2}x)(1+q^{n-1/2}/x)\,,
\end{equation}
and finally
\begin{equation}\label{e:ThetaIV}
\theta_4(x)
:=\prod_{n=1}^{\infty} (1-q^n)(1-q^{n-1/2}x)(1-q^{n-1/2}/x)\,.
\end{equation}
We will use the shorthand notation $\theta_a:=\theta_a(1)$
for $a=2,3,4$, or $\theta_\alpha(q)$ when needed.
A particular case of the  Jacobi identity is 
\begin{equation}\label{e:Jacobi}
\theta_3^2(v) \theta_3^2(u) +\theta_1^2(v) \theta_1^2(u)=\theta_2^2(v) \theta_2^2(u) +\theta_4^2(v) \theta_4^2(u)\,.
\end{equation}
 Applying this identity for $v=\exp(\ay\pi (a+b\tau))$ with $a,b\in\{0,1\}$ one
obtains the following quadratic relations satisfied by the theta functions 
\begin{equation}
  \label{eq:thetaRelations}
  \begin{pmatrix}
0& \theta_2^2& -\theta_3^2&\theta_4^2\cr
\theta_2^2&0&\theta_4^2& -\theta_3^2\cr
\theta_3^2&\theta_4^2&0&-\theta_2^2\cr
\theta_4^2&\theta_3^2&-\theta_2^2&0
  \end{pmatrix}
  \begin{pmatrix}
    \theta_1^2(u)\cr \theta_2^2(u)\cr \theta_3^2(u) \cr \theta_4^2(u) 
  \end{pmatrix}
=
\begin{pmatrix}
  0\cr0 \cr 0\cr 0
\end{pmatrix}\,.
\end{equation}

\section{The coefficients of the Picard-Fuchs equation}\label{sec:coeff}

In this appendix we give the explicit expressions for the coefficients
of the homogeneous polynomials used when deriving the sunset Picard-Fuchs equation.

\subsection{The coefficients $C_x$, $C_y$ and $C_z$}

The coefficients $C_x$, $C_y$ and $C_z$ are homogeneous polynomials of
degree 4 in $(x,y,z)$ of the form
\begin{align}
  C_x&= x y^2z C_x^{1,2,1}+ x^2 z^2 C_x^{2,0,2} +x^2 yz
         C_x^{2,1,1}+x^3z C_x^{3,0,1}\,,\nn\\[1.5ex]
  C_y&= x y z^2 C_y^{1,1,2}+ x y^2 z C_y^{1,2,1} +x^2 z^2
         C_y^{2,0,2}+x^2yz C_y^{2,1,1}\,,\\[1.5ex]
\nn  C_z&= x z^3 C_z^{1,0,3} + x y z^2 C_z^{1,1,2} +x y^2z
         C_z^{1,2,1}+x^2z^2 C_z^{2,0,2}+x^2yz C_z^{2,1,1}\,.
\end{align}

Their detailed expressions are given by  for $C_x$
\begin{align*}
&\quad\ 6 \prod_{i=1}^4(s\mu_i^2-1)  C_x \\
&= s x z \left(m_1^2 x (9 x+20 y)+3 m_2^2 y (6 x+y)+2 m_3^2 x (10 y-3
   z)\right)\\
&\quad +s^4 x z \left(m_1^4-2 m_1^2
   \left(m_2^2+m_3^2\right)+\left(m_2^2-m_3^2\right)^2\right)\\
&\qquad\times
   \left(m_1^4 x (x+y)+m_1^2 \left(m_2^2 \left(5 x^2+8 x y+3
    y^2\right)-m_3^2 x (5 x+2 (y+z))\right)\right.\\
&\qquad\quad +\left.\left(m_2^2-m_3^2\right)
   \left(3 m_2^2 y (x+y)-m_3^2 x (y-2 z)\right)\right)\\
&\quad-s^2 x z
   \left(m_1^4 x (17 x+18 y)\right.\\
&\qquad\qquad +m_1^2 \left(m_2^2 \left(13 x^2+46 x y+3
   y^2\right) +3 m_3^2 x (-3 x+4 y+2 z)\right)\\
&\qquad\qquad +\left.3 m_2^4 y (4 x+y)+m_2^2
   m_3^2 \left(10 x y-14 x z+9 y^2\right)+2 m_3^4 x (9 y-5
    z)\right)\\
&\quad +s^3 x z
   \left(m_1^6 x (7 x+4 y)\right.\\
	&\qquad\qquad +m_1^4 \left(m_2^2 \left(18 x^2+22 x y-3
   y^2\right)-2 m_3^2 x (5 x+2 y-7 z)\right)\\
&\qquad\qquad -m_1^2 \left(m_2^4
   \left(x^2-24 x y-30 y^2\right)\right.\\
   &\qquad\qquad\qquad\ +2 m_2^2 m_3^2 \left(-7 x^2-22 x y+2 x z+3
   y^2\right)\\
&\qquad\qquad\qquad\  +\left.m_3^4 x (13 x+4 y+28 z)\right)\\
&\qquad\qquad -\left(m_2^2-m_3^2\right)
   \left(m_2^4 y (2 x+3 y)+m_2^2 m_3^2 \left(2 x (y+5 z)+9
    y^2\right)\right.\\
	&\hspace{10em}+\left.\left.2   m_3^4 x (2 y-z)\right)\right) -7 x^2 y z\,,
\end{align*}
for $C_y$
\begin{align*}
&\quad\ 3 \prod_{i=1}^4(s\mu_i^2-1) C_y\\
 &=-2 s x y z \left(m_1^2 (3 x+2
   y)+3 m_2^2 y+m_3^2 (2 y+3 z)\right)\\
&\quad - 2 s^4
   x z \left(m_1^4-2 m_1^2
   \left(m_2^2+m_3^2\right)+\left(m_2^2-m_3^2\right)^2\right)\\
&\qquad\times
   \left(m_1^4 y (x+y)-m_1^2 \left(m_2^2 y (x+y)+m_3^2 \left(5 x
   y+6 x z+2 y^2+5 y z\right)\right)\right.\\
&\qquad\quad +\left.m_3^2 y \left(m_3^2-m_2^2\right)
   (y+z)\right)\\
&\quad +2 s^2 x z \left(5 m_1^4 x y+m_1^2 \left(m_2^2 y (7
   x+y)-3 m_3^2 (7 x y+6 x z+7 y z)\right)\right.\\
	&\hspace{5em}+\left.y \left(3 m_2^4 y+m_2^2
   m_3^2 (y+7 z)+5 m_3^4 z\right)\right)\\
&\quad -2 s^3 x z \left(m_1^6 y (x-2
   y)+m_1^4 \left(m_2^2 y (y-6 x)+m_3^2 (2 x (y-6 z)+y (2 y-19
   z))\right)\right.\\
	&\hspace{5em}+ m_1^2 \left(5 m_2^4 x y-2 m_2^2 m_3^2 (x (5 y+6
   z)+5 y (y+z))\right.\\
	&\hspace{7.5em}\left.+m_3^4 (2 y (y+z)-x (19 y+12 z))\right)\\
	&\hspace{5em} +\left.y \left(m_2^6 y+5
   m_2^4 m_3^2 z+m_2^2 m_3^4 (y-6 z)+m_3^6 (z-2
   y)\right)\right)+2 x y^2 z\,,
\end{align*}
and for $C_z$
\begin{align*}
 &\quad\ 6 \prod_{i=1}^4(s\mu_i^2-1) C_z\\
&=s x z \left(-2 m_1^2 (6 x y+9 x
                                        z-y z)-3 m_2^2 y (y-2 z)+m_3^2
                                        z (2 y+3 
   z)\right)\\
&\quad +s^4 x z \left(-\left(m_1^4-2 m_1^2
   \left(m_2^2+m_3^2\right)+\left(m_2^2-m_3^2\right)^2\right)\right)\\
&\qquad \times \left(m_1^4 z (2 x-y)\right.\\
&\qquad\quad +m_1^2 \left(m_2^2 \left(12 x y+10 x z+3
   y^2+10 y z\right)+m_3^2 z (2 x+2 y+5
   z)\right)\\
&\qquad\quad +\left.\left(m_2^2-m_3^2\right) (y+z) \left(3 m_2^2 y+m_3^2
   z\right)\right)\\
&\quad +s^2 x z \left(m_1^4 x (24 y+34 z)\right.\\
&\qquad\qquad +m_1^2 \left(m_2^2
   \left(12 x y+26 x z+3 y^2+32 y z\right)+3 m_3^2 \left(8 x y+6 x z+7
   z^2\right)\right)\\
&\qquad\qquad +\left.3 m_2^4 y (y-4 z)+m_2^2 m_3^2 \left(9 y^2-4 y z-7
   z^2\right)-5 m_3^4 z^2\right)\\
&\quad-s^3 x z \left(2 m_1^6 (6 x y+7 x z+y
   z)\right.\\
	&\qquad\qquad +m_1^4 \left(m_2^2 \left(48 x y+36 x z-3 y^2+26 y
    z\right)\right.\\
&\qquad\qquad\qquad\quad +\left.m_3^2
   (4 x (z-6 y)+z (19 z-2 y))\right)\\
&\qquad\qquad -2 m_1^2 \left(m_2^4 (x (6 y+z)-15 y
   (y+z))\right.\\
&\qquad\qquad\qquad\quad - m_2^2 m_3^2 \left(24 x y+26 x z-3 y^2+2 y z+5
   z^2\right)\\
&\qquad\qquad\qquad\quad \left. +m_3^4 (x (z-6 y)+z (y+z))\right)\\
&\qquad\qquad -\left(m_2^2-m_3^2\right)
   \left(m_2^4 y (3 y+10 z)+m_2^2 m_3^2 \left(9 y^2+4 y z+5
   z^2\right)\right.\\
	&\hspace{10em} \left.\left. +m_3^4 z (2 y-z)\right)\right)-x y z^2\,.
\end{align*}

\subsection{The coefficients $\tilde C_x$, $\tilde C_y$ and $\tilde
  C_z$}

 \enlargethispage{1em}
The coefficients $\tilde C_x$, $\tilde C_y$ and $\tilde
  C_z$ are homogeneous polynomials of degree one in $(x,y,z)$ with the
  detailed expressions given below.

Setting $N=3(s^2\prod_{i=1}^4
\mu_i-2sM_2+3)\prod_{i=1}^4(\mu_i^2s-1) $ we have for $\tilde C_x$
\begin{align*}
2N\tilde
	C_x&=
-s x \left(55 m_1^2+43 m_2^2+49 m_3^2\right)\\
&\quad +2 s^2 \left(21
   m_1^4 x+m_1^2 \left(m_2^2 (52 x+6 y)+6 m_3^2 (6
   x+z)\right)+m_2^4 x\right.\\
&\qquad\qquad +\left.2 m_2^2 m_3^2 (35 x-3 (y+z))+9 m_3^4
   x\right)\\
&\quad +s^5 \left(-\left(m_1^4-2 m_1^2
   \left(m_2^2+m_3^2\right)+\left(m_2^2-m_3^2\right)^2\right)\right)\\
&\qquad \times\left(3 m_1^6 x+m_1^4 \left(m_2^2 (5 x+12 y)+3 m_3^2
   (x+4 z)\right)\right.\\
&\qquad\qquad -\left.m_1^2 \left(m_2^4 (7 x+12 y)+22 m_2^2 m_3^2
   x+3 m_3^4 (x+4 z)\right)\right.\\
&\qquad\qquad -\left.\left(m_2^2-m_3^2\right) \left(m_2^4
   x+2 m_2^2 m_3^2 (x-6 y+6 z)-3 m_3^4 x\right)\right)
\end{align*}
\begin{align*}
& -2 s^3 \left(3
   m_1^6 x+m_1^4 \left(m_2^2 (23 x+18 y)+m_3^2 (13 x+18
   z)\right)\right.\\
&\qquad\quad +m_1^2 \left(7 m_2^4 (5 x-6 y)+70 m_2^2 m_3^2
   x+m_3^4 (23 x-42 z)\right)-21 m_2^6 x\cr
&\qquad\quad +\left.3 m_2^4 m_3^2 (17 x+14
   y-6 z)+m_2^2 m_3^4 (49 x-18 y+42 z)-15 m_3^6 x\right)\\
&+s^4
   \left(m_1^8 x-4 m_1^6 \left(m_2^2 (4 x-9 y)-9 m_3^2
   z\right)+2 m_1^4 \left(m_2^4 (11 x-28 y)\right.\right.\\
&\qquad+\left.2 m_2^2 m_3^2 (5
   x-19 (y+z))-m_3^4 (5 x+28 z)\right)+4 m_1^2 \left(m_2^6 (4 x+5
   y)\right.\\
&\qquad + \left.m_2^4 m_3^2 (32 x+19 z)+m_2^2 m_3^4 (20 x+19
   y)+m_3^6 (8 x+5 z)\right)\\
&\qquad -\left(m_2^2-m_3^2\right) \left(23
   m_2^6 x+m_2^4 m_3^2 (11 x+20 y+36 z)\right.\\
&\qquad -\left.\left.m_2^2 m_3^4 (11
   x+36 y+20 z)-23 m_3^6 x\right)\right)+21 x,
\end{align*}
for $\tilde C_y$
\begin{align*}
	N\tilde C_y&=
7 s y \left(m_1^2+m_2^2+m_3^2\right)\\
&\quad +s^2 \left(-6 m_1^4 y+2
   m_1^2 \left(m_2^2 (3 x-7 y)-3 m_3^2 (x+2 y+z)\right)\right.\\
&\qquad\quad\ -2\left.
   \left(m_2^4 y+m_2^2 m_3^2 (7 y-3 z)+3 m_3^4
   y\right)\right)\\
&\quad +s^5 \left(m_1^4-2 m_1^2
   \left(m_2^2+m_3^2\right)+\left(m_2^2-m_3^2\right)^2\right)\\
&\qquad\times
   \left(3 m_1^6 y+m_1^4 \left(m_2^2 (6 x-y)-3 m_3^2 (2 x+y-2
   z)\right)\right.\\
&\qquad\quad\ -m_1^2 \left(m_2^4 (6 x+y)-2 m_2^2 m_3^2 y+3
   m_3^4 (-2 x+y+2 z)\right)\\
&\qquad\quad\ -\left.\left(m_2^2-m_3^2\right)
   \left(m_2^4 y+2 m_2^2 m_3^2 (y+3 z)+3 m_3^4
   y\right)\right)\\
&\quad +s^3 \left(6 m_1^6 y+6 m_1^4 \left(7 m_2^2
   x-m_3^2 (7 x+y-3 z)\right)\right.\\
&\qquad\quad -2 m_1^2 \left(m_2^4 (9 x-4 y)-28
   m_2^2 m_3^2 y+3 m_3^4 (-3 x+y+7 z)\right)\\
&\qquad\quad -\left.6 m_2^6 y+2
   m_2^4 m_3^2 (4 y-9 z)+42 m_2^2 m_3^4 z+6 m_3^6
   y\right)\\
&\quad +s^4 \left(-7 m_1^8 y+2 m_1^6 \left(m_2^2 (5 x+7
   y)+m_3^2 (-5 x+10 y-9 z)\right)\right.\\
&\qquad\quad-2 m_1^4 \left(m_2^4 (14 x+5
   y)+m_2^2 m_3^2 (7 y-19 z)\right.\\
&\qquad\qquad\qquad \left.+m_3^4 (-14 x+13 y-14 z)\right)\\
&\qquad\quad +2
   m_1^2 \left(m_2^6 (9 x-y)-m_2^4 m_3^2 (19 x+10 y+19
   z)\right.\\
&\qquad\qquad\qquad \left.+m_2^2 m_3^4 (19 x-7 y)+ m_3^6 (-9 x+10 y-5
   z)\right)\\
&\qquad\quad +\left(m_2^2-m_3^2\right) \left(5 m_2^6 y+3 m_2^4
   m_3^2 (y+6 z)-m_2^2 m_3^4 (7 y+10 z)\right.\\
	&\hspace{9em}\left.\left.+7 m_3^6
    y\right)\right)-3 y ,
\end{align*}
and for $\tilde C_z$
\begin{align*}
2N\tilde C_z&=
-s z \left(m_1^2+13 m_2^2+7 m_3^2\right)\\
&\quad -2 s^2 \left(9 m_1^4 z+2
   m_1^2 \left(m_2^2 (3 x+3 y-5 z)-3 m_3^2 (x+4 z)\right)\right.\\
&\qquad\qquad  \left.-7
   m_2^4 z+2 m_2^2 m_3^2 (4 z-3 y)-3 m_3^4 z\right)
\end{align*}
\begin{align*}
	&+s^5
   \left(m_1^4-2 m_1^2
   \left(m_2^2+m_3^2\right)+\left(m_2^2-m_3^2\right)^2\right)\\
&\quad \times
   \left(3 m_1^6 z+m_1^4 \left(m_2^2 (-12 x+12 y+z)+3 m_3^2 (4
   x+z)\right)\right.\\
	&\qquad  +m_1^2 \left(m_2^4 (12 x-12 y-11 z)+10 m_2^2 m_3^2
   z-3 m_3^4 (4 x+z)\right)\\
&\qquad \left. +\left(m_2^2-m_3^2\right) \left(7
   m_2^4 z+2 m_2^2 m_3^2 (6 y+z)+3 m_3^4 z\right)\right)\\
&+2 s^3
   \left(15 m_1^6 z+m_1^4 \left(m_2^2 (-42 x+18 y-19 z)+m_3^2
   (42 x-23 z)\right)\right.\\
&\qquad\quad +m_1^2 \left(3 m_2^4 (6 x-14 y-5 z)+14 m_2^2
   m_3^2 z-m_3^4 (18 x+13 z)\right)\\
&\qquad\quad \left. +3 m_2^6 z+m_2^4 m_3^2
   (42 y+z) +m_2^2 m_3^4 (7 z-18 y)-3 m_3^6 z\right)\\
&+s^4 \left(-17
   m_1^8 z-4 m_1^6 \left(m_2^2 (5 x+9 y-9 z)-m_3^2 (5 x+2
   z)\right)\right.\\
&\qquad \left.+m_1^4 \left(m_2^4 (56 x+56 y-38 z)+4 m_2^2 m_3^2
   (19 y+14 z)+2 m_3^4 (13 z-28 x)\right)\right.\\
&\qquad -4 m_1^2 \left(m_2^6 (9 x+5
   y-9 z)+m_2^4 m_3^2 (10 z-19 x)\right.\\
&\qquad\qquad\quad \left. +m_2^2 m_3^4 (19 x+19 y+z) +3 m_3^6 (2 z-3 x)\right)\\
	&\qquad -\left(m_2^2-m_3^2\right) \left(17
   m_2^6 z-m_2^4 m_3^2 (20 y+23 z)\right.\\
	&\hspace{8em} \left.\left. +3 m_2^2 m_3^4 (12 y+5
   z)+ 7 m_3^6 z\right)\right)+3 z  \,.
\end{align*}



\end{document}